\documentclass[12pt]{article} 

\usepackage{amsmath,amssymb,amsfonts}
\usepackage{psfrag}
\usepackage{color}
\definecolor{darkblue}{rgb}{0.1,0.1,.7}
\usepackage[colorlinks, linkcolor=darkblue, citecolor=darkblue, urlcolor=darkblue, linktocpage]{hyperref} 
\usepackage[square, comma, compress,numbers]{natbib}
\usepackage[]{amsmath}
\usepackage[]{graphicx}
\usepackage[]{latexsym}
\usepackage{geometry}
\usepackage{amscd}
\usepackage[all,cmtip]{xy}
\usepackage{mathrsfs}
\usepackage[margin=10pt,font=small,labelfont=bf]{caption}
\usepackage{dsdshorthand}
\usepackage{simplewick}
\usepackage{changepage}
\usepackage{tikz}
\usetikzlibrary{positioning,arrows}
\usetikzlibrary{decorations.pathmorphing}
\usetikzlibrary{decorations.markings}
\usetikzlibrary{patterns,snakes}
\numberwithin{equation}{section}

\def\be#1\ee{\begin{align}#1\end{align}}
\newcommand\nn\nonumber 

\theoremstyle{definition}
\newcommand\cS{\mathcal{S}}
\newtheorem{mylemma}{Lemma}

\begin{document}
\pagenumbering{gobble}

\vspace*{-.6in} \thispagestyle{empty}
\begin{flushright}
\end{flushright}
\vspace{.2in} {\Large
\begin{center}
{\bf The Lightcone Bootstrap and the \\ Spectrum of the 3d Ising CFT \\\vspace{.1in}}
\end{center}
}
\vspace{.2in}
\begin{center}
{\bf 
David Simmons-Duffin$^{1,2}$} 
\\
\vspace{.2in} 
$^1${\it School of Natural Sciences, Institute for Advanced Study, Princeton, New Jersey 08540\/} \\
$^2${\it Walter Burke Institute for Theoretical Physics, Caltech, Pasadena, California 91125\/}
\end{center}

\vspace{.2in}

\begin{abstract}
We compute numerically the dimensions and OPE coefficients of several operators in the 3d Ising CFT, and then try to reverse-engineer  the solution to crossing symmetry analytically. Our key tool is a set of new techniques for computing infinite sums of $\SL(2,\R)$ conformal blocks. Using these techniques, we solve the lightcone bootstrap to all orders in an asymptotic expansion in large spin, and suggest a strategy for going beyond the large spin limit. We carry out the first steps of this strategy for the 3d Ising CFT, deriving analytic approximations for the dimensions and OPE coefficients of several infinite families of operators in terms of the initial data $\{\De_\s,\De_\e,f_{\s\s\e},f_{\e\e\e},c_T\}$. The analytic results agree with numerics to high precision for about 100 low-twist operators (correctly accounting for $O(1)$ mixing effects between large-spin families). Plugging these results back into the crossing equations, we obtain approximate analytic constraints on the initial data.
\end{abstract}

\newpage

\renewcommand{\baselinestretch}{0.75}\normalsize
\tableofcontents
\renewcommand{\baselinestretch}{1.0}\normalsize

\newpage
\pagenumbering{arabic}

\section{Introduction}

Despite the ubiquity of conformal field theories (CFTs) in $d>2$ spacetime dimensions, very little is known about their operator dimensions and OPE coefficients away from simplifying limits like large central charge (large-$N$) or weak coupling. Unlike in $d=2$ dimensions~\cite{Belavin:1984vu}, we have no nontrivial exactly-solvable CFTs in $d>2$ from which to draw lessons.

In this work, we produce a new numerical picture of the spectrum of the 3d Ising CFT, including about 100 operators, and use it as a guide to explore the theory analytically. In addition to the intrinsic interest of the 3d Ising CFT for its role in second-order phase transitions, our motivation is to develop analytical tools for solving crossing symmetry in general (and eventually apply them to wider classes of theories).

The current most powerful techniques for studying the spectrum of small central charge theories are numerical bootstrap techniques~\cite{Rattazzi:2008pe,Rychkov:2009ij,Caracciolo:2009bx,Rattazzi:2010gj,Poland:2010wg,Rattazzi:2010yc,Vichi:2011ux,Poland:2011ey,Rychkov:2011et,ElShowk:2012ht,Liendo:2012hy,ElShowk:2012hu,Gliozzi:2013ysa,Kos:2013tga,Alday:2013opa,Gaiotto:2013nva,Berkooz:2014yda,El-Showk:2014dwa,Nakayama:2014lva,Nakayama:2014yia,Chester:2014fya,Kos:2014bka,Caracciolo:2014cxa,Nakayama:2014sba,Paulos:2014vya,Bae:2014hia,Beem:2014zpa,Chester:2014gqa,Simmons-Duffin:2015qma,Bobev:2015jxa,Kos:2015mba,Chester:2015qca,Beem:2015aoa,Iliesiu:2015qra,Rejon-Barrera:2015bpa,Poland:2015mta,Lemos:2015awa,Kim:2015oca,Lin:2015wcg,Chester:2015lej,Chester:2016wrc,Behan:2016dtz,Dey:2016zbg,Nakayama:2016knq,El-Showk:2016mxr,Li:2016wdp,Pang:2016xno,Lin:2016gcl,Lemos:2016xke,Beem:2016wfs}, based on the conformal bootstrap~\cite{Ferrara:1973yt,Polyakov:1974gs} and the methods pioneered in~\cite{Rattazzi:2008pe}. For example, the numerical bootstrap has yielded precise predictions for dimensions of the lowest-dimension scalars $\s$ and $\e$ in the 3d Ising CFT~\cite{ElShowk:2012ht,El-Showk:2014dwa,Kos:2014bka,Simmons-Duffin:2015qma,Kos:2016ysd}. It is difficult to reproduce these results analytically because the 3d Ising CFT does not admit a (known) controlled expansion in a small coupling constant.\footnote{See~\cite{Gopakumar:2016wkt,Gopakumar:2016cpb,Dey:2016mcs} for recent attempts using Mellin space.}

But even strongly-coupled theories admit small parameters in kinematic limits. The authors of~\cite{Fitzpatrick:2012yx,Komargodski:2012ek} showed that every CFT admits a large-spin expansion, accessible via the lightcone limit of the crossing equations. By studying the lightcone limit, one can prove:
\begin{theorem}[Existence of double-twist operators~\cite{Fitzpatrick:2012yx,Komargodski:2012ek}]
\label{thm:existenceofdoubletwist} Suppose a CFT in $d>2$ dimensions contains primary operators $\cO_1,\cO_2$ with twists $\tau_1,\tau_2$.\footnote{Twist is defined as $\tau=\De-\ell$.} For each $n=0,1,2,\dots$, there exists an infinite family of primary operators with increasing spin and twists approaching $\tau_1+\tau_2+2n$ as $\ell\to \oo$.
\end{theorem}
\noindent Schematically, these operators are
\be
\cO_1 \ptl^{\mu_1} \cdots \ptl^{\mu_\ell} \ptl^{2n} \cO_2.
\label{eq:schematiclargespin}
\ee
Of course, composite operators like (\ref{eq:schematiclargespin}) don't make sense in a general strongly-coupled theory. However, theorem~\ref{thm:existenceofdoubletwist} implies that they do make sense in the large-$\ell$ limit. We denote the family with twist approaching $\tau_1+\tau_2+2n$ as $[\cO_1\cO_2]_n$ and refer to such operators as ``double-twist" operators (following~\cite{Komargodski:2012ek}).

Dimensions and OPE coefficients of double-twist operators have a computable expansion in (generically non-integer) powers of $1/\ell$, where terms in the expansion come from matching operators on the other side of the crossing equation. Recently, there has been significant progress in understanding this expansion~\cite{Fitzpatrick:2012yx,Komargodski:2012ek,Alday:2015eya,Alday:2015ota,Alday:2015ewa,Kaviraj:2015cxa,Kaviraj:2015xsa,Alday:2016mxe,Alday:2016njk,Alday:2016jfr}. The large-$\ell$ expansion is asymptotic in general~\cite{Alday:2015ewa}, so its usefulness for studying finite-spin operators is not immediately clear. Nevertheless, we might hope that large-spin techniques could enhance numerics or vice versa. Perhaps an analytical solution of the large-spin expansion could help make numerics more efficient, or even replace numerics entirely if crossing symmetry could be solved via the lightcone limit.

With our concrete numerical calculations as a guide, we find the following:
\begin{itemize}
\item Double-twist operators play an important role in the numerical bootstrap.

\item By truncating the asymptotic large spin expansion, and with the help of some new analytical techniques described below, we can describe a large part of the 3d Ising spectrum, including operators with spin as small as $\ell=2$ or $\ell=4$.

\item The large-spin expansion can be used to solve crossing symmetry systematically in a ``double lightcone" expansion in $z,1-\bar z$.\footnote{While this work was nearing completion,~\cite{Alday:2016njk,Alday:2016jfr} appeared which also develop this approach.}

\item The ``errors" associated with the fact that the expansion is asymptotic can be precisely characterized (they are ``Casimir-regular" terms defined in section~\ref{sec:families}). Requiring that they cancel gives nontrivial constraints on the spectrum.

\end{itemize}

Let us describe the structure of this paper in more detail.

In section~\ref{sec:numericsandlightconelimit}, we perform a (non-rigorous) numerical computation of the 3d Ising spectrum using the extremal functional method~\cite{Poland:2010wg,ElShowk:2012hu,El-Showk:2014dwa}.  Importantly, we use a trick from~\cite{Komargodski:2016auf} which lets us assign error bars to the resulting operator dimensions and thereby understand which predictions are robust and which ones aren't. The robust predictions turn out to be for low-{\it twist\/} operators (not just low-dimension operators). Specifically, we find relatively precise predictions for 112 operators in the 3d Ising CFT, of which only 9 do not fall into an obvious double-twist family. The remaining $103$ operators give a clear numerical picture of the families $[\s\s]_0$, $[\s\s]_1$, $[\e\e]_0$, and $[\s\e]_0$, up to spin $\ell\sim 40$. We give additional details of our computation in appendix~\ref{app:numerics}, and list the resulting operators in appendix~\ref{app:severaloperators}. Although many of the results in this work are analytical (and applicable to any CFT), this numerical picture is a crucial guide, helping us ask the right questions and find the right tools to answer them.

We then set out to describe the families $[\s\s]_0$, $[\s\s]_1$, $[\e\e]_0$, and $[\s\e]_0$ analytically using the large-spin expansion. To succeed, we must develop two new technologies:
\begin{itemize}
\item Techniques for summing an infinite family of large-spin operators in the conformal block expansion. (For example, this lets us compute the contribution of a twist family to its own anomalous dimensions.)
\item Techniques for describing mixing between multi-twist families.
\end{itemize}

Our key tool is a better understanding of infinite sums of $\SL(2,\R)$ conformal blocks, which we develop in section~\ref{sec:families} (after reviewing the lightcone bootstrap in section~\ref{sec:lightconereview}). By generalizing the conformal block expansion of $1$-dimensional Mean Field Theory, we show how to compute exactly, and in great generality, sums of $\SL(2,\R)$ blocks in an expansion in the crossed channel $z\to 1-z$. A simple example is\footnote{The sum over $k$ in (\ref{eq:examplesl2sum}) can be written in terms of ${}_3F_2$ hypergeometric functions.}
\be
&\sum_{\substack{h=h_0+\ell \\ \ell=0,1,\dots}} \frac{1}{\G(-a)^2} \frac{\G(h)^2}{\G(2h-1)} \frac{\G(h-a-1)}{\G(h+a+1)} z^h {}_2F_1(h,h,2h,z)\nn\\
&=
\p{\frac{1-z}{z}}^a - \frac{1}{\Gamma (-a)^2}\frac{\Gamma (h_0-a-1)}{\Gamma (h_0+a) }
\sum_{k=0}^\oo \pdr{}{k}\p{ \frac{\Gamma (h_0+k)}{k!(a-k)\Gamma (h_0-k-1)}\p{\frac{1-z}{z}}^k}.
\label{eq:examplesl2sum}
\ee
The crucial point is that the first term on the right-hand side, $\p{\frac{1-z}{z}}^a$, becomes arbitrarily singular at $z=1$ after repeated application of the quadratic Casimir of $\SL(2,\R)$, while the remaining terms do not. We compute general sums of $\SL(2,\R)$ blocks by exploiting this distinction.
Because $\SO(d,2)$ conformal blocks are sums of $\SL(2,\R)$ blocks, equation~(\ref{eq:examplesl2sum}) and similar identities can be used as building blocks for understanding crossing symmetry in general. Using them, we solve the asymptotic lightcone bootstrap to all orders (for both OPE coefficients and anomalous dimensions) in section~\ref{sec:allorderslightcone}.

In section~\ref{sec:asymptoticsappliedtoising}, we explore how well the truncated large-spin expansion describes the families $[\s\s]_0$ and $[\s\e]_0$. Surprisingly, we find that the first few terms (coming from $\e$ and $T_{\mu\nu}$ in the crossed-channel) fit the numerical data for $[\s\s]_0$ beautifully, even down to spin $\ell=2$!\footnote{This was conjectured in~\cite{Alday:2015ota,Alday:2015ewa}.} To describe $[\s\e]_0$, we must perform a nontrivial sum over the twist family $[\s\s]_0$ in the $\s\s\to\e\e$ OPE channel. The result is another beautiful fit that works down to spin $\ell=2$. In this way, we find analytical approximations for dimensions and OPE coefficients of $[\s\s]_0$ and $[\s\e]_0$ in terms of the data $\{\De_\s,\De_\e,f_{\s\s\e},f_{\e\e\e},c_T\}$.

Describing $[\s\s]_1$ and $[\e\e]_0$ requires a novel approach because the two families exhibit nontrivial mixing. (For example the OPE coefficient $f_{\e\e[\s\s]_1}$ is {\it larger\/} than $f_{\e\e[\e\e]_0}$ for spins $\ell \leq 26$.) In section~\ref{sec:twisthamiltonian}, using our solution of the asymptotic lightcone bootstrap, we show how to define a ``twist Hamiltonian" $H(\bar h=\frac{\De+\ell}{2})$ whose diagonalization correctly describes this mixing, and matches the numerics well for $\ell\geq 4$. In particular, diagonalizing $H(\bar h)$ leads to $O(1)$ anomalous dimensions and variations in OPE coefficients, despite the fact that we have truncated the asymptotic expansion for $H(\bar h)$ to only a few terms. Our tentative conclusion is that by using the appropriate twist Hamiltonian, the large-spin expansion can in practice be extended down to relatively small spins for all double-twist operators in the 3d Ising CFT (and perhaps other theories as well).

In section~\ref{sec:knot}, we ask what the asymptotic large spin expansion is missing. Part of the four-point function is invisible to this expansion, to all orders in $1/\ell$. Demanding that this part be crossing-symmetric gives additional nontrivial constraints on the CFT data. Using our analytical approximations from section~\ref{sec:asymptoticsappliedtoising}, we briefly explore some of these constraints. For example, we find conditions that approximately determine $c_T$ and $f_{\s\s\e}$ in terms of $\De_\s,\De_\e, f_{\e\e\e}$, using only the lightcone limit.

We discuss future directions in section~\ref{sec:discussion}.

\section{Numerics and the lightcone limit}
\label{sec:numericsandlightconelimit}

\subsection{A numerical picture of the 3d Ising spectrum}
\label{sec:numericalpicture}

Numerical bootstrap methods have become powerful enough to estimate several operator dimensions and OPE coefficients in the 3d Ising CFT\@. The strategy is as follows. Consider the four-point functions $\<\s\s\s\s\>$, $\<\s\s\e\e\>$, and $\<\e\e\e\e\>$ where $\s$ and $\e$ are the lowest-dimension $\Z_2$-odd and $\Z_2$-even scalars in the 3d Ising CFT, respectively. Crossing symmetry and unitarity for these correlators forces the dimensions $\De_\s,\De_\e$ and OPE coefficients $f_{\s\s\e}, f_{\e\e\e}$ to lie inside a tiny island given by~\cite{Kos:2016ysd}
\be
\label{eq:isingisland}
\De_{\s} &= 0.5181489(10), & f_{\s\s\e} &= 1.0518537(41),\nn\\
\De_\e   &= 1.412625(10), & f_{\e\e\e} &= 1.532435(19).
\ee
We can then ask: given that $(\De_\s,\De_\e, f_{\s\s\e}, f_{\e\e\e})$ lie in this island, what other operators are needed for crossing symmetry? Although it is possible in principle to compute rigorous bounds on more operators, it is difficult in practice because we must scan over the dimensions and OPE coefficients of those additional operators. 

Instead, we adopt the non-rigorous approach of~\cite{Komargodski:2016auf}, based on the extremal functional method~\cite{Poland:2010wg,ElShowk:2012hu,El-Showk:2014dwa}.  Consider $N$ derivatives of the crossing equation around $z=\bar z = \frac 1 2$, which we write as $\cF_N=0$, where $\cF_N$ is an $N$-dimensional vector depending on the CFT data.  We assume that OPE coefficients are real and operator dimensions are consistent with unitarity bounds~\cite{Mack:1975je}.  By the argument of~\cite{Rattazzi:2008pe}, there is an allowed region $\mathcal{A}_N$ in the space of CFT data such that any point outside $\mathcal{A}_N$ is inconsistent with $\cF_N=0$.\footnote{The island (\ref{eq:isingisland}) is the projection of $\cA_{1265}$ onto $(\De_\s,\De_\e,f_{\s\s\e},f_{\e\e\e})$-space, where we also assume that $\s$ and $\e$ are the only relevant scalars in the theory.} For every point $p$ on the boundary of $\cA_N$, there is a unique ``partial spectrum" $\mathcal{S}_N(p)$: a finite list of operator dimensions and OPE coefficients that solve $\cF_N=0$. The number of operators in $\cS_N(p)$ grows linearly with $N$.\footnote{It is impossible to solve the full crossing equations with a finite number of operators. $\cS_N(p)$ can be finite because we have truncated the crossing equations to $\cF_N=0$.}

If $p$ lies on the boundary of the Ising island and $N$ is large, we might expect that $\cS_N(p)$ is a reasonable approximation to the actual spectrum of the theory. However, it is not obvious how to assign error bars to $\cS_N(p)$. Firstly, the actual theory lies somewhere in the interior of the island, not on the boundary. It is important that the island is small enough that points on the interior are close to points on the boundary. Secondly, $\cS_N(p)$ depends on $p$, and there is no canonical choice of $p$.

In~\cite{Komargodski:2016auf}, we propose the following trick. We sample several different points $p$ on the boundary of the island, and compute $\cS_N(p)$ for each one. As we increase $N$ and vary $p$, some of the operators in $\cS_N(p)$ jump around, while others remain relatively stable. If an operator remains stable, we can guess that it is truly required by crossing symmetry.

In~\cite{Komargodski:2016auf}, we used this strategy to estimate the dimensions and OPE coefficients of a few low-dimension operators in the 3d Ising CFT\@. In figures~\ref{fig:evenspectrum} and~\ref{fig:oddspectrum}, we show a more complete computation, giving about a hundred stable operators. To produce figures~\ref{fig:evenspectrum} and~\ref{fig:oddspectrum}, we computed 60 different spectra by varying $(\De_\s,\De_\e, f_{\s\s\e}, f_{\e\e\e})$ and minimizing $c_T$. (We give more details in appendix~\ref{app:extremalfunctional}.) We then superimposed these 60 spectra, and grouped together operators with dimensions closer than $0.03$. Each circle represents a group, and the size of the circle is proportional to the number of operators in that group. Thus, large circles correspond to stable operators and small circles correspond to unstable operators. We list the dimensions and OPE coefficients of the stable operators in appendix~\ref{app:severaloperators}. Most of the stable operators also appear in figures~\ref{fig:tauSigSig0}, \ref{fig:fSigSigSigSig0}, \ref{fig:fEpsEpsSigSig0}, \ref{fig:tauSigEps0}, \ref{fig:fSigEps0}, \ref{fig:tauEpsEps0AndSigSig1}, \ref{fig:fSigSigSigSig1AndfSigSigEpsEps0}, and~\ref{fig:fEpsEpsSigSig1AndfEpsEpsEpsEps0}, where we compare to analytics.

\begin{figure}[!ht]
\begin{center}
\includegraphics[width=\textwidth]{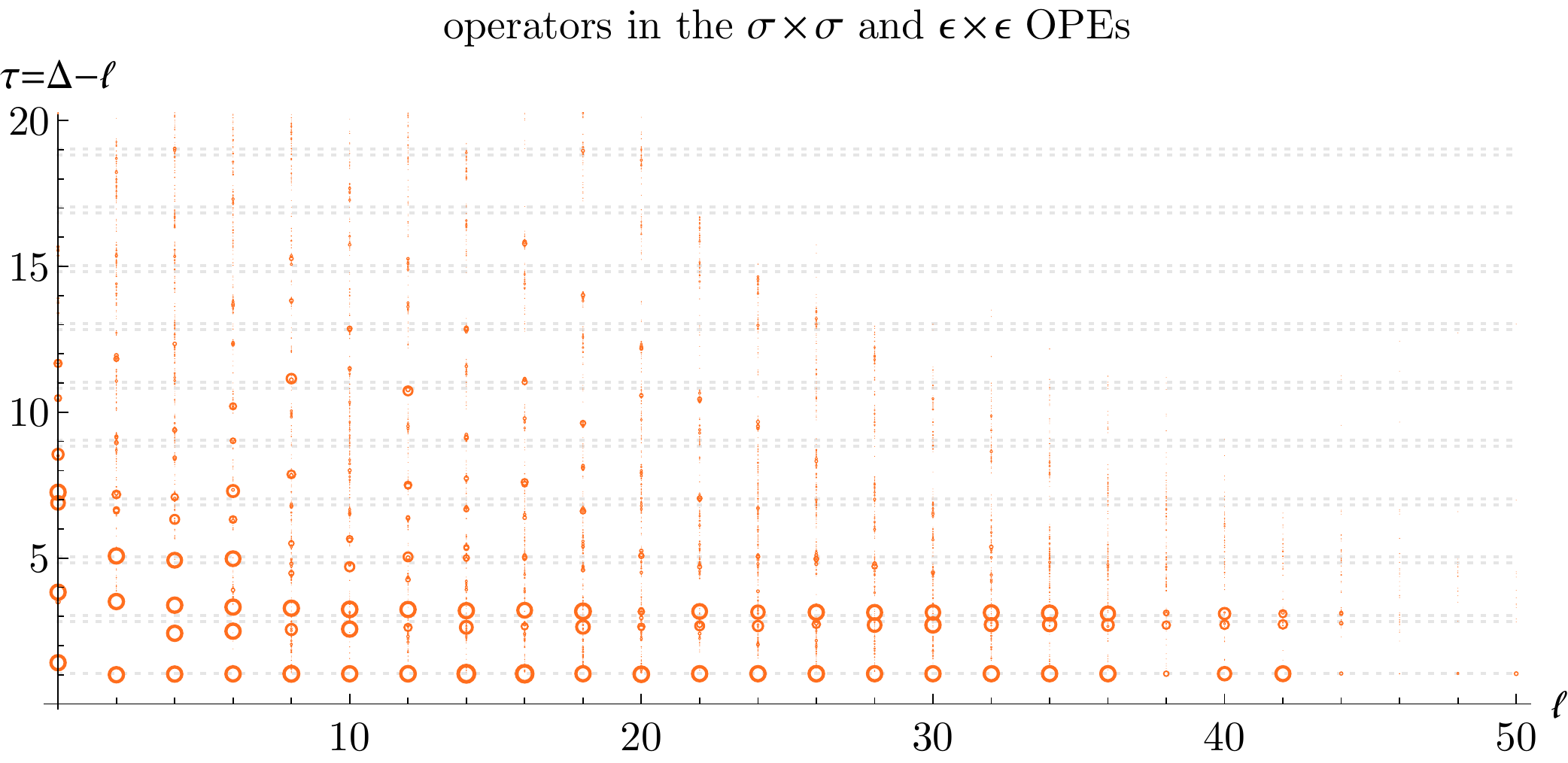}
\end{center}
\caption{
\label{fig:evenspectrum}
Estimates of $\Z_2$-even operators in the 3d Ising model. Larger circles represent ``stable" operators whose dimensions and OPE coefficients have small errors in our computation. We plot the twist $\De-\ell$ versus spin $\ell$. The grey dashed lines are $\tau=2\De_\s+2n$ and $\tau=2\De_\e+2n$ for nonnegative integer $n$.}
\end{figure}

\begin{figure}[!ht]
\begin{center}
\includegraphics[width=\textwidth]{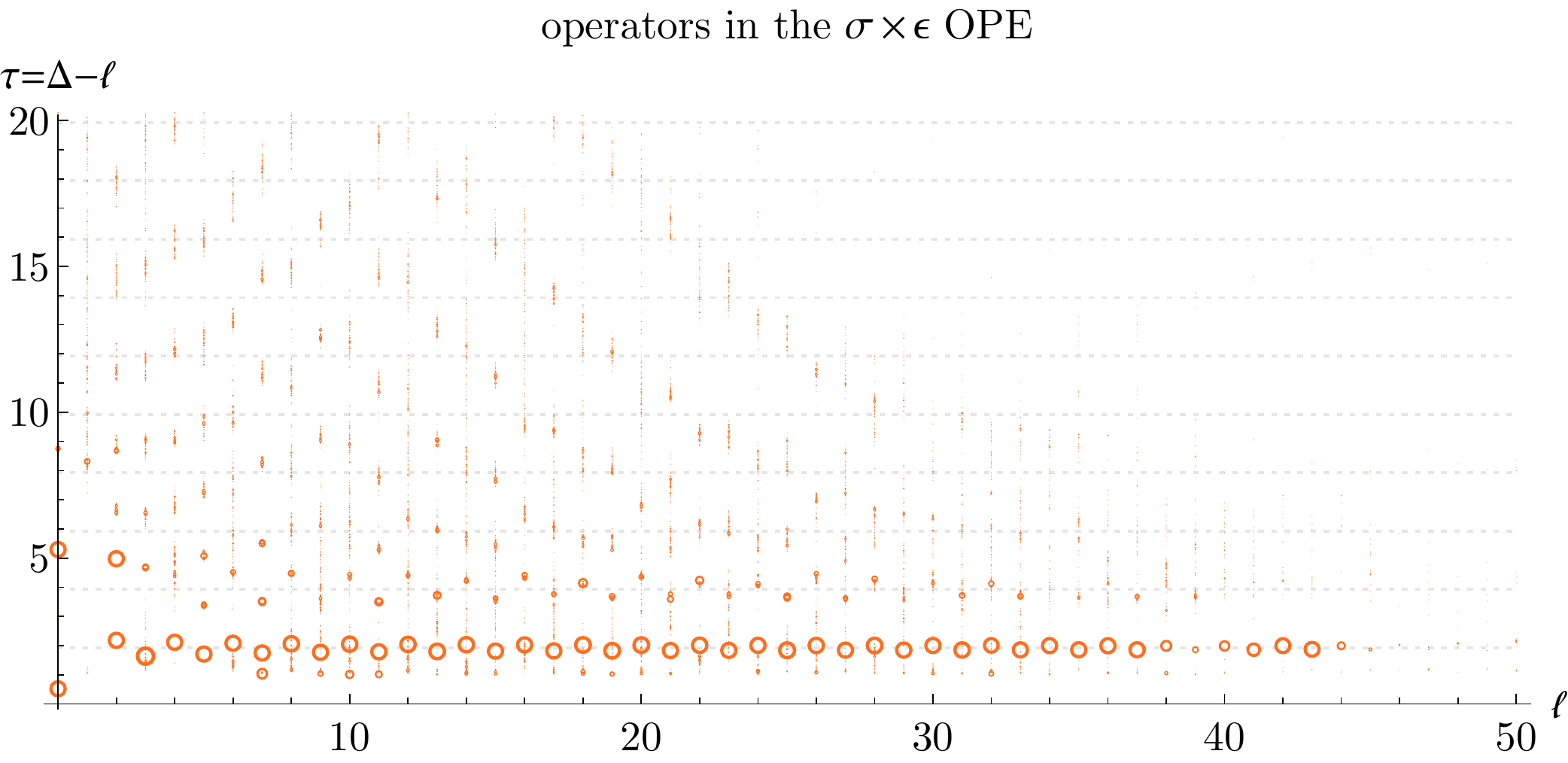}
\end{center}
\caption{
\label{fig:oddspectrum}
Estimates of $\Z_2$-odd operators in the 3d Ising model. Larger circles represent ``stable" operators. We plot the twist $\De-\ell$ versus spin $\ell$. The grey dashed lines are $\tau=\De_\s+\De_\e+2n$ for nonnegative integer $n$.}
\end{figure}

\subsection{Effectiveness of the large spin expansion}

Let us make some comments about these results. Firstly, most of the stable operators fall into families with increasing spin and nearly constant twist $\tau=\Delta-\ell$.
We immediately recognize these as double-twist operators --- specifically the families $[\s\s]_0$, $[\s\s]_1$, $[\e\e]_0$ in figure~\ref{fig:evenspectrum}, and $[\s\e]_0$ in figure~\ref{fig:oddspectrum}. (There are also vague hints of $[\s\e]_1$.) The fact that these families are stable implies that they play a crucial role in the numerical bootstrap for the 3d Ising CFT\@.\footnote{Note that even though our numerical calculation uses an expansion of the crossing equation around the Euclidean point $z=\bar z=\frac 1 2$, the results are sensitive to the Lorentzian physics of the lightcone limit. The prevailing lore was that, since the conformal block expansion converges exponentially in $\De$ in the Euclidean regime~\cite{Pappadopulo:2012jk}, numerical bootstrap methods should only be sensitive to low-dimension operators. Evidently this is incorrect because certain derivatives probe physics outside the Euclidean regime. Some hints that the numerical bootstrap probes the lightcone limit were given in~\cite{Mazac:2016qev}, where an exact extremal functional was constructed that involves the lightcone limit of conformal blocks.}

One can compute anomalous dimensions of double-twist operators in a large-$\ell$ expansion using the crossing equation~\cite{Fitzpatrick:2012yx,Komargodski:2012ek,Alday:2015eya,Alday:2015ota,Alday:2015ewa,Kaviraj:2015cxa,Kaviraj:2015xsa,Alday:2016mxe,Alday:2016njk,Alday:2016jfr}. The authors of~\cite{Alday:2015ewa} observed that the large-$\ell$ expansion appears to be asymptotic, but they conjectured that the anomalous dimensions of $[\s\s]_0$ should be well-described by the first few terms in this expansion, coming from the operators $\e$ and $T_{\mu\nu}$ appearing in the $\s\x\s$ OPE\@. The expansion is most naturally organized in terms of the ``conformal spin" $J$ defined by
\be
\label{eq:conformalspin}
J(\ell)^2 &\equiv \p{\frac {\tau(\ell)} 2 + \ell}\p{\frac {\tau(\ell)} 2 + \ell - 1}.
\ee
One finds\footnote{Note that $\tau_{[\s\s]_0}(\ell)$ depends on $J$, and $J$ depends on $\tau_{[\s\s]_0}(\ell)$. To obtain a series in $\ell$, one can repeatedly substitute the expressions for $\tau_{[\s\s]_0}(\ell)$ and $J$ into each other, starting with the initial seed $J=(\De_\s+\ell)(\De_\s+\ell-1)$.}
\be
\label{eq:largespinexpansion}
\tau_{[\s\s]_0}(J) &\approx 2\De_\s + \sum_{\cO=\e,T_{\mu\nu}}f_{\s\s\cO}^2 \frac{c_0(\tau_\cO,\ell_\cO)}{J^{\tau_\cO}}\p{1+\frac{c_1(\tau_\cO,\ell_\cO)}{J^2}},\nn\\
c_0(\tau,\ell) &\equiv -\frac{2(-1)^{\ell} \G(\tau+2\ell)\G(\De_\s)^2}{\G(\De_\s-\frac {\tau} 2)^2 \G(\ell+\frac {\tau} 2)^2},\nn\\
c_1(\tau,\ell) &\equiv -\frac{(-2 \Delta_\s +\tau +2)^2 \left(2 \ell^2 (d+\tau -2)+2 \ell (\tau -1) (d+\tau -2)+(d-4) \tau^2\right)}{8 (d+2 \ell-4) (d-2 (\ell+\tau +1))},\nn\\
&\quad\,-\frac{\tau}{12}(-3 \tau \Delta_\s +\tau^2 +3 \tau+2),
\ee
where
\be
f_{\s\s T_{\mu\nu}}^2 &= \frac{d \De_\s^2}{4(d-1)}\frac{c_\mathrm{free}}{c_T}.
\ee
Here, $d=3$ is the spacetime dimension and $c_\mathrm{free}$ is $c_T$ for the free boson~\cite{Osborn:1993cr}. We will rederive (\ref{eq:largespinexpansion}) and find its all-orders generalization in section~\ref{sec:allorderslightcone}.  Plugging in (\ref{eq:isingisland}) and the value 
\be
\frac{c_T}{c_\mathrm{free}} &\approx 0.946534(11)
\ee
computed in~\cite{El-Showk:2014dwa}, we find that this prediction fits the numerics beautifully, even at small $\ell$ (figure~\ref{fig:tauSigSig0})! This is surprising because the arguments leading to (\ref{eq:largespinexpansion}) only fix anomalous dimensions at asymptotically large $\ell$. Rigorously speaking, they say nothing about any finite value of $\ell$.

Nevertheless, inspired by this result, we might try to match the dimensions of $[\e\e]_0$ and $[\s\s]_1$ to the leading terms in their large-spin expansions. Unfortunately, the naive analytic predictions disagree wildly with the data. 
To fit $[\e\e]_0$ and $[\s\s]_1$, we will need a more sophisticated understanding of the large-spin expansion, which we develop over the course of this work.

A clue about what's going on is the fact that the twists of $[\e\e]_0$ and $[\s\s]_1$ move away from each other at small $\ell$. This is reminiscent of the behavior of the eigenvalues of
\be
\begin{pmatrix}
\tau_1 & 1/\ell \\
1/\ell & \tau_2
\end{pmatrix}
\ee
as $\ell\to 0$. If $|\tau_1-\tau_2|$ is small, the eigenvalues repel more at small $\ell$. Furthermore, the small-$\ell$ eigenvectors become nontrivial admixtures of the large-$\ell$ eigenvectors. In the 3d Ising CFT, it turns out that $\tau_1=2\De_\s+2\approx 3.04$ is numerically close to $\tau_2=2\De_\e\approx 2.83$. This suggests that the repulsion between $[\e\e]_0$ and $[\s\s]_1$ is due to large mixing  between these families. We will make this notion more precise in section~\ref{sec:twisthamiltonian} and compute the twists of $[\e\e]_0$ and $[\s\s]_1$ in section~\ref{sec:eezerosigsigone}. The off-diagonal terms will come from the $\s$ operator in the $\s\x\e$ OPE and behave like $\ell^{-\De_\s}$.

\section{Lightcone bootstrap review}
\label{sec:lightconereview}

\subsection{Double-twist operators}
\label{sec:doubletwistoperators}

Let us review the argument from~\cite{Fitzpatrick:2012yx,Komargodski:2012ek} for the existence of double-twist operators.
The crossing symmetry equation for a four-point function of scalar operators $\<\f\f\f\f\>$ is
\be
(z\bar z)^{-\De_\f} \sum_{\cO} f_{\f\f\cO}^2 g_{\De,\ell}(z,\bar z) &= ((1-z)(1-\bar z))^{-\De_\f}\sum_{\cO} f_{\f\f\cO}^2 g_{\De,\ell}(1-z,1-\bar z).
\ee
Here, $\cO$ runs over primary operators in the $\f\x\f$ OPE and $\De,\ell$ are the dimension and spin of $\cO$.  The functions $g_{\De,\ell}(z,\bar z)$ are conformal blocks for the $d$-dimensional conformal group $\SO(d,2)$.

The lightcone limit is given by $z\ll 1-\bar z \ll 1$.\footnote{By developing methods for summing infinite families of operators, we will eventually work in the limit $z,1-\bar z \ll 1$ (with no restrictions on $1-\bar z$ relative to $z$), sometimes called the ``double lightcone limit." We mostly abuse terminology and continue to call this the lightcone limit.} Let us replace $\bar z \to 1- \bar z$ so that we have
\be
\label{eq:variableswelike}
(z(1-\bar z))^{-\De_\f} \sum_{\cO} f_{\f\f\cO}^2 g_{\De,\ell}(z,1-\bar z) &= ((1-z)\bar z)^{-\De_\f}\sum_{\cO} f_{\f\f\cO}^2 g_{\De,\ell}(\bar z,1-z),
\ee
and the lightcone limit becomes $z\ll \bar z \ll 1$.
(We have used $g_{\De,\ell}(1-z,\bar z)=g_{\De,\ell}(\bar z, 1-z)$.) In this limit, the left-hand side is dominated by the unit operator,  $z^{-\De_\f}(1+O(\bar z))$.  On the right-hand side, no single term dominates the small $z$ limit.  However, because we also have small $\bar z$, we can replace each conformal block by its expansion in small $\bar z$~\cite{DO1,DO2},
\be
g_{\De,\ell}(\bar z,1-z) &= \bar z^{h} k_{2\bar h}(1-z) + O(\bar z^{h+1}),\\
k_{2 h}(x) &\equiv x^h {}_2F_1(h,h,2h,x),
\ee
where\footnote{These definitions are conventional in 2d CFT\@. In this work, we are considering $d>2$, but it is still convenient to use $h,\bar h$.}
\be
h \equiv \frac{\De-\ell}{2}=\frac{\tau}{2}, \qquad\qquad \bar h \equiv \frac{\De+\ell}{2}=\frac{\tau}{2}+\ell.
\ee
The function $k_{2 h}(x)$ is a conformal block for the 1-dimensional conformal group $\SL(2,\R)$.  Our equation becomes
\be
\label{eq:lightlimitofcrossing}
z^{-\De_\f} + \dots &= \sum_{\cO} f_{\f\f\cO}^2 \bar z^{h-\De_\f} k_{2\bar h}(1-z) + \dots.
\ee

The left-hand side of (\ref{eq:lightlimitofcrossing}) has a power-law singularity at small $z$.  However, each individual term on the right-hand side has a logarithmic singularity at small $z$,\footnote{$\psi(x)=\frac{\G'(x)}{\G(x)}$ is the digamma function.}
\be
\label{eq:approxofhypergeometricaroundone}
k_{2\bar h}(1-z) &= -\frac{\Gamma (2 \bar h)}{\Gamma (\bar h)^2} (2 \psi(\bar h)-2\psi(1)+\log (z)) + O(z \log z).
\ee
A power singularity can only come from the sum over an infinite number of operators on the right-hand side with $\bar h\to \oo$. Also, these operators must have $h\to \De_\f$ as $\bar h\to \oo$ to match fact that $z^{-\De_\f}$ on the left-hand side is independent of $\bar z$.  These are the double-twist operators $[\f\f]_0$.

One can determine the asymptotic growth of the OPE coefficients $f_{\f\f[\f\f]_0}$ by demanding that they reproduce the singularity $z^{-\De_\f}$. The leading growth is
\be
\label{eq:asymptoticdensity}
f_{\f\f[\f\f]_0}(\bar h)^2 &\sim \frac{2^{3-2\bar h}\sqrt{\pi}}{\G(\De_\f)^2} \bar h^{2\De_\f-\frac 3 2}.
\ee
The sum in (\ref{eq:lightlimitofcrossing}) is dominated by the regime $2\bar h \sqrt z\sim 1$,\footnote{The fact that $2\bar h \sqrt{z}\sim 1$ is the appropriate regime was shown in~\cite{Fitzpatrick:2012yx}. It also follows from the physical arguments of~\cite{Alday:2007mf}.} where the $\SL(2,\R)$ block becomes
\be
\label{eq:besselapprox}
k_{2\bar h}(1-z) &\approx 2^{2\bar h} \sqrt{\frac{\bar h}{\pi}}K_0(2\bar h\sqrt z)\qquad (\bar h \gg 1,\ 2\bar h\sqrt z\ \textrm{fixed}),
\ee
where $K_0(x)$ is a modified Bessel function. We can then approximate the sum over $[\f\f]_0$ as an integral, which reproduces the required singularity
\be
\sum_{\cO\in [\f\f]_0} f_{\f\f\cO}^2 k_{2\bar h}(1-z) &\approx \frac 1 2 \int d\bar h \frac{8}{\G(\De_\f)^2} \bar h^{2\De_\f-1}K_0(2\bar h\sqrt z)=z^{-\De_\f}.
\ee
(The factor of $\frac 1 2$ is because only even spin operators appear in $[\f\f]_0$.) 

Matching $z^{-\De_\f}$ only determines the asymptotic density of OPE coefficients $f_{\f\f[\f\f]_0}^2$ at large $\bar h$. The density (\ref{eq:asymptoticdensity}) could be distributed evenly, with one operator per spin, or with one operator every other spin, or in many different ways. We will not see evidence of this freedom when we compare to numerics. The OPE coefficients will always be distributed in the simplest way consistent with the large-spin expansion. 

We can determine the anomalous dimensions of double-twist operators by matching additional terms on the left-hand side of (\ref{eq:lightlimitofcrossing}). Let $\cO_0$ be the smallest-twist operator in the $\f\x\f$ OPE that is not the unit operator (often $\cO_0 = T_{\mu\nu}$).  Including the contribution of $\cO_0$ at small $z$ on the left-hand side of (\ref{eq:lightlimitofcrossing}), we have
\be
 &z^{-\De_\f} + z^{h_0-\De_\f} k_{2\bar h_0}(1-\bar z) + \dots \\
 &= z^{-\De_\f} + z^{h_0-\De_\f} \p{-\frac{\G(2\bar h_0)}{\G(\bar h_0)^2}(2\psi(\bar h_0)-2\psi(1)+\log \bar z) + \dots}+ \dots,\nn\\
&= \sum_{\cO} f_{\f\f\cO}^2 \bar z^{h-\De_\f} k_{2\bar h}(1-z) + \dots
\label{eq:alsologzbarterm}
\ee
where we have used (\ref{eq:approxofhypergeometricaroundone}), this time on the left-hand side of the crossing equation.
To match the $\log \bar z$ term, we can take
\be
\label{eq:anomdimlargespin}
h_{[\f\f]_0}(\bar h) &= \De_\f + \de(\bar h),\nn\\
f_{\f\f[\f\f]_0}(\bar h)^2 \de(\bar h) &\sim -\frac{\G(2\bar h_0)}{\G(\bar h_0)^2}\frac{2^{3-2\bar h}\sqrt\pi}{\G(\De_\f-h_0)^2}\bar h^{2\De_\f-2h_0-\frac 3 2}.
\ee
Dividing (\ref{eq:anomdimlargespin}) by (\ref{eq:asymptoticdensity}) gives the leading large-$\bar h$ expansion of the anomalous dimension $\de(\bar h)\propto \bar h^{-2h_0}$, agreeing with the leading term in (\ref{eq:largespinexpansion}).\footnote{$\de=\frac{\g}{2}$ is half of what is usually called the anomalous dimension.} Again, only the asymptotic density of the combination $f_{\f\f[\f\f]_0}^2 \de$ is determined by this computation.

\tikzset{
solid/.style={thick},
wavy/.style={thick, decorate, draw=black,
    decoration={coil,aspect=0}}
 }

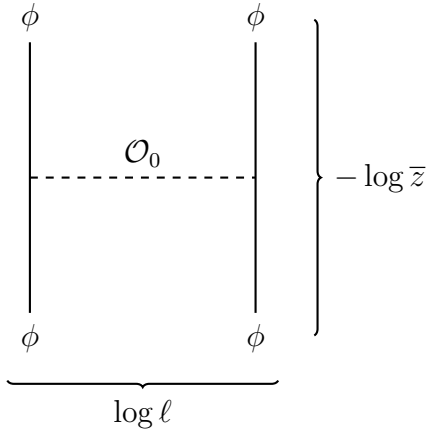
\begin{figure}
\begin{center}
\begin{tikzpicture}[xscale=0.6,yscale=0.6]
\draw[thick, decoration={brace}, decorate] (6.3,6.5) -- (6.3,-0.5);
\draw[thick, decoration={brace, mirror}, decorate] (-0.5,-1.5) -- (5.5,-1.5);
\node[right] at (6.5,3) {$-\log \bar z$};
\node[below] at (2.5,-1.7) {$\log \ell$};
\draw[thick] (0,0) -- (0,6);
\draw[thick] (5,0) -- (5,6);
\draw[thick, dashed] (0,3) -- (5,3);
\node[below] at (0,0) {$\phi$};
\node[below] at (5,0) {$\phi$};
\node[above] at (0,6) {$\phi$};
\node[above] at (5,6) {$\phi$};
\node[above] at (2.5,3) {$\cO_0$};
\end{tikzpicture}
\end{center}
\caption{
A diagram representing the contribution of the exchange of $\cO_0$ in one channel (left to right) to anomalous dimensions and OPE coefficients of double-twist operators $[\f\f]_n$ in the other channel (bottom to top). In the physical picture of~\cite{Alday:2007mf}, this diagram shows the exchange of virtual $\cO_0$-particles between $\phi$-particles separated by a distance $\log \ell$ over time $\log \bar z$.
}
\label{fig:diagramexample}
\end{figure}

An interesting feature of this argument (not realized in~\cite{Fitzpatrick:2012yx,Komargodski:2012ek}, but pointed out in~\cite{Alday:2015eya}) is that it most naturally determines a function $h(\bar h)$ instead of $\tau(\ell)$. We obtain actual operator dimensions by demanding that the spin be an even integer,
\be
\bar h - h(\bar h) &= \ell,\qquad \ell \in \{0,2,\dots\}.
\ee
Thinking in terms of $h(\bar h)$ will be even more important when we compute higher-order corrections to (\ref{eq:anomdimlargespin}).

It is often useful to draw the contribution of $\cO_0$ as a ``large-spin diagram" like figure~\ref{fig:diagramexample} (see, e.g.~\cite{Fitzpatrick:2015qma}). Such diagrams are particularly natural in the language of~\cite{Alday:2007mf}, where large-spin operators become widely separated particles in a massive two-dimensional effective theory. Figure~\ref{fig:diagramexample} represents a Yukawa potential between $\phi$-particles induced by exchange of a virtual massive $\cO_0$-particle. The distance between $\phi$'s (the width of the figure) is given by $\chi = \log \ell$, and the mass of $\cO_0$ is the twist $m=\tau_0$. The Yukawa potential has the form $e^{-m \chi} = \ell^{-\tau_0}$, in agreement with the large-$\ell$ behavior of $\de(\bar h)$. We can also think of figure~\ref{fig:diagramexample} as having height $-\log \bar z$, so that integration over the vertical position of the $\cO_0$ exchange gives a factor $-\log \bar z$, matching the $-\log \bar z$ term in the conformal block of $\cO_0$ (\ref{eq:alsologzbarterm}).

\subsubsection{What about $\log^n \bar z$?}
\label{sec:whatabout}

Above, we matched the $\log \bar z$ terms on the left-hand side of (\ref{eq:alsologzbarterm}) to anomalous dimensions on the right-hand side.
However, the expansion of $\bar z^{\de}$ contains higher-order terms in $\log \bar z$:
\be
\label{eq:obviousexpansionofexponential}
\bar z^{\de} &= 1 + \de \log \bar z + \frac{\de^2}{2} \log^2 \bar z + \dots.
\ee
What do they map to under crossing? Using (\ref{eq:asymptoticdensity}), (\ref{eq:besselapprox}), and  (\ref{eq:anomdimlargespin}), the $\log^n \bar z$ terms become
\be
\label{eq:integrategammasq}
\sum_\cO f_{\f\f\cO}^2 \de^n \log^n\bar z\, k_{2\bar h}(1-z) &\sim z^{nh_0 - \De_\f}\log^n \bar z.
\ee 
The $z$-dependence of (\ref{eq:integrategammasq}) is what one would expect from an operator of weight $n h_0$. Such operators exist: they are the multi-twist operators $[\cO_0 \dots \cO_0]_0$. The $\log^n \bar z$ behavior is not present in any individual conformal block --- instead it must come from a sum over all the operators in the family $[\cO_0 \dots \cO_0]_0$. We will see examples of $\log^2 \bar z$ coming from a sum over double-twist operators in sections~\ref{sec:asymptoticsappliedtoising} and~\ref{sec:twisthamiltonian}. We prove that double-twist operators always account for the correct $\log^2 \bar z$ terms (i.e.\ that exponentiation of $\de \log \bar z$  works automatically to second order) in appendix~\ref{app:boxdiagrams}.

\begin{figure}
\begin{center}
\begin{tikzpicture}[xscale=0.6,yscale=0.6]
\draw[thick] (0,0) -- (0,6);
\draw[thick] (5,0) -- (5,6);
\draw[thick, dashed] (0,1.5) -- (5,1.5);
\draw[thick, dashed] (0,2.5) -- (5,2.5);
\draw[thick, dashed] (0,4.5) -- (5,4.5);
\node[below] at (0,0) {$\phi$};
\node[below] at (5,0) {$\phi$};
\node[above] at (0,6) {$\phi$};
\node[above] at (5,6) {$\phi$};
\node[below] at (2.5,1.6) {$\cO_0$};
\node[below] at (2.5,2.6) {$\cO_0$};
\node[above] at (2.5,4.5) {$\cO_0$};
\node[below] at (2.5,4.35) {$\vdots$};
\end{tikzpicture}
\end{center}
\caption{
Exponentiation of the contribution of $\cO_0$ in the bottom-to-top channel becomes an exchange of multi-twist operators $[\cO_0\dots\cO_0]$ in the left-to-right channel.
}
\label{fig:multitwistexchange}
\end{figure}
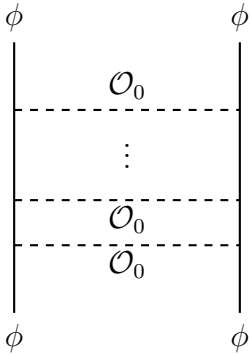

We could have immediately guessed this result using large-spin diagrams. Exponentiating the Yukawa potential in figure~\ref{fig:diagramexample} gives a sum of ``ladder diagrams" like figure~\ref{fig:multitwistexchange}. Reading these diagrams from left-to-right, they look like an exchange of multi-twist operators. If we interpret the figure as having height $\log \bar z$, then integration over the vertical positions of the exchanges gives $\log^n \bar z$. In practice, $n-1$ ``integrations" are achieved by summing over different distributions of derivatives among the operators $\ptl\cdots\ptl \cO_0 \cdots \ptl\cdots\ptl \cO_0$ (i.e.\ by summing over all members of the twist family $[\cO_0\cdots \cO_0]$), while one integration is encoded in the $\log \bar z$ factor in each individual conformal block. This makes it clear why we must sum over all multi-twist operators $[\cO_0\dots\cO_0]$ in one channel to recover exponentiation in the other channel.

In this way, crossing symmetry forces multi-twist operators $[\cO_1\dots \cO_n]$ to appear in the conformal block expansion whenever $\cO_1,\dots, \cO_n$ do individually. In particular, this implies that multi-twist operators built from the stress tensor and other low-spin operators should appear in the $\s\x\s$ OPE in the 3d Ising model. In figure~\ref{fig:evenspectrum}, we see some evidence of operators with twist near $2$, which would correspond to $[T T]_0$. However, none of them are numerically stable. This is likely because the anomalous dimension $\de_{[\s\s]_0}$ is small (of order $10^{-2}$), so higher terms in the expansion of $\bar z^{\de}$ (\ref{eq:obviousexpansionofexponential}) are highly suppressed. To get a better picture of these operators, one must study mixed correlators involving $\s$ and $T_{\mu\nu}$ together, or perhaps higher-point correlators like $\<\s\s\s\s\s\s\>$. We return to this point in section~\ref{sec:allowmixing}.

\subsection{The Casimir trick}
\label{sec:casimirtrick}

The derivation of (\ref{eq:anomdimlargespin}) makes sense when $h_0 < \De_\f$, so that the sum
\be
\label{eq:gammamatching}
 z^{h_0-\De_\f} &\approx \sum_{\bar h} f_{\f\f[\f\f]_0}(\bar h)^2 \de(\bar h) k_{2\bar h}(1-z)
\ee
diverges faster than any individual term ($\log z$) at small $z$. When this happens, the sum must be dominated by large $\bar h$ and can be approximated by an integral.\footnote{When $\cO_0 = T_{\mu\nu}$, we always have $h_{T_{\mu\nu}}=\frac{d-2}{2} \leq \De_\f$, by unitarity.} However,~\cite{Alday:2015eya} argued that the large-spin expansion can be extended to include contributions from operators with $h_0>\De_\f$.  For example, there is a calculable correction to $\de_{[\s\s]_0}$ in the 3d Ising CFT coming from $\e$, which has $h_\e\approx 0.7 > \De_\s \approx 0.52$.

To see why, suppose $h_0 > \De_\f$.  Since each term $k_{2\bar h}(1-z)$ is more singular than $z^{h_0-\De_\f}$, we cannot conclude that the sum is dominated by large $\bar h$.  However, $k_{2\bar h}(1-z)$ obeys a Casimir differential equation with eigenvalue $\bar h(\bar h - 1)$,
\be
\label{eq:caseq}
\cD k_{2\bar h}(1-z) &= \bar h(\bar h-1) k_{2\bar h}(1-z),\nn\\
\cD &\equiv (1-z)^2 z \ptl_z^2 + (1-z)^2 \ptl_z.
\ee
By repeatedly acting with the Casimir operator $\cD$ on a power $z^a$, we can make it arbitrarily singular,\footnote{$(a)_n=\frac{\G(a+n)}{\G(a)}$ is the Pochhammer symbol.}
\be
\cD^n z^{a} &= (a-n+1)_n^2 z^{a - n}(1+O(z)).
\ee
Acting $n$ times on (\ref{eq:gammamatching}), we obtain
\be
 (h_0-\De-n+1)_n^2 z^{h_0-\De_\f-n} + \dots &\approx \sum_{\bar h} f_{\f\f[\f\f]_0}(\bar h)^2 \de(\bar h) \p{\bar h(\bar h - 1)}^n k_{2\bar h}(1-z).
\ee
Taking $n$ big, the right-hand side is now dominated by large $\bar h$ when $z$ is small, and we can proceed as before.  The resulting correction to $f_{\f\f[\f\f]_0}^2 \de$ is again given by (\ref{eq:approxofhypergeometricaroundone}).

\subsection{Higher-order corrections}
\label{sec:higherorder}

By including $1/\bar h$-corrections in the approximation (\ref{eq:besselapprox}), one can compute higher-order corrections to the OPE coefficients $f_{\f\f[\f\f]_0}(\bar h)$ and anomalous dimensions $\de(\bar h)$.  After applying the Casimir operator enough times, each term in the $1/\bar h$-expansion contributes to a singularity at small $z$, and can thus be calculated by approximating the sum over $\bar h$ as an integral. This gives expansions of the form
\be
\label{eq:asymptoticsoffandgamma}
f_{\f\f[\f\f]_0}(\bar h)^2 &\sim 2^{-2\bar h} \sum_\cO \bar  h^{2\De_\f-2h_\cO-\frac 3 2}\p{a_0+\frac{a_1}{\bar h} + \frac{a_2}{\bar h^2} + \dots },\nn\\
f_{\f\f[\f\f]_0}(\bar h)^2\de(\bar h) &\sim 2^{-2\bar h} \sum_{\cO\neq 1}\bar h^{2\De_\f-2h_\cO-\frac 3 2}\p{b_0+\frac{b_1}{\bar h} + \frac{b_2}{\bar h^2} + \dots }.
\ee
The authors of~\cite{Alday:2015ewa} showed how to use the Casimir trick to compute the above coefficients. (Actually, they organize their expansion in terms of the Casimir eigenvalue $J^2 = \bar h(\bar h - 1)$, as in equations (\ref{eq:conformalspin}) and (\ref{eq:largespinexpansion}).) In section~\ref{sec:allorderslightcone}, we will write down an all-orders solution for    (\ref{eq:asymptoticsoffandgamma}).

We have written ``$\sim$" to indicate that both sides have the same asymptotic expansion at large $\bar h$. The arguments above only fix the asymptotic expansion of $f_{\f\f[\f\f]_0}(\bar h)^2$ and $\de(\bar h)$ because it is always possible to throw away a finite number of blocks $k_{2\bar h}(1-z)$ and still match the power $z^{h_\cO-\De_\f}$ on the other side of the crossing equation. We can only fix the behavior of $f_{\f\f[\f\f]_0}(\bar h)^2$ and $\de(\bar h)$ for $\bar h$ larger than some $\bar h_0$, where $\bar h_0$ might {\it grow\/} as we include more terms in (\ref{eq:asymptoticsoffandgamma}). Thus, (\ref{eq:asymptoticsoffandgamma}) should be interpreted as asymptotic series.

The behavior of $f_{\f\f[\f\f]_0}(\bar h)^2$ and $\de(\bar h)$ at finite $\bar h$ is still important --- it contributes to ``Casimir-regular" terms defined in the following section.

\section{Sums of $\SL(2,\R)$ blocks}
\label{sec:families}

Our main tool will be a better understanding of infinite sums of $\SL(2,\R)$ blocks,
\be
\label{eq:examplesum}
\sum_{\ell=0}^\oo p(h_\ell) k_{2h_\ell}(1-z),
\ee
where $h_\ell$ is an increasing series of weights that asymptotes to integer spacing, and $p(h)$ are coefficients that grow no faster than $2^{-2h}h^\textrm{const.}$ as $h\to \oo$.
We start from a simple example, Mean Field Theory (MFT) in 1-dimension, and then generalize it in several ways.

\subsection{Casimir-singular vs. Casimir-regular terms}

Sums of $\SL(2,\R)$ blocks have two parts that play different roles in the bootstrap.
As discussed in in section~\ref{sec:casimirtrick}, we can make a power $z^a$ arbitrarily singular by repeatedly applying the Casimir operator $\cD$,
\be
\cD^n z^a = (a-n+1)_n^2 z^{a-n}(1+O(z)).
\ee
We say that $z^a$ for generic $a$ is {\it Casimir-singular}.
An exception occurs when $a$ is a nonnegative integer, since then $(a-n+1)_n^2$ vanishes for $n\geq a+1$. In fact, terms of the form
\be
\label{eq:casimirregexamples}
z^n,\ z^n \log z\qquad n\in \Z_{\geq 0}
\ee
do not become arbitrarily singular when we repeatedly apply $\cD$. We call such terms {\it Casimir-regular}. 

The lesson of section~\ref{sec:higherorder} is that Casimir-singular terms can be matched unambiguously to an asymptotic expansion in large $h$.  Furthermore, to compute coefficients in this expansion, we can think of the sum over $h$ as an integral. By contrast, Casimir-regular terms are not determined by a large-$h$ expansion. This is consistent with the fact that a single $\SL(2,\R)$ block $k_{2\bar h}(1-z)$ is Casimir-regular, since it is an eigenvector of the Casimir operator. (We can also see that it is Casimir-regular by noting that its $z$-expansion (\ref{eq:approxofhypergeometricaroundone}) is a sum of terms of the form (\ref{eq:casimirregexamples}).)  For example, suppose
\be
\sum_{m=0}^\oo p(h_m) k_{2 h_m}(1-z) &= f(z).
\ee
Moving the first term on the left-hand side to the right-hand side, we have
\be
\sum_{m=1}^\oo p(h_m) k_{2 h_m}(1-z) &= f(z) - p(h_0) k_{2 h_0}(1-z)
\ee
The Casimir-regular part of the right-hand side has changed, but the large-$h$ expansion of $p(h)$ obviously hasn't.

It will often be useful to work modulo Casimir-regular terms. When we do so, we denote Casimir-regular terms by $[\dots]_z$.

\subsection{Matching a power-law singularity}
\label{sec:sl2blocksums}

Casimir-singular terms match to a unique asymptotic expansion for coefficients of $\SL(2,\R)$ blocks at large $h$. We can find the right expansion by looking at an example.  Consider the conformal block expansion of $\<\f_1(0)\f_2(z)\f_2(1)\f_1(\oo)\>$, where $\f_{1,2}$ are scalars of dimension $\De/2$ in 1-dimensional MFT,
\be
\label{eq:mftinoned}
\sum_{\ell=0}^\oo \frac{(\De)_\ell^2}{\ell!(\ell+2\De-1)_\ell} k_{2\De+2\ell}(z) &= \p{\frac{z}{1-z}}^{\De}.
\ee
Replacing $z\to 1-z$ and writing $\De=-a$, this can be written
\be
\label{eq:mftrewritten}
\sum_{\substack{h=-a+\ell \\ \ell=0,1,\dots}} S_{a}(h) k_{2h}(1-z) &= y^{a},
\ee
where
\be
y &\equiv \frac{z}{1-z},\\
S_a(h) &\equiv \frac{1}{\G(-a)^2}\frac{\G(h)^2}{\G(2h-1)} \frac{\G(h-a-1)}{\G(h+a+1)}.
\ee
Many formulae will be much simpler in the variable $y$ (and $\bar y$, defined similarly) instead of $z$. Note that $y^a$ is Casimir-singular for generic $a$, while $y^n$ and $y^n \log y$ for nonnegative integer $n$ are Casimir-regular. We will denote Casimir-regular terms by $[\dots]_y$. The crossing transformation $z\to 1-z$ maps $y\to 1/y$.

Casimir-singular terms can only come from an infinite sum of blocks, and they are sensitive only to the asymptotic density of OPE coefficients. Thus, if we change the weights $h$ entering (\ref{eq:mftrewritten}), while preserving the same asymptotic density, only the Casimir-regular terms should change.  For example, changing $-a+\ell\to h_0+\ell$, we expect
\be
\label{eq:generalizedweights}
\sum_{\substack{h=h_0+\ell \\ \ell=0,1,\dots }} S_a(h) k_{2h}(1-z) &= y^a + [\dots]_y.
\ee
The Casimir-singular term $y^a$ is independent of $h_0$, but the Casimir-regular terms $[\dots]_y$ depend on $h_0$. As a sanity check, (\ref{eq:generalizedweights}) is certainly true when $h_0=-a+n$ for nonnegative integer $n$, since we get it by moving the first $n$ terms of (\ref{eq:mftrewritten}) to the right-hand side.

The coefficients $S_a(h)$ will be our building blocks for solving the asymptotic lightcone bootstrap. They encode the all-orders large-$h$ expansion needed to match powers $y^a$. By taking linear combinations, we can match any Casimir-singular term we want. For example, to match an $\SL(2,\R)$ block $k_{2h'}(z)$ in the crossed channel, we can take a linear combination of $S_{h'+m}(h)$ which can be resummed into a ${}_4F_3$ hypergeometric function.

Casimir-regular terms depend on the detailed structure of the weights being summed over.
We can determine the Casimir-regular terms in (\ref{eq:generalizedweights}) as follows. Let us expand $k_{2h}(1-z)$ in small $y$ (equivalently small $z$) inside the sum,
\be
\label{eq:sl2expansionaroundone}
k_{2h}(1-z) &= -\frac{\G(2h)}{\G(h)^2}\sum_{k=0}^\oo \pdr{}{k}\p{T_{-k-1}(h)y^k},\\
\label{eq:expandinside}
\sum_{\substack{h=h_0+\ell \\ \ell=0,1,\dots }} S_a(h) k_{2h}(1-z) 
&=\sum_{\substack{h=h_0+\ell \\ \ell=0,1,\dots }} (1-2h) T_a(h) \sum_{k=0}^\oo \pdr{}{k}\p{T_{-k-1}(h) y^k}.
\ee
Here, we have introduced 
\be
T_a(h) &\equiv \frac{\G(2h-1)}{\G(h)^2} S_a(h) = \frac{1}{\G(-a)^2}\frac{\G(h-a-1)}{\G(h+a+1)}.
\ee

Naively, we might try to switch the order of summation in (\ref{eq:expandinside}),
\be
\sum_{\substack{h=h_0+\ell \\ \ell=0,1,\dots }} S_a(h) k_{2h}(1-z) &\stackrel{?}{=}\sum_{k=0}^\oo \pdr{}{k}\p{y^k\sum_{\substack{h=h_0+\ell \\ \ell=0,1,\dots }} (1-2h) T_a(h)T_{-k-1}(h)}.
\ee
  However, this cannot be correct. If the result converged, it would be Casimir-regular, a contradiction. Indeed, the summand
\be
(1-2h) T_a(h) T_{-k-1}(h)
\ee
grows like $h^{2k-2a-1}$, so the terms with $k>a$ diverge. However, let us analytically continue from the region $a>k$ for each term. After some gymnastics,\footnote{We obtained (\ref{eq:TTsumidentity}) in the following shameful way. When $b=-1$, we can use
\be
\label{eq:gammasumidentity}
\sum_{\ell=0}^\oo \frac{\G(\ell+\a)}{\G(\ell+\b)} &=  \frac{1}{\b-\a-1}\frac{\G(\a)}{\G(\b-1)}.
\ee
When $b=-k-1$ with $k$ a positive integer, $T_{-k-1}(h)$ is a polynomial in $h$ and we can write $T_a(h) T_{-k-1}(h)$ in terms of linear combinations of terms of the form $\frac{\G(\ell+\a)}{\G(\ell+\b)}$ and use (\ref{eq:gammasumidentity}). We did this for several positive integer $k$'s, guessed an answer for general $k$, analytically it continued away from integer $k$, and then checked the result numerically.
} we find
\be
\label{eq:TTsumidentity}
\sum_{\substack{h=h_0+\ell \\ \ell=0,1,\dots }} (1-2h) T_a(h) T_b(h) &= -\frac{(a+h_0)(b+h_0)}{a+b+1} T_a(h_0) T_b(h_0) \equiv \cA_{a,b}(h_0).
\ee

We claim that (\ref{eq:TTsumidentity}) gives the correct coefficients for the Casimir-regular terms in (\ref{eq:generalizedweights}). That is, we have the remarkable exact identity (equation~(\ref{eq:examplesl2sum}) from the introduction)
\be
\label{eq:magicSL2sum}
\sum_{\substack{h=h_0+\ell \\ \ell=0,1,\dots }} S_a(h) k_{2h}(1-z) &= y^a + \sum_{k=0}^\oo \pdr{}{k}\p{\cA_{a,-k-1}(h_0)y^k}.
\ee
One can verify that (\ref{eq:magicSL2sum}) is consistent with the fact that shifting $h_0\to h_0+1$ changes both sides by $-S_a(h_0)k_{2h_0}(1-z)$. We have also extensively checked (\ref{eq:magicSL2sum}) numerically.\footnote{We expect (\ref{eq:magicSL2sum}) can be derived using Sturm-Liouville theory for $\SL(2,\R)$ blocks~\cite{MatthijsBaltFuture}.} We slightly generalize  (\ref{eq:magicSL2sum}) in equation (\ref{eq:mixedblockremarkable}). The special case of this formula where $h_0=0$ was proven recently in~\cite{Perlmutter:2016pkf}, using hypergeometric function identities from~\cite{opac-b1078126}.

The key feature of (\ref{eq:magicSL2sum}) is that it expresses an integer-spaced family of conformal blocks in one channel as an expansion in the other channel. Since families of nearly integer-spaced operators are ubiquitous, we can use (\ref{eq:magicSL2sum}) as a building block for understanding crossing symmetry in general.

\subsection{General coefficients}
\label{sec:generalcoefficients}

Consider a sum of $\SL(2,\R)$ blocks with general coefficients $p(h)$ and integer-spaced weights,
\be
\label{eq:generalcoeffsum}
\sum_{\substack{h=h_0+\ell \\ \ell = 0,1,\dots}} p(h) k_{2h}(1-z).
\ee
If $p(h)$ has the same large-$h$ behavior as a sum of $S_a(h)$'s, the structure of (\ref{eq:generalcoeffsum}) will be similar to (\ref{eq:magicSL2sum}).  To determine the Casimir-singular terms, we match asymptotic expansions,
\be
\label{eq:matchasymptoticofp}
p(h) &\sim \sum_{a\in A} c_a S_a(h)\qquad (h\to \oo),
\ee
where $A$ is some discrete (possibly infinite) set of values depending on the function $p(h)$, and $c_a$ are constants.
We then have
\be
\sum_{\substack{h=h_0+\ell \\ \ell = 0,1,\dots}} p(h) k_{2h}(1-z) &= \sum_{a\in A} c_a y^a + [\dots]_y.
\ee

To compute the Casimir-regular terms, we expand $k_{2h}(1-z)$ inside the sum and then naively switch the order of summation,
\be
-\sum_{k=0}^\oo \pdr{}{k}\p{y^k \sum_{\substack{h=h_0+\ell \\ \ell = 0,1,\dots}} p(h) \frac{\G(2h)}{\G(h)^2}T_{-k-1}(h)}
\ee
Again, the sums in parentheses are divergent for sufficiently large $k$.  However, we can regulate them by adding and subtracting linear combinations of the known answer (\ref{eq:magicSL2sum}) until the sums become convergent.  This gives
\be
\sum_{h=h_0+\ell} p(h) k_{2h}(1-z) &= \sum_{a\in A} c_a y^a + \sum_{k=0}^\oo y^k\p{\a_k[p](h_0)\log y + \b_k[p](h_0)}\\
\a_k[p](h_0) &\equiv \sum_{\substack{a\in A \\ a<K}} c_a \cA_{a,-k-1}(h_0) -\sum_{\substack{h=h_0+\ell \\ \ell = 0,1,\dots}} \p{p(h)-\sum_{\substack{a\in A \\ a<K}} c_a S_a(h)} \frac{\G(2h)}{\G(h)^2}T_{-k-1}(h)\nn\\
\b_k[p](h_0) &\equiv \pdr{}{k} \a_k[p](h_0).
\label{eq:bk}
\ee
If we choose $K\geq k$, then the sum over $h$ in (\ref{eq:bk}) will converge. In fact, the larger we take $K$, the more quickly the sums  converge (since the quantity in parentheses falls off more quickly with $h$). Note that $\a_k$ is analytic in $k$, so we can evaluate its derivative $\b_k$.

\subsection{Non-integer spacing and reparameterization invariance}
\label{sec:nonintegerspacing}

We often encounter sums over $\SL(2,\R)$ blocks $k_{2h}(1-z)$ where the weights $h$ are not integer-spaced. The Casimir-singular terms depend only on the asymptotic density of OPE coefficients.  Thus, for a sequence $h_\ell$ that depends sufficiently nicely on $\ell$, we can compensate for uneven spacing by inserting a factor of $\pdr{h_\ell}{\ell}$, giving the same Casimir-singular part as an integer-spaced sum:
\be
\label{eq:nonintegerspace}
\sum_{\ell=0}^\oo \pdr{h_\ell}{\ell} p(h_\ell) k_{2h_\ell}(1-z) &= \sum_{\ell=0}^\oo p(h_0+\ell) k_{2(h_0+\ell)}(1-z) + [\dots]_y.
\ee
A way to understand (\ref{eq:nonintegerspace}) is that Casimir-singular terms come from asymptotically large $h$, where the sum can be treated as an integral. We are then free to redefine the integration variable and include a Jacobian $\pdr{h_\ell}{\ell}$.  We call this freedom ``reparameterization invariance."

Let us prove (\ref{eq:nonintegerspace}) for an important class of $h_\ell$.  Suppose $h_\ell$ is defined implicitly by
\be
\label{eq:classofh}
h_\ell = h_0 + \ell + \de(h_\ell),
\ee
where $\de(h)$ is an analytic function that behaves like a sum of powers $h^{-b}$ as $h\to \oo$. We have
\be
\pdr{h_\ell}{\ell} &= 1+\pdr{\de(h_\ell)}{h_0} = \p{1 - \de'(h_\ell)}^{-1}.
\ee
Working modulo Casimir-regular terms, we may restrict the sum (\ref{eq:nonintegerspace}) to $\ell\geq L$ for some large $L$ so that $\de(h)$ is small.  Expanding (\ref{eq:nonintegerspace}) in small $\de$, we find the following identity:
\be
\label{eq:derivativemanipulations}
\pdr{h_\ell}{\ell} p(h_\ell) k_{2h_\ell}(1-z) &= \p{1+\pdr{\de(h_\ell)}{h_0}} p(h_\ell) k_{2h_\ell}(1-z)\nn\\
&= \sum_{k=0}^\oo \ptl_{h_0}^k\p{\frac{\de(h_0+\ell)^k}{k!} p(h_0+\ell) k_{2(h_0+\ell)}(1-z)}.
\ee
(One way to motivate why an identity like (\ref{eq:derivativemanipulations}) should exist is to pretend the sum over $\ell$ is an integral and consider an infinitesimal change of variables in the integral.)
Now summing over $\ell$, the terms in parentheses are integer-spaced sums of the type in section~\ref{sec:generalcoefficients}.  They give Casimir-singular contributions that are independent of $h_0$.  Thus, only $k=0$ contributes in (\ref{eq:derivativemanipulations}), modulo Casimir-regular terms. This proves (\ref{eq:nonintegerspace}).

Another way to understand (\ref{eq:nonintegerspace}) is as follows.  The non-integer-spaced sum can be written as a contour integral
\be
\label{eq:contourrepresentation}
\sum_{\ell=0}^\oo \pdr{h_\ell}{\ell} p(h_\ell) k_{2h_\ell}(1-z)
&=
\oint_{-\e-i\oo}^{-\e+i\oo}dh \frac{\pi}{\tan(\pi(h-h_0-\de(h)))} p(h) k_{2h}(1-z).
\ee
The Casimir-singular part must come from the region of the integral $h\to \pm i\oo$, since any sum of blocks with bounded $h$ is Casimir-regular.  However, in this region the $\de$-dependent factor in the integrand approaches a $\de$-independent constant exponentially quickly (assuming $\de(h)$ grows slower than $h$ as $h\to \pm i\oo$):
\be
\frac{\pi}{\tan(\pi(h-h_0-\de(h)))} &\to \mp 1 + O(e^{\mp 2s})\qquad (h=\pm i s).
\ee
Thus, the Casimir-singular part is $\de$-independent and can be obtained by replacing $\tan(\pi(h-h_0-\de(h)))\to \tan(\pi h)$.\footnote{This point of view suggests that reparameterization invariance holds for any $\de(h)$ that grows slower than $h^{1-\e}$ for some $\e>0$ as $h\to \pm i\oo$. In particular, this includes logarithmically growing $\de(h)$, as in Regge trajectories in conformal gauge theories.}

Sums over general weights $h$ with general coefficients $p(h)$ can be computed using the same strategy as in section~(\ref{sec:generalcoefficients}). We obtain Casimir-singular terms from the asymptotic expansion of $p(h)$. We determine Casimir-regular terms by expanding $k_{2h}(1-z)$ inside the sum, naively reversing the order of summation, and regulating the resulting sums over $h$. We give more details in appendix~\ref{app:nonintegerspacing}.

\subsection{Alternating sums and even integer spacing}
\label{sec:alternating}

We will also encounter sums of $\SL(2,\R)$ blocks with insertions of $(-1)^\ell$. To understand these, consider the conformal block expansion of $\<\f_1(0)\f_2(z)\f_1(1)\f_2(\oo)\>$ in 1-dimensional MFT,
\be
\label{eq:mftinonedalternating}
\sum_{\ell=0}^\oo (-1)^\ell \frac{(\De)_\ell^2}{\ell!(\ell+2\De-1)_\ell} k_{2\De+2\ell}(z) &= z^{\De}.
\ee
Substituting $z\to 1-z$ and $\De\to -a$, this can be written
\be
\sum_{\substack{h=-a+\ell \\ \ell=0,1,\dots}} (-1)^\ell S_a(h) k_{2h}(1-z) &= (1+y)^a.
\ee
Note that $(1+y)^a$ is Casimir-regular.
Using the logic of the preceding sections, we conclude that general sums with $(-1)^\ell$ insertions are Casimir-regular,
\be
\sum_\ell (-1)^\ell \pdr{h}{\ell} S_a(h) k_{2h}(1-z) &= [\dots]_y,
\ee
where $h=h_\ell$ is any sequence of the form discussed in section (\ref{sec:nonintegerspacing}).

Let us describe how to compute the Casimir-regular terms in alternating sums.  For simplicity, consider the case of integer-spaced weights and general coefficients $p(h)$,
\be
\sum_{\substack{h=h_0+\ell \\ \ell=0,1,\dots}} (-1)^\ell p(h) k_{2h}(1-z).
\ee
The strategy is the same as before: we expand $k_{2h}(1-z)$ at small $y$, switch the order of summation, and regulate the resulting sums by adding and subtracting known answers. We find
\be
\sum_{h=h_0+\ell} (-1)^\ell p(h) k_{2h}(1-z) &= \sum_{k=0}^\oo y^k\p{\a_k^-[p](h_0)\log y + \b_k^-[p](h_0)},
\ee
where
\be
\a_k^-[p](h_0) &= \sum_{\substack{a\in A \\ a<K}} c_a \cA^-_{a,-k-1}(h_0) -\sum_{\substack{h=h_0+\ell \\ \ell = 0,1,\dots}} (-1)^\ell\p{p(h)-\sum_{\substack{a\in A \\ a<K}} c_a S_a(h)} \frac{\G(2h)}{\G(h)^2}T_{-k-1}(h),\nn\\
\b^-_k[p](h_0) &= \pdr{}{k} \a^-_k[p](h_0).
\ee
Again, $c_a$ are defined by matching asymptotic expansions $p(h)\sim \sum_{a\in A} c_a S_a(h)$.
The quantity $\cA^-_{a,b}(h_0)$ is given by
\be
\cA_{a,b}^-(h_0) &\equiv \sum_{\substack{h=h_0+\ell \\ \ell=0,1,\dots }} (-1)^\ell (1-2h) T_a(h) T_b(h),
\ee
analytically continued in $a$ from the region where the sum converges.

We have not found a simple closed-form expression for $\cA^-_{a,b}(h_0)$ for general $a,b$.  However, we can evaluate it to arbitrary accuracy as follows.  Using similar tricks to before, we can compute the case $b=-1$:
\be
\cA^-_{a,-1}(h_0) &= \sum_{\substack{h=h_0+\ell \\ \ell=0,1,\dots }} (-1)^\ell (1-2h) T_a(h) = -(h_0+a)T_a(h_0).
\ee
This can be used to regularize the sum for general $b\neq -1$.  Note that $T_a(h)T_b(h)$ has the same large-$h$ expansion as
\be
\label{eq:matchexpansions}
T_a(h)T_b(h) &\sim \sum_{k=0}^\oo t_{a,b}(k) T_{a+b+k+1}(h)\nn\\
t_{a,b}(k) &\equiv \frac{\G(-1-a-b)^2}{\G(-a)^2\G(-b)^2}  \frac{(a+1)_k(b+1)_k}{(a+b+2)_k}\frac{(-1)^k}{k!}.
\ee
Thus, we have
\be
\label{eq:convergentexprforaminus}
\cA_{a,b}^-(h_0) &=  \sum_{k = 0}^K t_{a,b}(k)\cA^-_{a+b+k+1,-1}(h_0) \nn\\
&\quad
+ \sum_{\substack{h=h_0+\ell \\ \ell=0,1,\dots}} (-1)^\ell (1-2h)\p{T_a(h)T_b(h) - \sum_{k=0}^K t_{a,b}(k) T_{a+b+k+1}(h)},
\ee
where $K > -a-b-5/2$ is taken large enough that the sum over $h$ converges. The larger we take $K$, the faster the sum converges. When $a$ or $b$ is a negative integer, the expansion (\ref{eq:matchexpansions}) truncates and becomes an equality, and we can omit the second line in (\ref{eq:convergentexprforaminus}).

We will also need to evaluate sums with even-integer spacing. These are an average of alternating and non-alternating sums,
\be
\sum_{\substack{h=h_0+\ell \\ \ell = 0,2,\dots}} p(h) k_{2h}(1-z) &= \frac 1 2 \sum_{a\in A} c_a y^a + \sum_{k=0}^\oo y^k\p{\a_k^\mathrm{even}[p](h_0) \log y + \b_k^\mathrm{even}[p](h_0)},\nn\\
\a_k^\mathrm{even}[p](h_0) &\equiv \frac 1 2 \p{\a_k[p](h_0)+\a^-_k[p](h_0)},\nn\\
\b_k^\mathrm{even}[p](h_0) &\equiv \frac 1 2 \p{\b_k[p](h_0)+\b^-_k[p](h_0)}.
\ee
Similarly, we define
\be
\cA^\mathrm{even}_{a,b}(h_0) &\equiv \frac 1 2 \p{\cA_{a,b}(h_0)+\cA^-_{a,b}(h_0)}.
\ee

\subsection{Mixed blocks}

Correlation functions of operators $\<\f_1\f_2\f_3\f_4\>$ with different scaling dimensions can be expanded in $\SL(2,\R)$ blocks of the form
\be
\label{eq:mixedsl2block}
k_{2h}^{r,s}(z) &\equiv z^h (1-z)^{-r} {}_2F_1(h-r,h+s,2h,z),
\ee
where $r=h_1-h_2, s=h_3-h_4$.
We include the unconventional factor $(1-z)^{-r}$ because it simplifies several formulae later on. It also ensures that $k^{r,s}_{2h}(z)$ is symmetric in $r$ and $s$, by elementary hypergeometric function identities.  Casimir-regular terms for the mixed block (\ref{eq:mixedsl2block}) are of the form $y^{n-r}$ and $y^{n-s}$ for nonnegative integer $n$.

The mixed block analog of $S_a(h)$ is
\be
S_a^{r,s}(h) &\equiv \frac{1}{\G(-a-r)\G(-a-s)}\frac{\G(h - r)\G(h - s)}{\G(2 h - 1)} \frac{\G(h - a - 1)}{\G(h + a + 1)}.
\ee
These coefficients satisfy the 1-dimensional MFT equation
\be
\sum_{\substack{\bar h = \ell - a \\ \ell = 0,1,\dots}} S_a^{r,s}(h) k_{2h}^{r,s}(1-z) &= y^a,
\ee
and its generalization in the spirit of the previous sections\footnote{The meaning of $[\dots]_y$ depends on what type of $\SL(2,\R)$ blocks we are summing over. Here, it refers to terms of the form $y^{n-r}$ and $y^{n-s}$.  For the case $r=s=0$, it refers to terms of the form $y^n$ and $y^n \log y$.}
\be
\label{eq:mixedsl2sumgeneralization}
\sum_{\ell=0}^\oo \pdr{h}{\ell} S_a^{r,s}(h) k_{2h}^{r,s}(1-z) &= y^a + [\dots]_y.
\ee

Using (\ref{eq:TTsumidentity}), we also find a generalized version of (\ref{eq:magicSL2sum}) giving the explicit Casimir-regular terms in an integer-spaced sum of mixed blocks 
\be
&\sum_{\substack{h=h_0+\ell \\ \ell=0,1,\dots}} S_a^{r,s}(h)k_{2h}^{r,s}(1-z)\nn\\
 &= y^a + \frac{\pi}{\sin(\pi(s-r))}\frac{\G(-a)^2}{\G(-a-r)\G(-a-s)}\sum_{k=0}^\oo \p{\frac{\G(k+1-r)^2 \cA_{a,r-k-1}(h_0)}{\G(k+1+s-r)k!}y^{k-r} - (r\leftrightarrow s)}.
\label{eq:mixedblockremarkable}
\ee

\section{Large spin asymptotics to all orders}
\label{sec:allorderslightcone}

\subsection{Basic idea}

Equipped with the results of section~\ref{sec:families}, we can solve the asymptotic lightcone bootstrap. The idea is to expand both sides of the crossing equation in $y,\bar y$ and match $y^a$ on one side to $S_a(\bar h)$ on the other. For the lowest family of double-twist operators $[\f\f]_0$, we have an equation of the form (\ref{eq:alsologzbarterm}), which in the $y$ variables reads
\be
\label{eq:lowesttwistfamily}
y^{-2h_\f} + \sum_i y^{h_i-2h_\f}\p{A_i \log \bar y + B_i + O(\bar y)} &= \sum_{\cO \in [\f\f]_0} f_{\f\f\cO}^2 \bar y^{h_\cO-2h_\f} k_{2\bar h_\cO}(1-z) + \dots.
\ee
Here, ``$\dots$" represents other operators that are unimportant for this computation.  Note that the $y$ variables make the unit operator block very simple. For other operators, expanding in $y$ instead of $z$ is equivalent to shuffling around contributions of descendants.

The $h_i$ are weights of primary and descendant operators in the $\f\x\f$ OPE\@. We can match the left-hand side by choosing
\be
h_{[\f\f]_0}(\bar h) &= 2h_\f + \de_{[\f\f]_0}(\bar h), \\
f_{\f\f[\f\f]_0}(\bar h)^2 &= \pdr{\bar h}{\ell}\l_{\f\f[\f\f]_0}(\bar h)^2  = \p{1-\pdr{\de_{[\f\f]_0}(\bar h)}{\bar h}}^{-1} \l_{\f\f[\f\f]_0}(\bar h)^2,
\label{eq:getopecoeff}
\ee
where
\be
\label{eq:examplelargespinallorders}
\l_{\f\f[\f\f]_0}(\bar h)^2 &\sim 2S_{-2h_\f}(\bar h) + 2\sum_i B_i S_{h_i-2h_\f}(\bar h),\nn\\
\l_{\f\f[\f\f]_0}(\bar h)^2 \de_{[\f\f]_0}(\bar h) &\sim 2 \sum_i A_i S_{h_i-2h_\f}(\bar h).
\ee
Here, ``$\sim$" means the two sides have the same large-$\bar h$ expansion.
We include factors of $2$ in (\ref{eq:examplelargespinallorders}) because the family $[\f\f]_0$ only contains even spin operators. Dividing, we find
\be
\label{eq:gammaequationallorders}
\de_{[\f\f]_0}(\bar h) &\sim \frac{\sum_i A_i S_{h_i-2h_\f}(\bar h)}{S_{-2h_\f}(\bar h) + \sum_i B_i S_{h_i-2h_\f}(\bar h)}.
\ee
Once we know $\de_{[\f\f]_0}(\bar h)$, we can obtain the OPE coefficients $f_{\f\f[\f\f]_0}$ from (\ref{eq:getopecoeff}). Expanding in large $\bar h$ gives a series with terms of the form $1/\bar h^{2(h_{i_1} + \dots + h_{i_k})+n}$.  

In (\ref{eq:gammaequationallorders}), we can see explicitly why the large-spin expansion for $\de_{[\f\f]_0}(\bar h)$ is naturally organized in terms of the Casimir eigenvalue $J^2=\bar h(\bar h - 1)$ as discussed in~\cite{Alday:2015eya}.  The reason is that ratios of $S_a(\bar h)$ are also ratios of $T_a(\bar h)=\frac{\G(2\bar h -1)}{\G(\bar h)^2}S_a(\bar h)$, which has a series expansion in $J^2$,
\be
T_a(\bar h) &= \frac{1}{J^{2a}} \p{t_0+ \frac{t_2}{J^2} + \frac{t_4}{J^4} + \dots}.
\ee

We have suppressed an important subtlety in (\ref{eq:lowesttwistfamily}). The OPE $\f\x\f$ contains an infinite number of operators with bounded $h$ (for example, the families $[\f\f]_n$) themselves. Thus the sum on the left-hand side,
\be
\sum_i y^{h_i} (A_i \log \bar y + B_i),
\ee
may not converge. 
For simplicity, suppose all the $h_i=h$ are the same. The correct procedure is to perform the sum over $i$ first, before expanding in $\bar y$, using the methods of section~\ref{sec:sl2blocksums}. This leads to
\be
\label{eq:correctlhssum}
y^h \sum_a c_a \bar y^a + y^h(\bar A \log \bar y + \bar B + O(\bar y)),
\ee
where $\bar A$ and $\bar B$ are regularized versions of the sums over $A_i$ and $B_i$.  The $\bar y^a$ terms are Casimir-singular in $\bar y$, and will be cancelled by other operators on the right-hand side of (\ref{eq:lowesttwistfamily}).  The remaining $\bar y$-Casimir-regular (but still $y$-Casimir-singular) terms $y^h \bar A \log \bar y$ and $y^h \bar B$ contribute to anomalous dimensions and OPE coefficients of $[\f\f]_0$, respectively.
The $\bar y$-Casimir-singular terms in (\ref{eq:correctlhssum}) can also include $\log^n \bar y$ contributions related to higher-order exponentiation of anomalous dimensions, and discussed in section~\ref{sec:whatabout}. We will see several examples in section~\ref{sec:asymptoticsappliedtoising}.

Thus, the techniques of section~\ref{sec:sl2blocksums} for summing $\SL(2,\R)$ blocks have two roles to play. Firstly, they let us match Casimir-singular terms in one channel to $\bar h$-dependence in the other channel. Secondly, they let us resum operators whose twists have accumulation points.

Naively this leads to an impasse: we must resum $[\f\f]_0$ before finding how it contributes to its own anomalous dimensions $\de_{[\f\f]_0}$. However, it turns out that $[\f\f]_0$ contributes to its own anomalous dimensions only at order $\de_{[\f\f]_0}^2$ and higher. (This is related to the fact that Mean Field Theory has no anomalous dimensions.) Thus, both the resummation and the matching to $\bar h$-dependence will be possible. We will see this explicitly in section~\ref{sec:contributiontoself}.

\subsection{Why asymptotic?}

We have been careful to write ``$\sim$" instead of ``$=$" because the relations~(\ref{eq:examplelargespinallorders}) are not necessarily equalities. In fact, taken literally, the expressions on the right-hand side may not even converge to functions of $\bar h$. 
Instead, they represent equivalence classes of functions with the same asymptotic expansions at large $\bar h$. For example, both sides of
\be
\label{eq:equivalenceclasses}
T_b(\bar h) T_a(\bar h) &\sim \frac{\G(-1-a-b)^2}{\G(-a)^2\G(-b)^2} \sum_{k=0}^\oo \frac{(a+1)_k(b+1)_k}{(a+b+2)_k}\frac{(-1)^k}{k!} T_{a+b+k+1}(\bar h)
\ee
formally have the same large-$\bar h$ expansion, but they are different. In fact, the sum on the right diverges. We must interpret (\ref{eq:equivalenceclasses}) in terms of large-$\bar h$ equivalence classes.

The asymptotic nature of the large-$\bar h$ expansion for double-twist operators makes mathematical and physical sense. Mathematically, a given Casimir-singular term only determines an asymptotic density of coefficients on the other side of the crossing equation. Any change in the density at finite $\bar h$ contributes to Casimir-regular terms. Thus, we cannot fix the actual function of $\bar h$ without simultaneously considering all Casimir-regular terms.

Physically, it is ambiguous which twist family (if any) we should assign a given operator to. For instance, should we assign $T_{\mu\nu}$ to the family $[\s\s]_0$, or should the family should start at spin-$4$ or higher? Twist families only make sense as infinite collections of operators with unbounded spin. We shouldn't necessarily expect to write analytic expressions that interpolate between their OPE coefficients and dimensions at finite $\ell$. On the other hand, we might expect a convergent large-$\bar h$ expansion for an object that packages together {\it all\/} operators in the theory, and does not try to distinguish them into twist families.

When our theory has extra structure, twist families may become well-defined even at finite spin. For example, in a large-$N$ expansion, we have a well-defined classification of operators into single-trace, double-trace, etc.. Consequently, large-$\bar h$ equivalence classes in large-$N$ theories should have distinguished representatives. See, for example, in~\cite{Aharony:2016dwx}. Similar remarks hold in weakly-coupled theories.

\subsection{General double-twist families}
\label{sec:generaldtfamilies}

Let us be more explicit and derive all-orders expansions for OPE coefficients and anomalous dimensions of double twist families $[\f_i\f_j]_n$ for all $n\geq 0$. For generality, we study mixed four-point functions $\<\f_1\f_2\f_3\f_4\>$ of scalars with possibly different external dimensions.
 
We use a slightly unconventional definition for $\SO(d,2)$ blocks,
\be
G_{h,\bar h}^{r,s}(z,\bar z) &\equiv  ((1-z)(1-\bar z))^{-r}g^{2r,2s}_{h+\bar h,\bar h - h}(z,\bar z),
\ee
where $g^{\De_{12},\De_{34}}_{\De,\ell}(z,\bar z)$ are the mixed scalar blocks of~\cite{DO3} with coefficient $c_\ell=1$.\footnote{Our blocks differ from those of~\cite{Kos:2014bka} by $G_\mathrm{ours}(u,v) = v^{-\frac{\De_{12}}{2}} (-1)^\ell \frac{4^\De (2\nu)_\ell}{(\nu)_\ell} g_{\mathrm{theirs}}(u,v)$.} Using identities from~\cite{DO3}, one can show that our $G_{h,\bar h}^{r,s}(z,\bar z)$ is symmetric under $r\leftrightarrow s$. The extra factors $((1-z)(1-\bar z))^{-r}=v^{-r}$ simplify the crossing equations in the $y,\bar y$ variables and make the symmetry between $r$ and $s$ manifest.  For brevity, we omit $r,s$ when they are zero.

The four-point function $\<\f_1\f_2\f_3\f_4\>$ has conformal block expansion\footnote{The ordering $f_{12\cO}f_{43\cO}$ differs from the $f_{12\cO}f_{34\cO}$ ordering in~\cite{Kos:2014bka} because our blocks differ by $(-1)^\ell$ times positive factors. We have reabsorbed this $(-1)^\ell$ by using $f_{34\cO}=(-1)^{\ell_\cO} f_{43\cO}$. A useful way to remember the correct sign is to note that $\<\phi_1(0)\f_2(z)|\cO|\f_2(1)\f_1(\oo)\>$ is the norm of a state in radial quantization, where $|\cO|$ is a projector onto the conformal multiplet of $\cO$. Thus, it should be positive, which implies that it should have coefficient $f_{12\cO}^2$ in the conformal block expansion.}
\be
& \<\f_1(x_1)\f_2(x_2)\f_3(x_3)\f_4(x_4)\> &= \frac{1}{x_{12}^{\De_1+\De_2} x_{34}^{\De_3+\De_4}}\frac{x_{14}^{\De_{34}} x_{23}^{\De_{12}}}{x_{13}^{\De_{12}+\De_{34}}} \sum_{\cO} f_{12\cO}f_{43\cO} G_{h_\cO,\bar h_\cO}^{h_{12},h_{34}}(z,\bar z),
\ee
where $h_{ij}\equiv h_i-h_j=\frac{\De_{ij}}{2}$. The coefficients $f_{ij\cO}$ are real in unitary theories.
Demanding symmetry under $1\leftrightarrow 3$ gives the crossing equation
\be
\label{eq:mixedcrossingequation}
  y^{-h_1-h_3} \sum_{\cO} f_{32\cO}f_{41\cO} G_{h_\cO,\bar h_\cO}^{h_{32},h_{14}}(z,1-\bar z)
 &=
 \bar y^{-h_1-h_3} \sum_{\cO} f_{12\cO}f_{43\cO} G_{h_\cO,\bar h_\cO}^{h_{12},h_{34}}(\bar z,1- z).
\ee

\subsubsection{Sums over $n$ and $\ell$}
\label{sec:sumovernandell}

The coefficients $S_a^{r,s}(\bar h)$ give a simple result when summed over a single family of $\SL(2,\R)$ blocks.  However, in $d$-dimensions, double-twist operators come in doubly-infinite families, labeled both by $\ell$ and $n$ such that $h\approx h_0+n$.  
The $d$-dimensional analog of $S_a^{r,s}(\bar h)$ will be coefficients $C^{(n)r,s}_a(h_0,\bar h)$ that, when summed over both $\ell$ and $n$, produce a simple result,
\be
\label{eq:magicidentityddims}
\sum_{n=0}^\oo \sum_{\ell=0}^\oo \pdr{\bar h}{\ell} C_a^{(n)r,s}(h_0,\bar h) G^{r,s}_{h_0+n,\bar h}(\bar z,1-z) &= \bar y^{h_0} y^a + [\dots]_{y}.
\ee

We can obtain the $C_a^{(n)r,s}(h_0,\bar h)$ by expanding $\SO(d,2)$ blocks in terms of $\SL(2,\R)$ blocks and using what we know about the coefficients $S_a^{r,s}(\bar h)$.
A simple example is in 2-dimensions, where $\SO(2,2)$ blocks are just products of $\SL(2,\R)$ blocks,\footnote{Here, we organize operators into irreps of $\SO(2)$, and not traceless symmetric tensors of $\SO(2)$. The latter convention would give an additional term $z\leftrightarrow \bar z$.} 
\be
G_{h,\bar h} (z,\bar z) &= k_{2h}(z) k_{2\bar h}(\bar z)\qquad (d=2),
\ee
(for simplicity we take $r=s=0$). Then we have
\be
C^{(n)}_a(h_0,\bar h) &= S_{-h_0}(h_0+n)S_a(\bar h) \qquad (d=2).
\ee
In general, $\SO(d,2)$ blocks have an expansion of the form\footnote{The 2d global conformal group $\SL(2,\R)_L\x\SL(2,\R)_R$ is a subgroup of $\SO(d,2)$.  The expansion (\ref{eq:expansionforblocks}) follows from decomposing an $\SO(d,2)$ multiplet into multiplets of $\R^* \x \SL(2,\R)_R$, where $\R^*$ is the Cartan of $\SL(2,\R)_L$.}
\be
\label{eq:expansionforblocks}
G^{r,s}_{h,\bar h}(z,\bar z) &= \sum_{n=0}^\oo \sum_{j=-n}^n A^{r,s}_{n,j}(h,\bar h) y^{h+n} k^{r,s}_{2(\bar h + j)}(\bar z).
\ee
The coefficients $A^{r,s}_{n,j}(h,\bar h)$ can be determined, for example, by solving the $\SO(d,2)$ Casimir equation order-by-order in $y$. Alternatively, we can obtain them from the decomposition of $d$-dimensional blocks into $2$-dimensional blocks~\cite{Hogervorst:2016hal}. The first few coefficients are
\be
A_{0,0}^{r,s}(h,\bar h) &= 1,\nn\\
A_{1,-1}^{r,s}(h,\bar h) &= \frac{\nu  ({\bar h}-h)}{{\bar h}-h+\nu -1},\nn\\
A_{1,0}^{r,s}(h,\bar h) &=\frac{s+r-h}{2}- \frac{r s ({\bar h}^2-{\bar h}-h \nu+\nu)}{2({\bar h}-1) {\bar h} (h-\nu)},\nn\\
A_{1,1}^{r,s}(h,\bar h) &= \frac{\nu  (h+{\bar h}-1) ({\bar h}-r) ({\bar h}+r) ({\bar h}-s) ({\bar h}+s)}{4 {\bar h}^2 (2 {\bar h}-1) (2 {\bar h}+1) (h+{\bar h}-\nu)},
\label{eq:somecoefficients}
\ee
where $\nu=\frac{d-2}{2}$.\footnote{Equations (\ref{eq:expansionforblocks}) and (\ref{eq:somecoefficients}) are subtle in even dimensions because the limit $\nu\to \frac{d-2}{2}$ does not commute with the limit $\bar h\to h+\ell$ when both $d/2$ and $\ell$ are integers. This is easily visible for the case $\nu=1$ and $\bar h - h=0$ in $A_{1,-1}^{r,s}$ in (\ref{eq:somecoefficients}). To get the correct block, one must take the limit $\nu\to \frac{d-2}{2}$ last. On the other hand, in even dimensions the blocks have simple analytic formulae~\cite{DO1,DO2}, and one can simplify the present analysis by using those specialized formulae.  For example, after multiplying the crossing equation in $4d$ by $\frac{z-\bar z}{z\bar z}$, one obtains products of $\SL(2,\R)$ blocks, and the analysis becomes similar to 2d.}

Since the leading $\bar y$-dependence of $G_{h_0,\bar h}^{r,s}(\bar z,1-z)$ is simply $\bar y^{h_0} k_{2\bar h}^{r,s}(1-z)$, if we take
\be
C_a^{(0)r,s}(h_0,\bar h) &= S_a^{r,s}(\bar h),
\ee
then the $\bar y^{h_0}$ terms on both sides of (\ref{eq:magicidentityddims}) will agree, by equation~(\ref{eq:mixedsl2sumgeneralization}). We can then choose the $n>0$ coefficients to cancel higher-order terms in $\bar y$. This gives a recursion relation
\be
\label{eq:recursionrelationforCs}
C_a^{(n)r,s}(h_0,\bar h) &= -\sum_{m=1}^n \sum_{j=-m}^m C^{(n-m)r,s}_a(h_0,\bar h - j) A_{m,j}^{r,s}(h_0+n-m,\bar h - j)\qquad(n>0),
\ee
that determines all the higher $C^{(n)}$'s.

As a cross-check, recall that $d$-dimensional MFT has conformal block expansion
\be
\sum_{n,\ell=0}^\oo C^{\mathrm{MFT}}_{n,\ell}(\De_1,\De_2) G^{\frac{\De_{12}}{2},\frac{\De_{21}}{2}}_{\frac{\De_1+\De_2}{2}+n,\frac{\De_1+\De_2}{2}+n+\ell}(z,\bar z)&= y^{\frac{\De_1+\De_2}{2}} \bar y^{\frac{\De_1+\De_2}{2}},
\ee
with coefficients given by~\cite{Fitzpatrick:2011dm}
\be
&C^{\mathrm{MFT}}_{n,\ell}(\De_1,\De_2) = \nn \\
&\frac{(\De_1-\nu)_n (\De_2-\nu)_n (\De_1)_{\ell+n} (\De_2)_{\ell+n}}{\ell! n! (\ell+\nu+1)_n (\De_1+\De_2+n-2 \nu-1)_n (\De_1+\De_2+2n+\ell-1)_\ell (\De_1+\De_2+n+\ell-\nu-1)_n}.
\ee
To be consistent with (\ref{eq:magicidentityddims}), we must have
\be
C_{-\frac{\De_1+\De_2}{2}}^{(n)\frac{\De_{12}}{2},\frac{\De_{21}}{2}}\p{\frac{\De_1+\De_2}{2},\frac{\De_1+\De_2}{2}+n+\ell} &= C^\mathrm{MFT}_{n,\ell}(\De_1,\De_2).
\ee
We have checked this explicitly for $n=0,1,2$.  Although $C_{n,\ell}^\textrm{MFT}(\De_1,\De_2)$ has a simple formula, we have not found a  closed-form expression for $C^{(n)r,s}_a(h_0,\bar h)$ in general dimensions.

\subsubsection{Small $\bar y$ expansion of the left-hand side}

On the left-hand side of the crossing equation, we should expand the blocks $G_{h,\bar h}(z,1-\bar z)$ in small $\bar y$.  As a starting point, the $\SL(2,\R)$ blocks have an expansion
\be
k_{2\bar h}^{r,s}(1-\bar z) &=  \sum_{k=0}^\oo \p{K^{r,s}_{k}(\bar h) \bar y^{k-r} + K^{s,r}_{k}(\bar h) \bar y^{k-s}},\\
K^{r,s}_k(\bar h) &\equiv  
\frac{\G(r-s)\G(1+s-r)}{\G(k+1)\G(k+1+s-r)}\frac{\G(2\bar h)}{\G(\bar h - r)\G(\bar h -s)}\frac{\G(\bar h + k-r)}{\G(\bar h - k+r)}.
\ee
Thus, we have
\be
G_{h,\bar h}^{r,s}(z, 1-\bar z) &= \sum_{m,k=0}^{\oo} y^{h+m}\p{P_{m,k}^{r,s}(h,\bar h) \bar y^{k-r} + P_{m,k}^{s,r}(h,\bar h) \bar y^{k-s}},
\label{eq:fullblockaroundone}\nn
\\
P_{m,k}^{r,s}(h,\bar h) &\equiv \sum_{j=-m}^m A_{m,j}^{r,s}(h,\bar h) K_k^{r,s}(\bar h + j).
\ee
In the special case $r=s$, this becomes
\be
G_{h,\bar h}^{r,r}(z, 1-\bar z) &= \sum_{m,k=0}^\oo y^{h+m} \pdr{}{k}(Q^{r}_{m,k}(h,\bar h) \bar y^{k-r}),\\
Q^r_{m,k}(h,\bar h) &\equiv \lim_{s\to r}\p{ (s-r) P_{m,k}^{r,s}(h,\bar h)}\nn\\
&=-\sum_{j=-m}^m A_{m,k}^{r,s}(h,\bar h) \frac{\Gamma (2 \bar h + 2j)}{\Gamma (k+1)^2 \Gamma (\bar h + j-r)^2}\frac{ \Gamma (\bar h + j+k-r)}{ \Gamma (\bar h + j-k+r)}.
\label{eq:defofQ}
\ee

\subsubsection{Matching the two sides}
\label{sec:matchingtosides}

Using (\ref{eq:fullblockaroundone}), the left-hand side of the crossing equation (\ref{eq:mixedcrossingequation}) is
\be
\label{eq:lhscrossingexpanded}
 &y^{-h_1-h_3}\sum_{\cO} f_{32\cO}f_{41\cO} G_{h_\cO,\bar h_\cO}^{h_{32},h_{14}}(z,1-\bar z)\nn\\
  &= \bar y^{-h_1-h_3} \nn\\
  & \quad \x \sum_\cO f_{32\cO}f_{41\cO} \sum_{m,k=0}^\oo y^{h_\cO+m-h_1-h_3} \p{P_{m,k}^{h_{32},h_{14}}(h_\cO,\bar h_\cO) \bar y^{k+h_1+h_2} + P^{h_{14},h_{32}}_{m,k}(h_\cO,\bar h_\cO) \bar y^{k+h_3+h_4}}.
\ee
Let us assume that the terms $\bar y^{k+h_1+h_2}$ match the families $[\f_1\f_2]_n$ with $n\leq k$ on the right-hand side, while $\bar y^{k+h_3+h_4}$ match $[\f_3\f_4]_n$ with $n\leq k$. (We return to this assumption in section~\ref{sec:twisthamiltonian}.) As before, define $\l_{ij[kl]_n}$ by
\be
\label{eq:jacobianstuff}
f_{ij[kl]_n}(\bar h) &= \l_{ij[kl]_n}(\bar h)\p{\pdr{\bar h}{\ell}}^{1/2} = \l_{ij[kl]_n}(\bar h) \p{ 1 - \pdr{\de_{[kl]_n}(\bar h)}{\bar h}}^{-1/2}.
\ee
Using (\ref{eq:magicidentityddims}) and working order-by-order in $\bar y$, we find 
\be
\label{eq:productsoflargespinope}
 \l_{12[12]_n}(\bar h) \l_{43[12]_n}(\bar h)
&\sim
\sideset{}{'}\sum_{\substack{\cO \in 2\x 3 \\ m\geq 0}} f_{32\cO}f_{41\cO} U^{(n)1234}_{\cO,m}(\bar h),
\ee
where
\be
U^{(n)1234}_{\cO,m}(\bar h) &\equiv \sum_{k=0}^n  P_{m,n-k}^{h_{32},h_{14}}(h_\cO,\bar h_\cO) C_{h_\cO+m-h_3-h_1}^{(k)h_{12},h_{34}}(h_1+h_2+n-k,\bar h).
\ee
The sum $\sum'_{\cO \in 2\x 3, m \geq 0}$ runs over operators $\cO$ in the $\f_2\x \f_3$ OPE and their descendants organized by weights under $\SL(2,\R)_L$.  The prime indicates that we must regularize the sum, as discussed above and demonstrated in sections~\ref{sec:asymptoticsappliedtoising} and~\ref{sec:twisthamiltonian}. 

By the same logic with $4\leftrightarrow 3$ swapped, we obtain
\be
 \l_{12[12]_n}(\bar h) \l_{34[12]_n}(\bar h) &= (-1)^\ell \l_{12[12]_n}(\bar h)\l_{43[12]_n}(\bar h)
 \sim  \sideset{}{'}\sum_{\substack{\cO \in 2\x 4 \\ m\geq 0}} f_{42\cO}f_{31\cO} U^{(n)1243}_{\cO,m}(\bar h),\label{eq:flip34}
\ee
where we used $\l_{ij\cO}=(-1)^{\ell_\cO} \l_{ji\cO}$.
Naively, equations (\ref{eq:productsoflargespinope}) and (\ref{eq:flip34}) seem to contradict each other. However, the meaning of (\ref{eq:productsoflargespinope}) and (\ref{eq:flip34}) is that the $\bar h$-dependence above reproduces the correct Casimir-singular terms on the other side of the crossing equations. We are free to add contributions that do not change the Casimir-singular part of the sum over blocks. As we learned in section~\ref{sec:alternating}, sums with a $(-1)^\ell$ insertion are Casimir-regular.  Thus, we can safely add the two contributions,
\be
\label{eq:addspinflip}
\l_{12[12]_n}(\bar h) \l_{43[12]_n}(\bar h) &\sim 
\sideset{}{'}\sum_{\substack{\cO \in 2\x 3 \\ m\geq 0}} f_{32\cO}f_{41\cO} U^{(n)1234}_{\cO,m}(\bar h)
+
(-1)^\ell (3\leftrightarrow 4)
,
\ee
and this single formula produces the correct Casimir-singular terms in both cases.\footnote{One can check that (\ref{eq:addspinflip}) is consistent with the symmetry $\l_{ij\cO}=(-1)^{\ell_\cO} \l_{ji\cO}$ for both $\l_{12[12]_n}$ and $\l_{43[12]_n}$.} The two terms in (\ref{eq:addspinflip}) are illustrated in figure~\ref{fig:generallargespindiag}.

\begin{figure}
\begin{center}
\be
\begin{array}{ccc}
\begin{tikzpicture}[xscale=0.6,yscale=0.6,baseline=(B.base)]
\draw[thick] (0,0) -- (0,6);
\draw[thick] (5,0) -- (5,6);
\draw[thick, dashed] (0,3) -- (5,3);
\node[below] at (0,0) {$\phi_1$};
\node[below] at (5,0) {$\phi_2$};
\node[above] at (0,6) {$\phi_4$};
\node[above] at (5,6) {$\phi_3$};
\node[above] at (2.5,3) {$\cO$};
\node[left] at (0,3) {$f_{41\cO}$};
\node[right] at (5,3) {$f_{32\cO}$};
\node (B) [] at (2.5,2.8) {};
\end{tikzpicture}
&
+
&
(-1)^\ell
\begin{tikzpicture}[xscale=0.6,yscale=0.6,baseline=(B.base)]
\draw[thick] (0,0) -- (0,6);
\draw[thick] (5,0) -- (5,6);
\draw[thick, dashed] (0,3) -- (5,3);
\node[below] at (0,0) {$\phi_1$};
\node[below] at (5,0) {$\phi_2$};
\node[above] at (0,6) {$\phi_3$};
\node[above] at (5,6) {$\phi_4$};
\node[above] at (2.5,3) {$\cO$};
\node[left] at (0,3) {$f_{31\cO}$};
\node[right] at (5,3) {$f_{42\cO}$};
\node (B) [left] at (2.5,2.8) {};
\end{tikzpicture}
\end{array}
\ee
\end{center}
\caption{
Large-spin diagrams for the contribution of $\cO$ to $\l_{12[12]_n}\l_{43[12]_n}$ in (\ref{eq:productsoflargespinope}).
}
\label{fig:generallargespindiag}
\end{figure}

In the special case $h_1+h_2=h_3+h_4$, (\ref{eq:lhscrossingexpanded}) develops $\log\bar y$-dependence (because $P_{m,k}^{r,s}$ has a pole at $r=s$), and we instead find a formula for products of OPE coefficients and anomalous dimensions,
\be
\label{eq:mixingpossible}
&\l_{12[12]_n}(\bar h) \l_{43[12]_n}(\bar h) \de_{[12]_n}(\bar h) + \l_{12[34]_n}(\bar h) \l_{43[34]_n}(\bar h) \de_{[34]_n}(\bar h) \nn\\
&\qquad\sim
\sideset{}{'}\sum_{\substack{\cO \in 2\x 3 \\ m\geq 0}}
f_{32\cO}f_{41\cO} V_{\cO,m}^{(n)1234}(\bar h)
+
(-1)^\ell (3\leftrightarrow 4)
,
\\
&\l_{12[12]_n}(\bar h) \l_{43[12]_n}(\bar h) + \l_{12[34]_n}(\bar h) \l_{43[34]_n}(\bar h)\nn\\
&\qquad\sim
\sideset{}{'}\sum_{\substack{\cO \in 2\x 3 \\ m\geq 0}}
f_{32\cO}f_{41\cO} W_{\cO,m}^{(n)1234}(\bar h)
+
(-1)^\ell (3\leftrightarrow 4),
\ee
where $V,W$ are defined by
\be
\label{eq:implicitdefofVW}
(V^{(n)1234}_{\cO,m}(\bar h)\log \bar y + W^{(n)1234}_{\cO,m}(\bar h))\bar y^{h_1+h_2} &\equiv
\lim_{h_3+h_4 \to h_1+h_2} U^{(n)1234}_{\cO,m}(\bar h) \bar y^{h_1+h_2}+U^{(n)3412}_{\cO,m}(\bar h) \bar y^{h_3+h_4}.
\ee
More explicitly, they are given by
\be
V^{(n)1234}_{\cO,m}(\bar h) &=
\sum_{k=0}^n Q_{m,n-k}^{h_{32}}(h_\cO,\bar h_\cO) C_{h_\cO+m-h_3-h_1}^{(k)h_{12},h_{34}}(h_1+h_2+n-k,\bar h),\\
W^{(n)1234}_{\cO,m}(\bar h) &= \sum_{k=0}^n \pdr{}{n}\p{Q_{m,n-k}^{h_{32}}(h_\cO,\bar h_\cO) C_{h_\cO+m-h_3-h_1}^{(k)h_{12},h_{34}}(h_1+h_2+n-k,\bar h)}.
\ee

Specializing further, we will need the case where the pairs of operators $\f_{1,2}$ and $\f_{3,4}$ are actually the same.  Since now only a single family $[12]_n$ reproduces $\bar y^{k+h_1+h_2}$ and $\bar y^{k+h_1+h_2}\log\bar y$ in (\ref{eq:lhscrossingexpanded}), we must drop the $[34]_n$ terms in (\ref{eq:mixingpossible}) before setting $12=43$. This gives
\be
\label{eq:anomdimwithminusoneell}
\l_{12[12]_n}(\bar h)^2\de_{[12]_n}(\bar h) &\sim 
\sideset{}{'}\sum_{\substack{\cO \in 1\x 1 \\ m\geq 0}}
f_{11\cO} f_{22\cO} V_{\cO,m}^{(n)1221}(\bar h)
+
(-1)^\ell\sideset{}{'}\sum_{\substack{\cO \in 1\x 2 \\ m\geq 0}}
(-1)^{\ell_\cO}f_{12\cO}^2 V_{\cO,m}^{(n)1212}(\bar h) 
,
\\
\label{eq:mixedopecorrection}
\l_{12[12]_n}(\bar h)^2 &\sim 
\sideset{}{'}\sum_{\substack{\cO \in 1\x 1 \\ m\geq 0}}
f_{11\cO} f_{22\cO} W_{\cO,m}^{(n)1221}(\bar h)
+
(-1)^\ell\sideset{}{'}\sum_{\substack{\cO \in 1\x 2 \\ m\geq 0}}
(-1)^{\ell_\cO}f_{12\cO}^2 W_{\cO,m}^{(n)1212}(\bar h) 
.
\ee
The identity operator is the leading contribution to (\ref{eq:mixedopecorrection}). Its coefficients are those of Mean Field Theory, analytically continued to $\ell=\bar h - h_1 -h_2 -n$,
\be
W_{\mathbf 1,m}^{(n)1221}(\bar h) &= \de_{m,0} C_{-h_1-h_2}^{(n)h_{12},h_{21}}(h_1+h_2,\bar h)
= \de_{m,0} C^\mathrm{MFT}_{n,\ell = \bar h - h_1-h_2-n}(2h_1,2h_2).
\ee

Finally, when all the operators are equal, (\ref{eq:anomdimwithminusoneell}) and (\ref{eq:mixedopecorrection}) become
\be
\label{eq:allequalanom}
\l_{11[11]_n}(\bar h)^2\de_{[11]_n}(\bar h) &\sim (1+(-1)^\ell)
\sideset{}{'}\sum_{\substack{\cO \in 1\x 1 \\ m\geq 0}}
f_{11\cO}^2 V_{\cO,m}^{(n)1111}(\bar h), \\
\l_{11[11]_n}(\bar h)^2&\sim (1+(-1)^\ell)
\sideset{}{'}\sum_{\substack{\cO \in 1\x 1 \\ m\geq 0}}
f_{11\cO}^2 W_{\cO,m}^{(n)1111}(\bar h).
\label{eq:allequalope}
\ee
We will often replace $1+(-1)^\ell\to 2$ and simply remember that only even-spin operators appear in the OPE $\f_1\x \f_1$.

\subsubsection{Checks}

Knowing CFT data up to weight $h_\mathrm{max}$ unambigiously determines the large-$\bar h$ corrections up to order $\bar h^{-2h_\mathrm{max}}$, or equivalently $J^{-\tau_\mathrm{max}}$.  To get this information, we could alternatively use the technology of~\cite{Alday:2015ewa}. It is straightforward to check that the first few $J^{-\tau_\cO}$ corrections to anomalous dimensions agree:
\be
2\frac{V^{(0)\f\f\f\f}_{\cO,0}(\bar h) + V^{(0)\f\f\f\f}_{\cO,1}(\bar h)}{C_{-2h_\f}^{(0)}(2h_\f,\bar h)}
&= 2\frac{Q_{0,0}(h_0,\bar h_0) S_{h_\cO-2h_\f}(\bar h) + Q_{1,0}(h_\cO,\bar h_\cO) S_{h_\cO+1-2h_\f}(\bar h)}{S_{-2h_\f}(\bar h)}\nn\\
&= \frac{c_0(\tau_\cO,\ell_\cO)}{J^{\tau_\cO}}\p{1+\frac{c_1(\tau_\cO,\ell_\cO)}{J^2} + \dots},
\ee
where $c_{0,1}(\tau_\cO,\ell_\cO)$ are the coefficients computed in~\cite{Alday:2015ewa} and given in equation~(\ref{eq:largespinexpansion}). (The factor of $2$ is because $\tau(\bar h)=2\De_\f + 2\de(\bar h)$.) The numerator above includes the contributions to anomalous dimensions from an operator $\cO$ and its descendants at level $1$ (\ref{eq:allequalanom}). The denominator includes the leading OPE coefficient coming from the unit operator. Additional terms in the denominator would give corrections of the form $J^{-\tau_1-\dots-\tau_n-k}$ not computed in~\cite{Alday:2015ewa}.

\subsubsection{Meaning of $\pdr{\bar h}{\ell}$}

Equation~(\ref{eq:anomdimwithminusoneell}) implies that the anomalous dimension $\de_{[12]_n}$ is not a smooth function of $\bar h$ alone, but also depends on $(-1)^\ell$.  Our proof of reparameterization invariance in section~\ref{sec:nonintegerspacing} does not apply to this case, but it can be fixed with a small modification. Suppose
\be
\bar h &= \bar h_0 + \ell + \de(\ell,\bar h),
\ee
where $\de(\ell, \bar h)$ has a large-$\bar h$ expansion that includes powers of $\bar h$ and factors of $(-1)^\ell$,
\be
\de(\ell,\bar h) &\sim \sum_{b_+} \bar h^{-b_+} + (-1)^\ell \sum_{b_-} \bar h^{-b_-}, \qquad \bar h\to \oo.
\ee
The proof in section~\ref{sec:nonintegerspacing} then works, provided we replace
\be
\label{eq:actualmeaningoflderiv}
\pdr{\bar h}{\ell} &\to \pdr{\bar h}{\bar h_0} = 1+\pdr{\de}{\bar h_0} = \p{1-\pdr{\de}{\bar h}}^{-1},
\ee
where in the derivative $\pdr{\de}{\bar h}$ we treat $(-1)^\ell$ as constant.

\section{Application to the 3d Ising CFT}
\label{sec:asymptoticsappliedtoising}

Let us now apply these results to the 3d Ising CFT\@. We would like to see how well the truncated large-$\bar h$ expansion describes the spectrum at finite $\bar h$. The more operators we can describe precisely, the better the prospects for hybrid analytical/numerical approaches like those discussed in section~\ref{sec:lessonsfornumerics}. We will find that a few terms in the expansion match numerics surprisingly well, even down to relatively small spins.

We will organize our expansions in terms of $S_a(\bar h)$'s. This simplifies several computations (in particular it makes it simpler to compute Casimir-regular terms). However, one could just as well use powers of the $\SL(2,\R)$ Casimir $J^2$, as in~\cite{Alday:2015eya,Alday:2015ota,Alday:2015ewa,Alday:2016mxe,Alday:2016njk,Alday:2016jfr}. A sum of $S_a(\bar h)$'s is a partial resummation of a series in $J^2$. 

We will work our way upwards in twist, first understanding $[\s\s]_0$ in section~\ref{sec:sigsig0section}, then $[\s\e]_0$ in section~\ref{sec:sigepszero}, and finally $[\s\s]_1$ and $[\e\e]_0$ in section~\ref{sec:eezerosigsigone}.  Because $2h_\s$ is so small, the family $[\s\s]_0$ is particularly important. Its contribution to other large-$\bar h$ expansions is competitive with those $T_{\mu\nu}$ and $\e$. Thus, we will use our formulae for OPE coefficients and dimensions of $[\s\s]_0$ in several subsequent computations. We expect this approach should also work well for the $O(N)$ models. It is an interesting question whether it works in a general CFT.

\subsection{$[\s\s]_0$}
\label{sec:sigsig0section}

\begin{figure}
\begin{center}
\begin{tikzpicture}[xscale=0.6,yscale=0.6]
\draw[thick] (0,0) -- (0,6);
\draw[thick] (5,0) -- (5,6);
\draw[thick, dashed] (0,3) -- (5,3);
\node[below] at (0,0) {$\s$};
\node[below] at (5,0) {$\s$};
\node[above] at (0,6) {$\s$};
\node[above] at (5,6) {$\s$};
\node[above] at (2.5,3) {$\e,T$};
\end{tikzpicture}
\end{center}
\caption{
The contributions of $\e,T$ to $\l_{\s\s[\s\s]_0}$ and $\de_{[\s\s]_0}$ in (\ref{eq:sigsigfit0}) and (\ref{eq:sigsigfit1}).
}
\label{fig:etcontributiontosigsigsigsig}
\end{figure}

The OPE coefficients and anomalous dimensions of $[\s\s]_0$ fit nicely to the first few terms in (\ref{eq:allequalanom}), (\ref{eq:allequalope}), illustrated in figure~\ref{fig:etcontributiontosigsigsigsig},
\be
\label{eq:sigsigfit0}
\l_{\s\s[\s\s]_0}^2 &\approx 2\p{S_{-2h_\s}(\bar h) + f_{\s\s\e}^2 W^{(0)\s\s\s\s}_{\e,0}(\bar h) + f_{\s\s T}^2 W^{(0)\s\s\s\s}_{T,0}(\bar h)},\\
\label{eq:sigsigfit1}
\l_{\s\s[\s\s]_0}^2 \de_{[\s\s]_0} &\approx 2\p{f_{\s\s\e}^2 V^{(0)\s\s\s\s}_{\e,0}(\bar h) + f_{\s\s T}^2 V^{(0)\s\s\s\s}_{T,0}(\bar h)},
\ee
where
\be
V^{(0)\s\s\s\s}_{\cO,0}(\bar h) &= -\frac{\G(2\bar h_\cO)}{\G(\bar h_\cO)^2}S_{h_\cO-2h_\s}(\bar h),\nn\\
W^{(0)\s\s\s\s}_{\cO,0}(\bar h) &= -\frac{\G(2\bar h_\cO)}{\G(\bar h_\cO)^2}\p{2\psi(\bar h_\cO) -2 \psi(1)}S_{h_\cO-2h_\s}(\bar h),
\ee
and 
\be
\label{eq:actualvalues}
\De_\s = 2h_\s &\approx 0.5181489,\nn\\
\De_\e = 2h_\e &\approx 1.412625,\nn\\
f_{\s\s\e} &\approx 1.0518539,\nn\\
f_{\s\s T} = \sqrt{\frac{3}{8 c_T}} \De_\s &\approx 0.326138.
\ee
In other words, we have
\be
\de_{[\s\s]_0} &\approx \frac{2\p{f_{\s\s\e}^2 V^{(0)\s\s\s\s}_{\e,0}(\bar h) + f_{\s\s T}^2 V^{(0)\s\s\s\s}_{T,0}(\bar h)}}{2\p{S_{-2h_\s}(\bar h) + f_{\s\s\e}^2 W^{(0)\s\s\s\s}_{\e,0}(\bar h) + f_{\s\s T}^2 W^{(0)\s\s\s\s}_{T,0}(\bar h)}}
\label{eq:deltasszeroprecition}
\\
f_{\s\s[\s\s]_0}^2 &\approx \p{1-\pdr{\de_{[\s\s]_0}(\bar h)}{\bar h}}^{-1} \l_{\s\s[\s\s]_0}(\bar h)^2,
\label{eq:fsssszeroprediction}
\ee
where we used equation (\ref{eq:actualmeaningoflderiv}) for the Jacobian $\pdr{\bar h}{\bar \ell}$ that relates $f_{\s\s[\s\s]_0}$ to $\l_{\s\s[\s\s]_0}$. The actual operator dimensions are determined by solving $\bar h-2 h_\s - \de(\bar h) = 0, 2, 4,\dots$.

A comparison between the above formula and numerics for $\tau_{[\s\s]_0}=2\De_\s + 2\de_{[\s\s]_0}$ is shown in figure~\ref{fig:tauSigSig0}. The discrepancy between analytics and numerics is $3\x 10^{-3}$ and $5\x 10^{-4}$ for spins $\ell=2,4$, respectively, and $\sim 5\x 10^{-5}$ for $\ell> 4$. Including additional higher-twist operators (primaries or descendants) in (\ref{eq:sigsigfit0}) and (\ref{eq:sigsigfit1}) does not improve the fit for low spins, and barely affects it for high spins.

\begin{figure}[ht!]
\begin{center}
\includegraphics[width=0.9\textwidth]{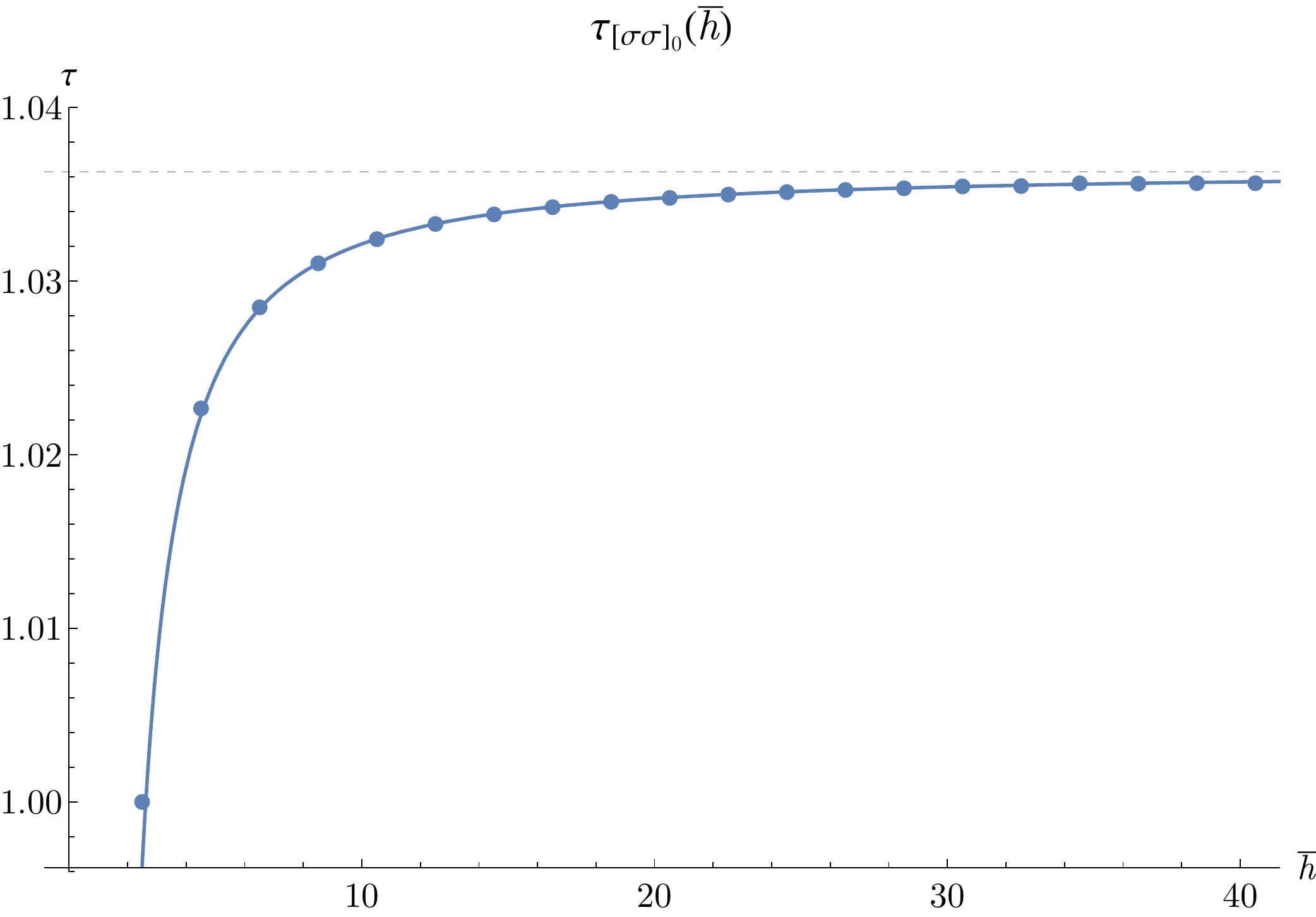}
\end{center}
\caption{A comparison between the analytical prediction (\ref{eq:deltasszeroprecition}) (blue curve) and numerical data (blue dots) for $\tau_{[\s\s]_0}$. The two agree with accuracy $3\x 10^{-3}$ and $5\x 10^{-4}$ for spins $\ell=2,4$, respectively, and $\sim 5\x 10^{-5}$ for $\ell> 4$. The grey dashed line is the asymptotic value $\tau=2\De_\s$. The curve (\ref{eq:largespinexpansion}) from~\cite{Alday:2015ewa} looks essentially the same.}
\label{fig:tauSigSig0}
\end{figure}

\subsubsection{Differences from~\cite{Alday:2015ewa}}

Let us comment briefly on the (inconsequential) differences between the above calculation and the series (\ref{eq:largespinexpansion}) computed in~\cite{Alday:2015ewa}. Firstly, we have not included descendants of $\e,T$, namely terms of the form $W_{\cO,m}^{(0)\s\s\s\s}$ and $V_{\cO,m}^{(0)\s\s\s\s}$ with $m \geq 1$, whereas~\cite{Alday:2015ewa} included descendants at first order in $z$. This is because it doesn't make sense to include level-1 descendants of $\e,T$ without also including the double-twist operators $[\e T]_0$, $[TT]_0$, which contribute at the same order in the large-$\bar h$ expansion. Also, because we organize everything as a series in $y$ instead of $z$, the contributions of descendants will differ somewhat (though the sum over all of them will be the same). In addition, we have partially resummed the $J$ series into sums of $S_a(\bar h)$'s.

All these alternatives represent different choices of subleading terms in a series that we are truncating anyway. Fortunately, they turn out to be inconsequential at the truncation order and precision at which we are working. A plot of (\ref{eq:largespinexpansion}) looks essentially identical to figure~\ref{fig:tauSigSig0}.  However, $S_a(\bar h)$'s will begin to differ from powers of $J$ when $a=h_\cO+m-2h_\s$ is larger (i.e.\ for higher-twist primaries and descendants in the crossed-channel). This is because $S_a(\bar h)$ has poles at $\bar h = a+1,a,a-1,\dots$, whereas $J^{-2a}$ does not. (In Sturm-Liouville theory for $\SL(2,\R)$ blocks~\cite{MatthijsBaltFuture}, these poles come from the region near $y\sim 1$, outside the validity of the small-$y$ expansion. Thus, they are artifacts of our expansion in small-$y$ in the crossed-channel.) These differences reflect the fact that we are comparing different truncations of an asymptotic expansion outside the regime of validity of those truncations.

\subsubsection{Contributions of $[\s\s]_0$ to itself}
\label{sec:contributiontoself}

We should also include higher-spin members of the family $[\s\s]_0$ in (\ref{eq:sigsigfit0}), (\ref{eq:sigsigfit1}). 
Their contributions for $\ell=4,6,\dots$ are small because
\be
W^{(0)\s\s\s\s}_{\cO,0}(\bar h), V^{(0)\s\s\s\s}_{\cO,0}(\bar h) &\propto \frac{1}{\G(h_\cO-2h_\s)^2} \sim \de_\cO^2,
\ee
where $\de_\cO=h_\cO-2h_\s$ is half the anomalous dimension of $\cO$, and $\de_\cO$ decreases with $\ell$.  
Nevertheless, we can sum the whole family $[\s\s]_0$ by by expanding in the anomalous dimension $\de_{[\s\s]_0}$ and using the methods we have developed for summing $\SL(2,\R)$ blocks.\footnote{An alternative approach to computing corrections to anomalous dimensions from an infinite family of operators is given in~\cite{Alday:2016njk,Alday:2016jfr}.}

Using (\ref{eq:mynonintegerspaced}), we have
\be
& \sum_{\ell=\ell_0,\ell_0+2,\dots} f_{\s\s[\s\s]_0}^2 y^{h_{[\s\s]_0}-2h_\s} k_{2\bar h}(1-\bar z) \nn\\
&= \sum_{m=0}^\oo \log^m y \sum_{k=0}^\oo \pdr{}{k}\p{\bar y^k \a_k^\mathrm{even}\left[\l_{\s\s[\s\s]_0}^2 \frac{\de_{[\s\s]_0}^m}{m!}, \de_{[\s\s]_0} \right](2h_\s+\ell_0)} + \textrm{casimir-singular}.\nn\\
\ee
The terms with $k=0$ contribute to $\l_{\s\s[\s\s]_0}$ and $\de_{[\s\s]_0}$ as follows
\be
\l_{\s\s[\s\s]_0}(\bar h)^2 &\sim \textrm{above} + 2\sum_{m=2}^\oo  \b_0^\mathrm{even}\left[\l_{\s\s[\s\s]_0}^2 \frac{\de_{[\s\s]_0}^m}{m!}, \de_{[\s\s]_0} \right](2h_\s+\ell_0) \left.\pdr{{}^m S_a(\bar h)}{a^m}\right|_{a=0}\nn\\
& \approx \textrm{above} -0.000572238  \left.\pdr{{}^2 S_a(\bar h)}{a^2}\right|_{a=0} +8.92146\.10^{-7}\left.\pdr{{}^3 S_a(\bar h)}{a^3}\right|_{a=0}+\dots\nn\\
\l_{\s\s[\s\s]_0}(\bar h)^2\de_{[\s\s]_0}(\bar h) &\sim \textrm{above} + 2\sum_{m=2}^\oo  \a_0^\mathrm{even}\left[\l_{\s\s[\s\s]_0}^2 \frac{\de_{[\s\s]_0}^m}{m!}, \de_{[\s\s]_0} \right](2h_\s+\ell_0) \left.\pdr{{}^m S_a(\bar h)}{a^m}\right|_{a=0},\nn\\
& \approx \textrm{above} -0.000123342 \left.\pdr{{}^2 S_a(\bar h)}{a^2}\right|_{a=0} + 2.1276\.10^{-7}\left.\pdr{{}^3 S_a(\bar h)}{a^3}\right|_{a=0}+\dots,
\label{eq:includinghigherspincorrections}
\ee
where ``above" represents terms already present in (\ref{eq:sigsigfit0}) and (\ref{eq:sigsigfit1}), and $\b_k^\mathrm{even} = \pdr{}{k} \a_k^\mathrm{even}$. The sums start at $m=2$ because $S_a(h)$ has a second-order zero at $a=0$. (Equivalently, the terms proportional to $\log^m y$ are Casimir-regular in the other channel when $m=0,1$.)   We illustrate the contributions (\ref{eq:includinghigherspincorrections}) in figure~\ref{fig:higherspincontribtosigsig}.

Equation~(\ref{eq:includinghigherspincorrections}) might look complicated because $\l_{\s\s[\s\s]_0}$ and $\de_{[\s\s]_0}$ are defined in terms of themselves. However, $\de_{[\s\s]_0}$ is small, so (\ref{eq:includinghigherspincorrections}) is easily solved by iteration starting with the approximations (\ref{eq:sigsigfit0}), (\ref{eq:sigsigfit1}). Above, we show the result from plugging in (\ref{eq:sigsigfit0}), (\ref{eq:sigsigfit1}) and setting $\ell_0=4$. 
The corrections in (\ref{eq:includinghigherspincorrections}) are so small that we mostly omit them in what follows. By contrast, similar corrections for $[\s\e]_0$ begin at $m=1$, and for $[\e\e]_0$ they begin at $m=0$.  In these cases, one must sum the whole family $[\s\s]_0$ to get accurate results.

\begin{figure}
\begin{center}
\be
\begin{array}{ccccccc}
\begin{tikzpicture}[xscale=0.6,yscale=0.6,baseline=(B.base)]
\draw[thick] (0,0) -- (0,1.5);
\draw[thick, dashed] (0,1.5) -- (0,5.5);
\draw[thick] (0,5.5) -- (0,7);
\draw[thick] (5,0) -- (5,1.5);
\draw[thick, dashed] (5,1.5) -- (5,5.5);
\draw[thick] (5,5.5) -- (5,7);
\draw[thick] (0,1.5) -- (5,1.5);
\draw[thick] (0,5.5) -- (5,5.5);
\node[below] at (0,0) {$\s$};
\node[below] at (5,0) {$\s$};
\node[above] at (0,7) {$\s$};
\node[above] at (5,7) {$\s$};
\node[below] at (2.5,1.5) {$\s$};
\node[above] at (2.5,5.5) {$\s$};
\node[left] at (0,3.5) {$\e,T$};
\node[right] at (5,3.5) {$\e,T$};
\node(B) [left] at (2.8,3.4) {};
\end{tikzpicture}
&
+
&
\cdots
&
+
&
\begin{tikzpicture}[xscale=0.6,yscale=0.6,baseline=(B.base)]
\draw[thick] (0,0) -- (0,1.5);
\draw[thick, dashed] (0,1.5) -- (0,5.5);
\draw[thick] (0,5.5) -- (0,7);
\draw[thick] (5,0) -- (5,1.5);
\draw[thick, dashed] (5,1.5) -- (5,5.5);
\draw[thick, dashed] (1.2,1.5) -- (1.2,5.5);
\draw[thick, dashed] (2.4,1.5) -- (2.4,5.5);
\draw[thick] (5,5.5) -- (5,7);
\draw[thick] (0,1.5) -- (5,1.5);
\draw[thick] (0,5.5) -- (5,5.5);
\node[below] at (0,0) {$\s$};
\node[below] at (5,0) {$\s$};
\node[above] at (0,7) {$\s$};
\node[above] at (5,7) {$\s$};
\node[below] at (2.5,1.5) {$\s$};
\node[above] at (2.5,5.5) {$\s$};
\node[left] at (0,3.5) {$\e,T$};
\node[] at (3.8,3.5) {$\dots$};
\node[right] at (5,3.5) {$\e,T$};
\node(B) [left] at (2.8,3.4) {};
\end{tikzpicture}
&
+
&
\cdots
\end{array}
\ee
\end{center}
\caption{
Contributions to $\l_{\s\s[\s\s]_0}$ and $\de_{[\s\s]_0}$ (bottom-to-top channel) from the exchange of double-twist operators $[\s\s]_0$ (left-to-right channel). We can further expand the contribution of the family $[\s\s]_0$ in small $\de_{[\s\s]_0}$.  We illustrate the $m$-th order term in this expansion by adding $m$ vertical exchanges between $\s$ lines, coming from the operators that contribute to $\de_{[\s\s]_0}$ ($\e$ and $T$ in our approximations (\ref{eq:sigsigfit0}) and (\ref{eq:sigsigfit1})). The leading nonzero term has $m=2$, corresponding to two vertical lines, or a ``box diagram."
}
\label{fig:higherspincontribtosigsig}
\end{figure}

We compare analytics and numerics for $f_{\s\s[\s\s]_0}$ in figure~\ref{fig:fSigSigSigSig0}. There is an interesting wrinkle in interpreting the numerics.  Although the numerical spectra include operators $\cO_\ell$ with twists $\tau_{[\s\s]_0}$, they also sometimes include spurious higher-spin currents $J_\ell$ at the unitarity bound with small but nonzero OPE coefficients. Because $\tau_{[\s\s]_0}$ is close to the unitarity bound, these spurious operators can ``fake" the contribution of $\cO_\ell$ in the conformal block expansion.\footnote{Higher spin currents are disallowed in interacting CFTs~\cite{Maldacena:2011jn,Alba:2013yda,Boulanger:2013zza,Alba:2015upa}.} The $J_\ell$ are artifacts of the extremal functional method. They should disappear at sufficiently high derivatives, but working at higher derivatives is not currently feasible. Instead, we remove them by hand and add their OPE coefficients to the correct operators $\cO_\ell$.  In other words, we use $(f_{\s\s\cO_\ell}^2 + f_{\s\s J_\ell}^2)^{1/2}$ as our numerical prediction for $f_{\s\s[\s\s]_0}$.  Indeed, the numerical errors in in this modified quantity are smaller than the errors in $f_{\s\s J_\ell}$, and the results agree beautifully with the analytical prediction. We show numerical data both before and after the modification in figure~\ref{fig:fSigSigSigSig0}.

\begin{figure}[ht!]
\begin{center}
\includegraphics[width=0.9\textwidth]{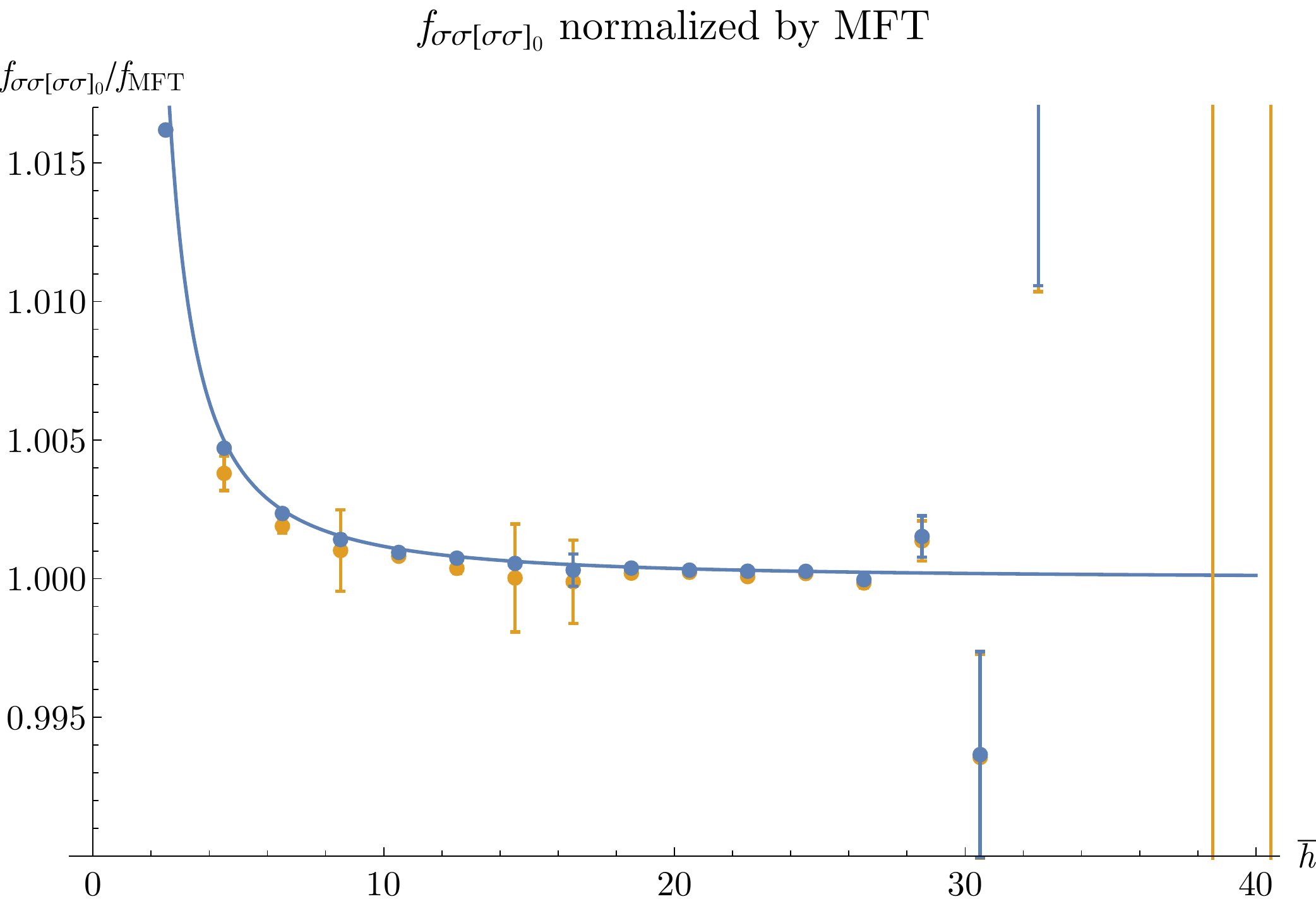}
\end{center}
\caption{A comparison between the analytical prediction (\ref{eq:fsssszeroprediction}) and numerics for $f_{\s\s[\s\s]_0}$, both normalized by dividing by the Mean Field Theory OPE coefficients $f_\mathrm{MFT}=(2S_{-2h_\s}(\bar h))^{1/2}$.  We show two sets of numerical data.  The orange series gives the OPE coefficients of the operators $\cO_\ell$ with twist closest to $\tau_{[\s\s]_0}$ for each spin $\ell$. 
 The blue series combines the contributions of $\cO_\ell$ and spurious higher-spin currents $J_\ell$ into $(f_{\s\s\cO_\ell}^2 + f_{\s\s J_\ell}^2)^{1/2}$. The latter quantities have smaller errors and better match the analytical prediction. The fact that the errors shrink after this modification supports the idea that the correct OPE coefficient is being shared between the real operators $\cO_\ell$ and ``fake" operators $J_\ell$.
 }
\label{fig:fSigSigSigSig0}
\end{figure}

The leading contribution to the OPE coefficients $\l_{\e\e[\s\s]_0}$ comes from $\s$-exchange in the $\s\e\to\s\e$ channel,
\be
\l_{\s\s[\s\s]_0} \l_{\e\e[\s\s]_0} &\approx 2f_{\s\s\e}^2 U^{(0)\s\s\e\e}_{\s,0}(\bar h).
\label{eq:feess0fit}
\ee
This agrees with numerics within $1\%$ for all spins $\ell \geq 2$. In the next section, we compute additional corrections from the family $[\s\e]_0$ and improve the agreement.

\subsection{$[\s\e]_0$}
\label{sec:sigepszero}

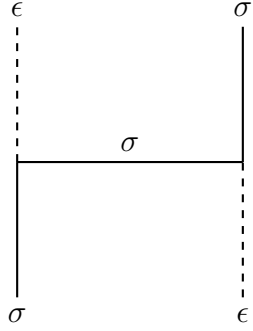
\begin{figure}
\begin{center}
\begin{tikzpicture}[xscale=0.6,yscale=0.6]
\draw[thick] (0,0) -- (0,3);
\draw[thick,dashed] (0,3) -- (0,6);
\draw[thick,dashed] (5,0) -- (5,3);
\draw[thick] (5,3) -- (5,6);
\draw[thick] (0,3) -- (5,3);
\node[below] at (0,0) {$\s$};
\node[below] at (5,0) {$\e$};
\node[above] at (0,6) {$\e$};
\node[above] at (5,6) {$\s$};
\node[above] at (2.5,3) {$\s$};
\end{tikzpicture}
\end{center}
\caption{
Contribution of $\s$-exchange (left-to-right) to the $[\s\e]_0$ family (bottom-to-top). We get a factor of $(-1)^\ell$ in (\ref{eq:firstapproximationfse}) and (\ref{eq:firstapproximationdeltase}) because $\s$ and $\e$ switch places.
}
\label{fig:sigcontribtosigeps}
\end{figure}

The leading correction to OPE coefficients and anomalous dimensions of $[\s\e]_0$ comes from exchange of $\s$ in the $\s\e\to \e\s$ channel (figure~\ref{fig:sigcontribtosigeps}),
\be
\l_{\s\e[\s\e]_0}^2 
&\approx S_{-h_\s-h_\e}^{h_{\s\e},h_{\e\s}}(\bar h)
+ (-1)^\ell f_{\s\s\e}^2 W_{\s,0}^{(0)\s\e\s\e}(\bar h) + \dots
\label{eq:firstapproximationfse}
\\
\l_{\s\e[\s\e]_0}^2 \de_{[\s\e]_0}
&\approx (-1)^\ell f_{\s\s\e}^2 V_{\s,0}^{(0)\s\e\s\e}(\bar h) + \dots.
\label{eq:firstapproximationdeltase}
\ee
To go further, we must include the contribution of the family $[\s\s]_0$ in $\s\s\to \e\e$. Doing so will provide a nontrivial test of the tools we have developed.

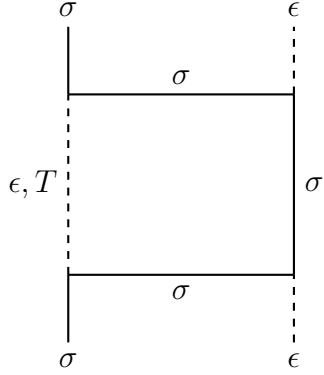
\begin{figure}
\begin{center}
\begin{tikzpicture}[xscale=0.6,yscale=0.6]
\draw[thick] (0,0) -- (0,1.5);
\draw[thick,dashed] (0,1.5) -- (0,5.5);
\draw[thick] (0,5.5) -- (0,7);
\draw[thick,dashed] (5,0) -- (5,1.5);
\draw[thick] (5,1.5) -- (5,5.5);
\draw[thick,dashed] (5,5.5) -- (5,7);
\draw[thick] (0,1.5) -- (5,1.5);
\draw[thick] (0,5.5) -- (5,5.5);
\node[below] at (0,0) {$\s$};
\node[below] at (5,0) {$\e$};
\node[above] at (0,7) {$\s$};
\node[above] at (5,7) {$\e$};
\node[below] at (2.5,1.5) {$\s$};
\node[above] at (2.5,5.5) {$\s$};
\node[right] at (5,3.5) {$\s$};
\node[left] at (0,3.5) {$\e,T$};
\end{tikzpicture}
\end{center}
\caption{
Contribution of $[\s\s]_0$-exchange (left-to-right) to the $[\s\e]_0$ family (bottom-to-top). In general, any operators can appear in the internal legs of the box diagram. Here we highlight the contributions computed below.
}
\label{fig:sigsigcontribtosigeps}
\end{figure}

Because we will discuss both channels simultaneously, let us write the crossing equation in a way that emphasizes the important terms:
\be
\sum_{\cO\in [\s\e]_0} f_{\s\e\cO}^2 y^{h_\cO-h_\s-h_\e} k^{h_{\s\e},h_{\e\s}}_{2\bar h}(1-\bar z)+\dots
&=
\sum_{\cO\in [\s\s]_0} f_{\s\s\cO}f_{\e\e\cO} \bar y^{h_\cO-h_\s-h_\e} k_{2\bar h}(1-z) + \dots.
\label{eq:referencecrossingequation}
\ee
Our first goal is to compute the sum over $[\s\s]_0$ on the right-hand side,
\be
\label{eq:exampleformforsigsigsum}
&\sum_{\cO\in [\s\s]_0} f_{\s\s\cO}f_{\e\e\cO} \bar y^{h_\cO-h_\s-h_\e} k_{2\bar h}(1-z)\nn\\
 &= \textrm{$y$-Casimir-singular} + \a(\bar y) \log y + \b(\bar y) + O(y).
\ee
The terms $\a(\bar y)\log y$ and $\b(\bar y)$ have the correct form to contribute to anomalous dimensions and OPE coefficients of $[\s\e]_0$ on the left-hand side of (\ref{eq:referencecrossingequation}). However, the Casimir-singular terms do not, and must be cancelled in some other way.  We work through an explicit example in section~\ref{sec:cancellationofcasimirsing}.

Before performing the sum over $[\s\s]_0$, let us understand what part we will need. Consider $\cO=[\s\s]_{0,\ell}$ on the right-hand side of (\ref{eq:referencecrossingequation}), and suppose $\ell$ is large so that $\de_{[\s\s]_{0,\ell}}=h_\cO-2h_\s$ is small.
The $\bar y$-dependence of the $\cO$-block maps to the following $\bar h$-dependence of $\l_{\s\e[\s\e]_0}$ on the left-hand side:
\be
\label{eq:mixedinversegamma}
\bar y^{h_\s-h_\e + \de_{[\s\s]_0}} = \sum_{k=0}^\oo \frac{\de^k_{[\s\s]_{0}}}{k!} \bar y^{h_{\s\e}} \log^k \bar y &\to \sum_{k=0}^\oo \frac{\de_{[\s\s]_0}^k}{k!} \left.\pdr{{}^k}{a^k} S^{h_{\s\e},h_{\e\s}}_{h_{\s\e}+a}(\bar h)\right|_{a=0}.
\ee
The $k=0$ term vanishes because $S^{r,s}_{r+a}$ has a simple zero at $a=0$.
The first nontrivial correction has $k=1$  (figure~\ref{fig:sigsigcontribtosigeps}). Thus, the leading correction to $\l_{\s\e[\s\e]_0}$ and $\de_{[\s\e]_0}$ in the sum over $[\s\s]_0$ comes from expanding to {\it first\/}-order in the anomalous dimension $\de_{[\s\s]_0}$:
\be
\bar y^{h_\s-h_\e} \log \bar y \sum_{\substack{\bar h = \bar h_0 + \ell \\ \ell=0,2,\dots}} \l_{\s\s[\s\s]_0}(\bar h) \l_{\e\e[\s\s]_0}(\bar h) \de_{[\s\s]_0}(\bar h) k_{2\bar h}(1-z) + \dots.
\ee
Here, ``$\dots$" represents non-$\log \bar y$ terms that do not contribute to $\l_{\s\e[\s\e]_0}$ and $\de_{[\s\e]_0}$.
We will treat $T_{\mu\nu}$ separately, so the family $[\s\s]_0$ starts at $\bar h_0=2h_\s+\ell_0$ with $\ell_0=4$.

The quantities $\l_{\s\s[\s\s]_0}$, $\l_{\e\e[\s\s]_0}$, and $\de_{[\s\s]_0}$ can be obtained from (\ref{eq:sigsigfit0}), (\ref{eq:sigsigfit1}), and (\ref{eq:feess0fit}).  For simplicity, we approximate their product by the first two leading terms at large $\bar h$, coming from the corrections to $\de_{[\s\s]_0}$ due to $\e$ and $T$,
\be
&\l_{\s\s[\s\s]_0} \l_{\e\e[\s\s]_0} \de_{[\s\s]_0}\nn\\
&\approx
\sum_{\cO=T,\e}-2 f_{\s\s\cO}^2 f_{\s\s\e}^2 \frac{ \Gamma (2 \bar h_\cO) \Gamma (2 h_\s)^3 \Gamma (h_\e-h_\cO)^2 \Gamma (2 h_\e-2 h_\s)}{\Gamma (\bar h_\cO)^2 \Gamma (h_\e)^4 \Gamma (2 h_\s-h_\cO)^2}S_{h_\cO-h_\e}(\bar h).
\label{eq:fsssseessdess}
\ee
This approximation has the correct asymptotics and also matches numerics within $1\%$ for all $\ell \geq 4$. This is sufficient accuracy for our purposes, since we are already computing a small correction to $[\s\e]_0$.

For the $\cO=T$ term, we have
\be
\label{eq:easyTsum}
\sum_{\substack{\bar h = \bar h_0 + \ell \\ \ell=0,2,\dots}} S_{h_T - h_\e}(\bar h) k_{2\bar h}(1-z)
&=
\frac 1 2 y^{h_T-h_\e} + \a^\mathrm{even}_0[S_{h_T-h_\e}](\bar h_0) \log y + \b^\mathrm{even}_0[S_{h_T-h_\e}](\bar h_0) + O(y),
\ee
where
\be
\a^\mathrm{even}_0[S_{a}](\bar h_0) &= \cA^\mathrm{even}_{a,-1}(\bar h_0)
\label{eq:Tcontribtose}
= -\frac{\Gamma (\bar h_0-a-1)}{2 a \Gamma (-a)^2 \Gamma (\bar h_0+a-1)},\\
\b^\mathrm{even}_0[S_{a}](\bar h_0) &= \left.\pdr{}{k}\cA^\mathrm{even}_{a,k}(\bar h_0)\right|_{k=-1}.
\ee
Equation~(\ref{eq:easyTsum}) has the form anticipated in (\ref{eq:exampleformforsigsigsum}).
As we prove in appendix~\ref{app:boxdiagrams}, the Casimir-singular term $y^{h_T-h_\e}$ is cancelled by the exchange of $[T\s]_0$ in the $\s\e\to \s\e$ OPE\@. The remaining terms give nontrivial contributions to $\l_{\s\e[\s\e]_0}$ and $\de_{[\s\e]_0}$.
We have not found an analytic formula for $\b^\mathrm{even}_0[S_a](\bar h_0)$ in general, but it can be computed to arbitrary accuracy using (\ref{eq:TTsumidentity}) and (\ref{eq:convergentexprforaminus}).

The $\cO=\e$ term in (\ref{eq:fsssseessdess}) takes more care to evaluate.  Taking $h_\cO\to h_\e$ gives
\be
\label{eq:extractfunnylimit}
-2 f_{\s\s\e}^4 \frac{ \Gamma (2 h_\e) \Gamma (2 h_\s)^3 \Gamma (2 h_\e-2 h_\s)}{ \Gamma (h_\e)^6 \Gamma (2 h_\s-h_\e)^2}\lim_{a\to 0} \G(-a)^2S_{a}(\bar h).
\ee
The function $\lim_{a\to 0} \G(-a)^2S_{a}(\bar h)$ is finite, but when we insert it in a sum over blocks, both the Casimir-singular and Casimir-regular terms are naively infinite. However, $1/a^2$ and $1/a$ poles cancel between them, leaving a finite result:
\be
&\lim_{a\to 0}\sum_{\substack{\bar h = \bar h_0+\ell \\ \ell = 0,2,\dots}} \G(-a)^2 S_a(\bar h) k_{2\bar h}(1-z)\nn\\
&=
\lim_{a\to 0} \p{\frac 1 2 \G(-a)^2 y^a + \G(-a)^2 \a_0^\mathrm{even}[S_a](\bar h_0)\log y + \G(-a)^2 \b_0^\mathrm{even}[S_a](\bar h_0)} + O(y)\nn\\
&= \frac 1 4 \log^2 y + A_0(\bar h_0) \log y + B_0(\bar h_0) + O(y),
\label{eq:funnyatozeroresult}
\ee
where
\be
A_0(\bar h_0) &\equiv \lim_{a\to 0} \p{\G(-a)^2 \a_0^\mathrm{even}[S_a](\bar h_0) + \frac 1 {2a} + \g} \nn\\
&= \psi(\bar h_0 - 1) + \g,\\
B_0(\bar h_0) &\equiv \lim_{a\to 0}\p{ \G(-a)^2 \b_0^\mathrm{even}[S_a](\bar h_0) + \frac{1}{2 a^2}+\frac{\gamma }{a}+\gamma ^2+\frac{\pi ^2}{12}}\nn\\
&= \frac{\pi ^2}{12} +(\psi(\bar h_0)+\gamma)\p{\psi(\bar h_0)+\g-\frac{2}{\bar h_0-1}}+ \frac{1}{4}\p{\psi^{(1)}\p{\frac{\bar h_0}{2}}-\psi^{(1)}\p{\frac{\bar h_0+1}{2}}}.
\label{eq:b0formula}
\ee
Here, $\psi^{(m)}(z) \equiv \rdr{{}^{m+1}}{z^{m+1}} \log \G(z)$ is the polygamma function, $\psi(z)=\psi^{(0)}(z)$, and $\g=-\psi(1)$ is the EulerÐ-Mascheroni constant. (Even though we do not have a simple formula for $\b_0^\mathrm{even}[S_a](\bar h_0)$ in general, the limit $B_0(\bar h_0)$ is computable in closed form and given by (\ref{eq:b0formula}).)

\subsubsection{Cancellation of Casimir-singular terms}
\label{sec:cancellationofcasimirsing}

Equation~(\ref{eq:funnyatozeroresult}) again has the form anticipated in (\ref{eq:exampleformforsigsigsum}), where
the $\log^2 y$ term in (\ref{eq:funnyatozeroresult}) is $y$-Casimir-singular. Combining (\ref{eq:extractfunnylimit}) and (\ref{eq:funnyatozeroresult}), this term is
\be
&\bar y^{-h_\s - h_\e} \sum_{\bar h} f_{\s\s[\s\s]_0}f_{\e\e[\s\s]_0}\de_{[\s\s]_0} \bar y^{2h_\s} \log \bar y\, k_{2\bar h}(1-z)\nn\\
&\sim
 -\frac 1 2 f_{\s\s\e}^4 \frac{ \Gamma (2 h_\e) \Gamma (2 h_\s)^3 \Gamma (2 h_\e-2 h_\s)}{ \Gamma (h_\e)^6 \Gamma (2 h_\s-h_\e)^2} \log^2 y\, \bar y^{h_\s-h_\e}\, \log \bar y
\label{eq:contributionssee}
\ee
in the $\s\s\to \e\e$ channel.

We claimed earlier that the Casimir-singular terms in (\ref{eq:exampleformforsigsigsum}) should be canceled by other contributions, and it is instructive to see how this works explicitly. The expression~(\ref{eq:contributionssee}) has the correct form to match the exchange of $[\s\e]$ in the $\s\e\to \s\e$ channel, where $\log^2 y$ comes from expanding $y^{\de_{[\s\e]_0}}$ to second order in $\de_{[\s\e]_0}$. We could have guessed this by reinterpreting figure~\ref{fig:sigsigcontribtosigeps} as the second order term in the exponentiation of figure~\ref{fig:sigcontribtosigeps} (in the bottom-to-top channel). 

The important terms in $\de_{[\s\e]_0}^2$ come from squaring the contribution of $\s$-exchange. 
From (\ref{eq:firstapproximationfse}) and (\ref{eq:firstapproximationdeltase}), we have
\be
\l_{\s\e[\s\e]_0}^2 \frac 1 2 \de_{[\s\e]_0}^2 &\sim \frac 1 2\frac{\p{f_{\s\s\e}^2 V_{\s,0}^{(0)\s\e\s\e}(\bar h)}^2}{S_{-h_\s-h_\e}^{h_{\s\e},h_{\e\s}}(\bar h)}\\
&= \frac 1 2 f_{\s\s\e}^4 \frac{\G(2h_\e)\G(2h_\s)^3\G(2h_\e-2h_\s)}{\G(h_\e)^6 \G(2h_\s-h_\e)^2} \p{\lim_{a\to 0} \G(-a) S^{h_{\s\e},h_{\e\s}}_{h_\s-h_\e+a}(\bar h)} + \dots.\nn\\
\ee
Using (\ref{eq:mixedblockremarkable}), the relevant sum over blocks is
\be
\sum_{h=h_0+\ell} \lim_{a\to 0} \G(-a) S^{r,s}_{a-s}(h)k_{2h}(1-\bar z)
&=
-\bar y^{-s} \log \bar y + [\dots]_{\bar y},\nn\\
[\dots]_{\bar y} &= -\bar y^{-s}(\psi(h_0-s)+\psi(h_0+s-1)-\psi(s-r)+\g)\nn\\
& \quad + \bar y^{-r} (\dots) + O(\bar y^{1-s},\bar y^{1-r}).
\label{eq:otherexamplewithpolecancelling}
\ee
(We have written the $\bar y^{-s}$ part of the Casimir-regular terms because we will need them shortly.)
Again, $1/a$ poles cancel between the Casimir-regular and Casimir-singular part, leaving a finite result.
It follows that 
\be
&\sum_{\bar h=\bar h_0+\ell} \l_{\s\e[\s\e]_0}^2 \frac 1 2 \de_{[\s\e]_0}^2 \log^2 y\, k_{2\bar h}^{h_{\s\e},h_{\e\s}}(1-\bar z)\nn\\
&=
-\frac 1 2 f_{\s\s\e}^4 \frac{\G(2h_\e)\G(2h_\s)^3\G(2h_\e-2h_\s)}{\G(h_\e)^6 \G(2h_\s-h_\e)^2} \log^2 y\,\bar y^{h_\s-h_\e} \log \bar y + [\dots]_{\bar y},
\label{eq:matchcassingse}
\ee
which exactly matches (\ref{eq:contributionssee}).

Thus, the other channel indeed cancels the Casimir-singular term in (\ref{eq:funnyatozeroresult}). This phenomenon, which has been explored previously in~\cite{Alday:2015ewa,Alday:2016mxe}, is a special case of a more general result. The $y$-Casimir-singular part of the exchange of double-twist operators in one channel matches the $\bar y$-Casimir-singular part of the exchange of double-twist operators in the other channel. Another way to say this is that box diagrams like figure~\ref{fig:sigcontribtosigeps} give the same Casimir-singular parts when interpreted from bottom-to-top or from left-to-right.\footnote{However, their Casimir-regular parts are not necessarily the same.} We prove this claim in appendix~\ref{app:boxdiagrams}.\footnote{We conjecture that it should be possible to prove a much more general result: that the Casimir-singular terms in a general large-spin diagram, given by an arbitrary network of operator exchanges, are crossing-symmetric.}

The Casimir-regular term proportional to $\bar y^{h_\s-h_\e}$ in (\ref{eq:matchcassingse}) determines the leading correction to $f_{\e\e[\s\s]_0}$ coming from $[\s\e]_0$ exchange. Including also level-one descendants of $\s$, which contribute at similar order in the $1/\bar h$-expansion to $[\s\e]_0$, we have
\be
&\l_{\s\s[\s\s]_0} \l_{\e\e[\s\s]_0}\nn\\
&\approx 2f_{\s\s\e}^2 (U^{(0)\s\s\e\e}_{\s,0}(\bar h)+U^{(0)\s\s\e\e}_{\s,1}(\bar h))\nn\\
&\quad- f_{\s\s\e}^4 \frac{\G(2h_\e)\G(2h_\s)^3\G(2h_\e-2h_\s)}{\G(h_\e)^6 \G(2h_\s-h_\e)^2} \p{\psi(2h_\s+\ell_0)+\psi(2h_\e+\ell_0-1)-\psi(2h_\e-2h_\s)+\g} \nn\\
&\qquad \x \left.\pdr{{}^2}{a^2} S_a(\bar h)\right|_{a=0},
\label{eq:betterfitforfeess0}
\ee
where $\ell_0=2$ is the lowest spin appearing in the $[\s\e]_0$ family. As we show in figure~\ref{fig:fEpsEpsSigSig0}, (\ref{eq:betterfitforfeess0}) agrees with numerics for all spins with accuracy $\sim 10^{-3}$.

\begin{figure}[ht!]
\begin{center}
\includegraphics[width=0.9\textwidth]{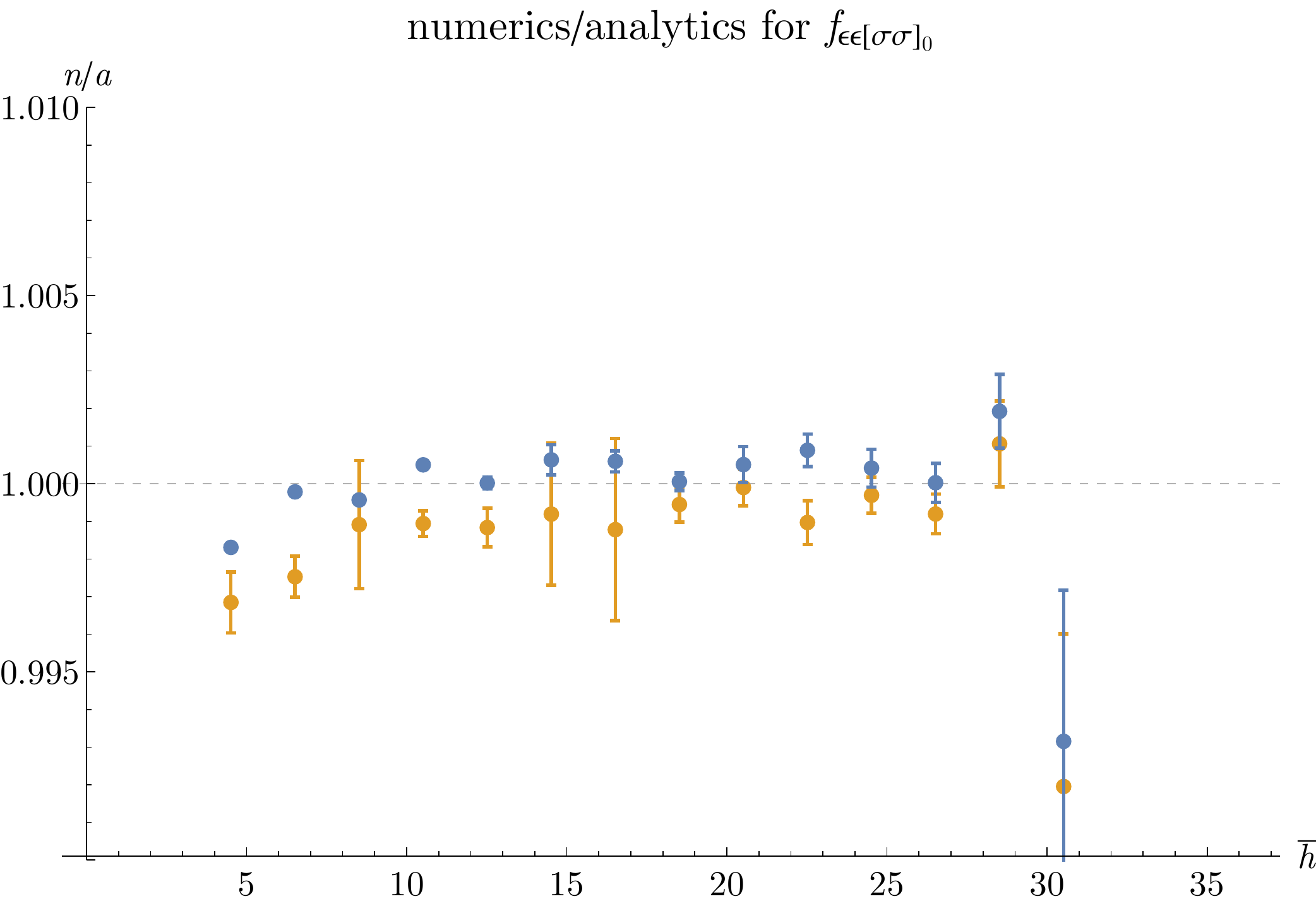}
\end{center}
\caption{Ratios $n/a$ of numerical results to the analytical prediction (\ref{eq:sigsigfit0}, \ref{eq:betterfitforfeess0}) for $f_{\e\e[\s\s]_0}$. (One must multiply by the Jacobian $\pdr{\bar h}{\ell}$ to relate $f_{\e\e[\s\s]_0}$ to $\l_{\e\e[\s\s]_0}$.) As in figure~\ref{fig:fSigSigSigSig0}, we show two sets of numerical data. The orange series are the raw OPE coefficients $f_{\e\e\cO_\ell}$ of operators with twists $\tau_{[\s\s]_0}$. The blue series are the improved coefficients $(f_{\e\e\cO_\ell}^2 + f_{\e\e J_\ell}^2)^{1/2}$ discussed in section~\ref{sec:sigsig0section}.}
\label{fig:fEpsEpsSigSig0}
\end{figure}

\subsubsection{Putting everything together}

Combining the Casimir-regular terms from (\ref{eq:Tcontribtose}) and (\ref{eq:funnyatozeroresult}), we have
\be
&\bar y^{h_\s-h_\e} \log \bar y \sum_{\substack{\bar h = 2h_\s + 4+ \ell \\ \ell=0,2,\dots}} \l_{\s\s[\s\s]_0}(\bar h) \l_{\e\e[\s\s]_0}(\bar h) \de_{[\s\s]_0}(\bar h) k_{2\bar h}(1-z)\nn\\
&\approx \bar y^{h_\s-h_\e} \log \bar y \left(-2 f_{\s\s T}^2 f_{\s\s\e}^2 \frac{ \Gamma (2 \bar h_T) \Gamma (2 h_\s)^3 \Gamma (h_\e-h_T)^2 \Gamma (2 h_\e-2 h_\s)}{\Gamma (\bar h_T)^2 \Gamma (h_\e)^4 \Gamma (2 h_\s-h_T)^2}\right.\nn\\
& \left.\qquad\qquad\qquad\qquad \x\p{\a^\mathrm{even}_0[S_{h_T-h_\e}](2h_\s+4) \log y + \b^\mathrm{even}_0[S_{h_T-h_\e}](2h_\s+4)}\right.\nn\\
&\qquad\qquad\qquad\qquad\left. -2 f_{\s\s\e}^4 \frac{ \Gamma (2 h_\e) \Gamma (2 h_\s)^3 \Gamma (2 h_\e-2 h_\s)}{ \Gamma (h_\e)^6 \Gamma (2 h_\s-h_\e)^2} \p{A_0(2h_\s+4) \log y + B_0(2h_\s+4)}\right)\nn\\
& + \textrm{Casimir-singular}+ O(y).
\ee
From the above, we can read off the contributions to $\l_{\s\e[\s\e]_0}$ and $\de_{[\s\e]_0}$ from exchange of the family $[\s\s]_0$. Including also the corrections from exchange of $\e$ and $T_{\mu\nu}$, we have
\be
\l_{\s\e[\s\e]_0}^2 &\approx S_{-h_\s-h_\e}^{h_{\s\e},h_{\e\s}}(\bar h) + (-1)^\ell f_{\s\s\e}^2 W_{\s,0}^{(0)\s\e\s\e}(\bar h)\nn\\
&\quad + f_{\s\s\e} f_{\e\e\e} W^{(0)\s\e\e\s}_{\e,0}(\bar h) + f_{\s\s T} f_{\e\e T} W_{T,0}^{(0)\s\e\e\s}(\bar h)\nn\\
&\quad + \left(
-2 f_{\s\s T}^2 f_{\s\s\e}^2 \frac{ \Gamma (2 \bar h_T) \Gamma (2 h_\s)^3 \Gamma (h_\e-h_T)^2 \Gamma (2 h_\e-2 h_\s)}{\Gamma (\bar h_T)^2 \Gamma (h_\e)^4 \Gamma (2 h_\s-h_T)^2}\b^\mathrm{even}_0[S_{h_T-h_\e}](2h_\s+4)\right.\nn\\
&\quad\quad\quad \left. -2 f_{\s\s\e}^4 \frac{ \Gamma (2 h_\e) \Gamma (2 h_\s)^3 \Gamma (2 h_\e-2 h_\s)}{ \Gamma (h_\e)^6 \Gamma (2 h_\s-h_\e)^2} B_0(2h_\s+4)\right) \left.\pdr{}{a} S^{h_{\s\e},h_{\e\s}}_{h_\s-h_\e + a}(\bar h)\right|_{a=0},
\label{eq:secondapproximationfse}
\\
\l_{\s\e[\s\e]_0}^2 \de_{[\s\e]_0} &\approx (-1)^\ell f_{\s\s\e}^2 V_{\s,0}^{(0)\s\e\s\e}(\bar h)\nn\\
&\quad + f_{\s\s\e} f_{\e\e\e} V^{(0)\s\e\e\s}_{\e,0}(\bar h) + f_{\s\s T} f_{\e\e T} V_{T,0}^{(0)\s\e\e\s}(\bar h)\nn\\
&\quad + \left(
-2 f_{\s\s T}^2 f_{\s\s\e}^2 \frac{ \Gamma (2 \bar h_T) \Gamma (2 h_\s)^3 \Gamma (h_\e-h_T)^2 \Gamma (2 h_\e-2 h_\s)}{\Gamma (\bar h_T)^2 \Gamma (h_\e)^4 \Gamma (2 h_\s-h_T)^2}\a^\mathrm{even}_0[S_{h_T-h_\e}](2h_\s+4)\right.\nn\\
&\quad\quad\quad \left. -2 f_{\s\s\e}^4 \frac{ \Gamma (2 h_\e) \Gamma (2 h_\s)^3 \Gamma (2 h_\e-2 h_\s)}{ \Gamma (h_\e)^6 \Gamma (2 h_\s-h_\e)^2} A_0(2h_\s+4)\right) \left.\pdr{}{a} S^{h_{\s\e},h_{\e\s}}_{h_\s-h_\e + a}(\bar h)\right|_{a=0}.
\label{eq:secondapproximationdeltase}
\ee

\subsubsection{Comparison to numerics}

We plot the twists $\tau_{[\s\e]_0}=\De_\s+\De_\e + 2\de_{[\s\e]_0}$ in figure~\ref{fig:tauSigEps0} and OPE coefficients $f_{\s\e[\s\e]_0}$ in figure~\ref{fig:fSigEps0}, comparing the formulae (\ref{eq:secondapproximationfse}) and (\ref{eq:secondapproximationdeltase}) to numerical results. In both cases, analytics matches numerics to high precision $(\sim10^{-4})$ at large $\bar h$, and moderate precision ($<10^{-2}$) for all $\bar h$. The agreement is particularly impressive because the corrections are large compared to Mean Field Theory, in contrast to the case of $[\s\s]_0$. Correctly summing the family $[\s\s]_0$ is crucial for achieving this.

\begin{figure}[ht!]
\begin{center}
\includegraphics[width=0.9\textwidth]{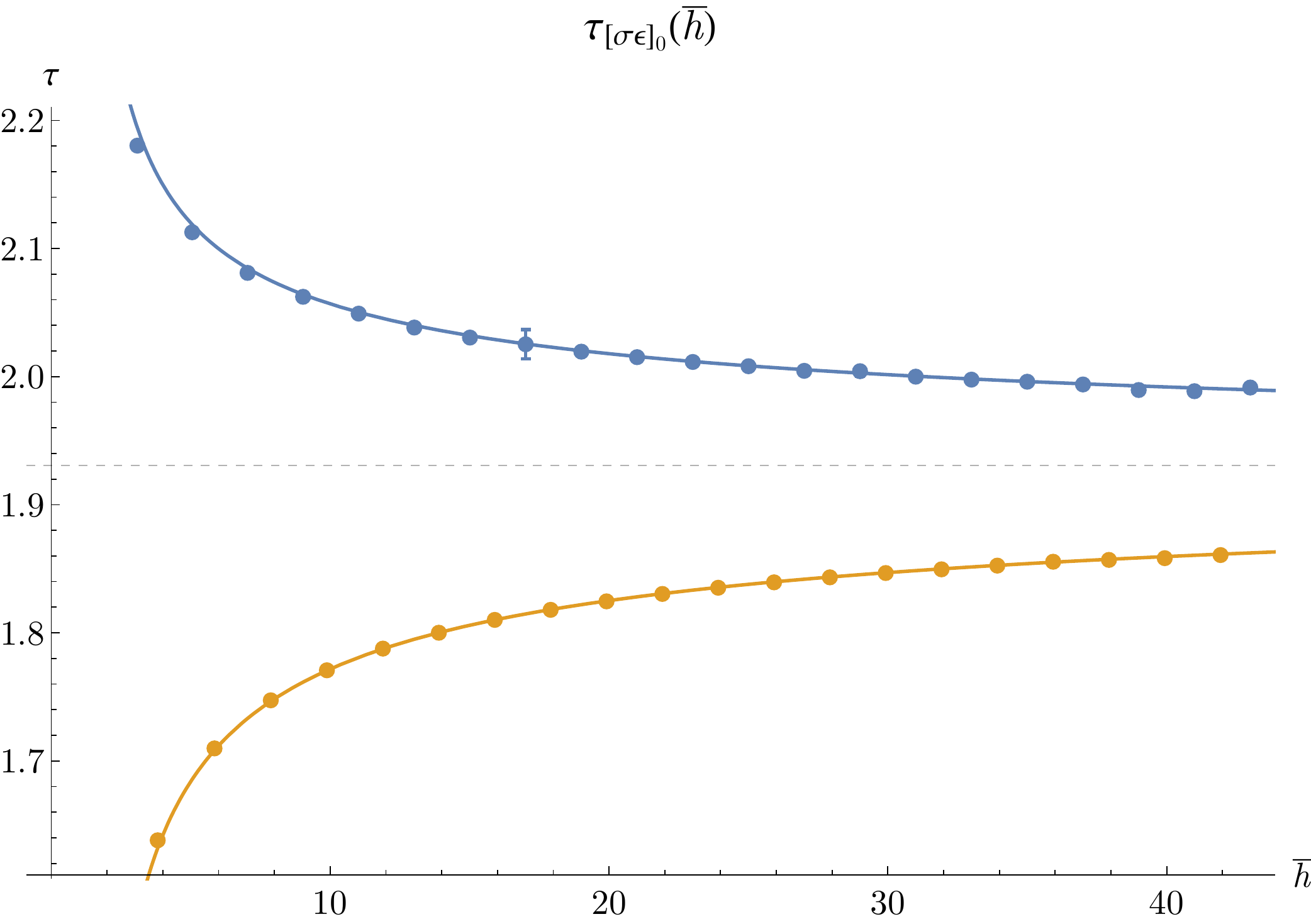}
\end{center}
\caption{Comparison between numerical data and the analytical prediction (\ref{eq:secondapproximationfse}, \ref{eq:secondapproximationdeltase}) for $\tau_{[\s\e]_0}$. The blue curve and points correspond to even-spin operators and the orange curve and points correspond to odd-spin operators. The dashed line is the asymptotic value $\tau=\De_\s+\De_\e$.}
\label{fig:tauSigEps0}
\end{figure}

\begin{figure}[ht!]
\begin{center}
\includegraphics[width=0.9\textwidth]{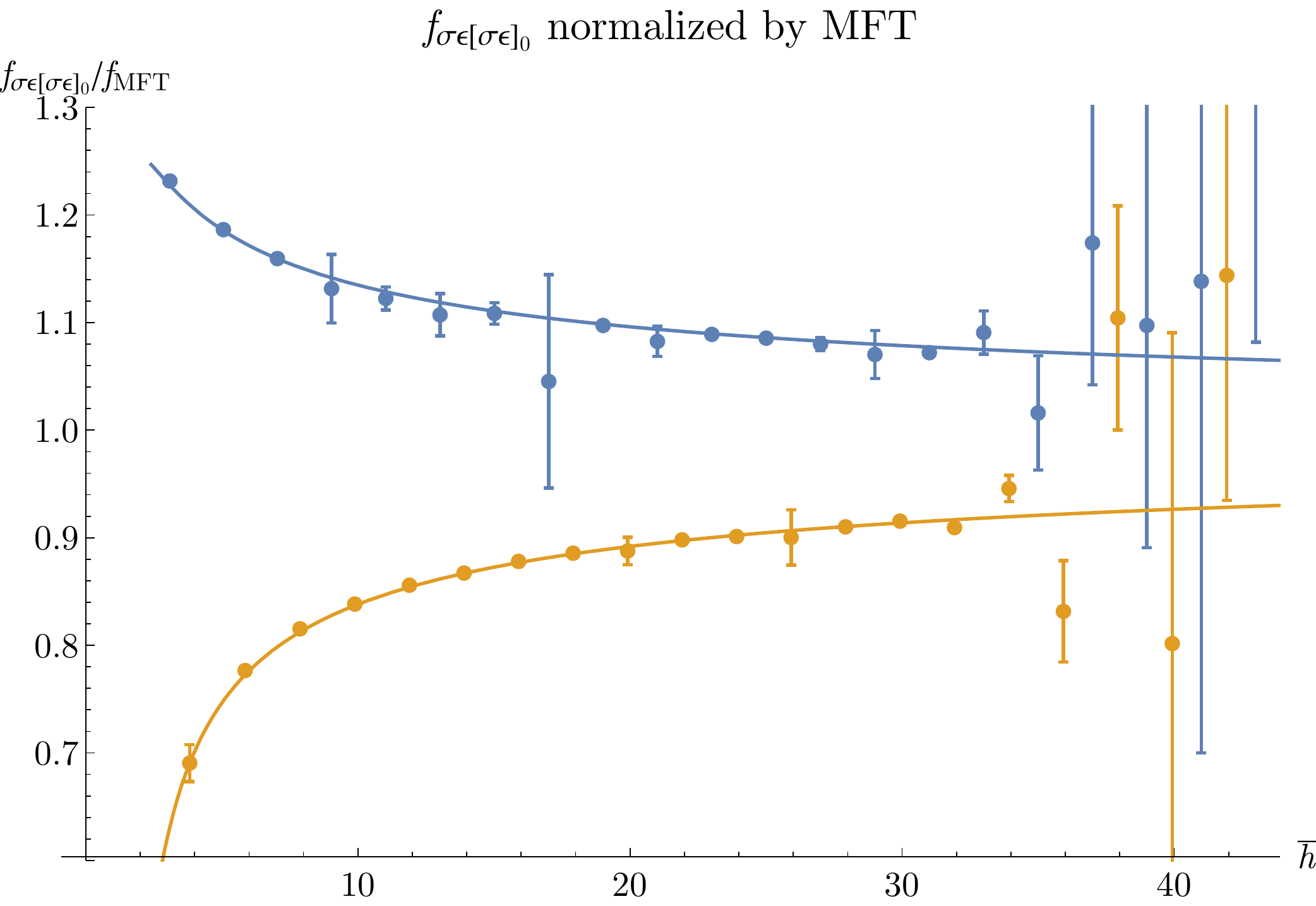}
\end{center}
\caption{Comparison between numerical data and the analytical prediction (\ref{eq:secondapproximationfse}) for $f_{\s\e[\s\e]_0}$, both divided by the Mean Field Theory OPE coefficients $f_\mathrm{MFT}=S^{h_{\s\e},h_{\e\s}}_{-h_\s-h_\e}(\bar h)^{1/2}$. The blue curve and points correspond to even-spin operators and the orange curve and points correspond to odd-spin operators.}
\label{fig:fSigEps0}
\end{figure}

\section{Operator mixing and the twist Hamiltonian}
\label{sec:twisthamiltonian}

\subsection{Allowing for mixing}
\label{sec:allowmixing}

The naive large-$\bar h$ expansion of section~\ref{sec:allorderslightcone} describes the operators $[\s\s]_0$ and $[\s\e]_0$ nicely. However, it fails badly for $[\s\s]_1$ and $[\e\e]_0$. As mentioned in the introduction, the numerics indicate large mixing between these families.  As a striking illustration,  we plot the ratios $f_{\e\e[\e\e]_0}/f_\mathrm{MFT}$ and $f_{\e\e[\s\s]_1}/f_\mathrm{MFT}$ in figure~\ref{fig:fEpsEpsSigSig1AndfEpsEpsEpsEps0}. (We define $[\e\e]_0$ as the operator with lower twist.) For spins $\ell\lesssim 20$, the coefficient $f_{\e\e[\s\s]_1}$ is actually {\it larger\/} than $f_{\e\e[\e\e]_0}$.

One might guess that the asymptotic large-$\bar h$ expansion simply breaks down earlier for these operators --- that it just doesn't work for $\ell\lesssim 40$.  This turns out to be false. In this section, we give a procedure that extends the validity of the large-$\bar h$ expansion down to smaller values of $\bar h$.

The key idea is to relax the assumption from section~\ref{sec:generaldtfamilies} that the double-twist operators $[ij]_n$ on one side of the crossing equation map only to terms of the form $\bar y^{h_i + h_j+k}$ on the other side. Instead, we will compute a fully $\bar y$-dependent asymptotic expansion in $\bar h$ and identify operators by diagonalizing an effective ``twist Hamiltonian."

Let
\be
H(\bar h) &= \begin{pmatrix}
h_{[\s\s]_0}(\bar h) & 0 & 0 \\
0 & h_{[\s\s]_1}(\bar h) & 0 \\
0 & 0 & h_{[\e\e]_0}(\bar h) \\
\end{pmatrix},\\
\L(\bar h) &= \begin{pmatrix}
\l_{\s\s[\s\s]_0}(\bar h) & \l_{\s\s[\s\s]_1}(\bar h) & \l_{\s\s[\e\e]_0}(\bar h) \\
\l_{\e\e[\s\s]_0}(\bar h) & \l_{\e\e[\s\s]_1}(\bar h) & \l_{\e\e[\e\e]_0}(\bar h) 
\end{pmatrix}.
\ee
Suppose that, using crossing symmetry, we can find the combination
\be
\label{eq:formofm}
 \L(\bar h)\bar y^{H(\bar h)} \L(\bar h)^T 
&= \sum_{\cO=[\s\s]_0,[\s\s]_1,[\e\e]_0} \begin{pmatrix}
\l_{\s\s\cO}(\bar h)^2  & \l_{\s\s\cO}(\bar h) \l_{\e\e\cO}(\bar h)  \\
\l_{\s\s\cO}(\bar h) \l_{\e\e\cO}(\bar h)  & \l_{\e\e\cO}(\bar h)^2 
\end{pmatrix}\bar y^{h_\cO(\bar h)}\nn\\
& \equiv \begin{pmatrix}
M_{\s\s\s\s}(\bar y, \bar h) & M_{\s\s\e\e}(\bar y, \bar h) \\
M_{\s\s\e\e}(\bar y, \bar h) & M_{\e\e\e\e}(\bar y, \bar h)
\end{pmatrix}.
\ee
One way to extract the twist Hamiltonian is as follows.
Given the elements $M_{ijkl}(\bar y, \bar h)$, we form the matrix
\be
M(\bar y, \bar h) &\equiv \begin{pmatrix}
M_{\s\s\s\s}(\bar y, \bar h) & \bar\ptl M_{\s\s\s\s}(\bar y, \bar h) & M_{\s\s\e\e}(\bar y, \bar h) \\
\bar \ptl M_{\s\s\s\s}(\bar y, \bar h) & \bar\ptl^2 M_{\s\s\s\s}(\bar y, \bar h) & \bar\ptl M_{\s\s\e\e}(\bar y, \bar h)\\
M_{\s\s\e\e}(\bar y, \bar h) & \bar \ptl M_{\s\s\e\e}(\bar y, \bar h) & M_{\e\e\e\e}(\bar y, \bar h)
\end{pmatrix},
\ee
where for brevity, we've defined 
\be
\bar\ptl \equiv \pdr{}{\log \bar y}.
\ee
The twist Hamiltonian $H(\bar h)$ is given by diagonalizing
\be
M(\bar y, \bar h)^{-1} \bar\ptl M(\bar y, \bar h).
\label{eq:twisthamiltonian}
\ee

If $M(\bar y, \bar h)$ indeed has the form (\ref{eq:formofm}), with only the twist families $[\s\s]_0$, $[\s\s]_1$, and $[\e\e]_0$ contributing,  then the combination (\ref{eq:twisthamiltonian}) will be $\bar y$-independent. In practice, we cannot completely single out $[\s\s]_0$, $[\s\s]_1$, and $[\e\e]_0$ on the other side of the crossing equation, so our $M(\bar y, \bar h)$  will have corrections from other operators in the $\s \x \s$ OPE, and we must choose a value $\bar y= \bar y_0$ at which to evaluate it.

The families $[\s\s]_n$ and $[\e\e]_n$ with higher $n$ will be exponentially suppressed if we choose a small value of $\bar y_0$. However, to single out $[\e\e]_0$ and $[\s\s]_1$ we must also assume that other twist families like $[TT]$, $[TTT]$, and $[\s\s\e]$, which contribute at similar order in $\bar y$, have small OPE coefficients in the $\s\x\s$ and $\e\x\e$ OPEs. This assumption is supported by numerics (which likely means that it follows from unitarity). However, we do not know how to derive it using the information in this work. Instead, we should enlarge our system of crossing equations to include additional external operators.
For example, by studying the matrix
\be
M &= 
\begin{array}{l}
\quad\,\overbrace{\hspace{1.42in}}^{\mbox{$N_{ij}$}}\,\,\,\, \overbrace{\hspace{1.42in}}^{\mbox{$\phantom{N_{ij}}N_{kl}\phantom{N_{ij}}$}}
\\
\left(
\begin{array}{ccc|ccc|c}
\ddots & \cdots & \cdots & \cdots & \cdots & \cdots & \cdots \\
\vdots & \bar\ptl^{m+p} M_{ijij} & \cdots & \cdots & \bar\ptl^{m+q} M_{ijkl} & \cdots & \cdots \\
\vdots & \vdots & \ddots & \cdots & \cdots & \cdots & \cdots\\
\hline
\vdots & \vdots & \vdots & \ddots & \cdots & \cdots & \cdots\\
\vdots & \bar\ptl^{n+p} M_{klij} & \vdots & \vdots & \bar\ptl^{n+q} M_{klkl} & \cdots & \cdots \\
\vdots & \vdots & \vdots & \vdots & \vdots & \ddots & \cdots \\
\hline
\vdots & \vdots & \vdots & \vdots & \vdots & \vdots & \ddots
\end{array}
\right)
\end{array}
,
\label{eq:generalizedMmatrix}
\ee
we can obtain the twists and OPE coefficients of double-twist operators $[ij]_0\dots [ij]_{N_{ij}-1}$, $[kl]_0\dots [kl]_{N_{kl}-1}$, $\dots$.
To build a more complete picture of the low-twist spectrum of the Ising model, it will be important to study (\ref{eq:generalizedMmatrix}) for $[\e T]$, $[TT]$, and other families, in addition to $[\s\s]$ and $[\e\e]$.

To summarize, we have
\be
H &= \mathrm{diag}\p{\mathrm{eigenvalues}\p{M_0^{-1} M_0'}}.
\label{eq:twisthamiltoniansummary}
\ee
where $M_0=M(\bar y_0,\bar h)$ and $M_0'=\bar \ptl M(\bar y,\bar h)|_{\bar y = \bar y_0}$.
The OPE coefficients $\L(\bar h)$ can be obtained as follows. Let
\be
\Lambda' &= \begin{pmatrix}
\l_{\s\s[\s\s]_0} & \l_{\s\s[\s\s]_1} & \l_{\s\s[\e\e]_0} \\
\l_{\s\s[\s\s]_0}h_{[\s\s]_0} & \l_{\s\s[\s\s]_1} h_{[\s\s]_1} & \l_{\s\s[\e\e]_0} h_{[\e\e]_0} \\
\l_{\e\e[\s\s]_0} & \l_{\e\e[\s\s]_1} & \l_{\e\e[\e\e]_0}
\end{pmatrix}.
\ee
(The generalization to many twist families as in (\ref{eq:generalizedMmatrix}) should be clear.)
Note that $M_0 = \L' \bar y_0^H \L'^T$ and $M_0'=\L' H \bar y_0^H \L'^T$.  Let us compute decompositions\footnote{$U_{1}$ and $U_2$ can be obtained in several ways, for example via Cholesky decomposition, or eigenvalue decomposition. If $M_0$ and $M_0'$ are positive semidefinite, then $U_{1,2}$ will be real.}
\be
M_0 &= U_1 U_1^T,\nn\\
M_0' &= U_2 U_2^T.
\ee
It must be the case that
\be
U_1 &= \L' \bar y_0^{H/2} Q_1^T,\nn\\
U_2 &= \L' \bar y_0^{H/2} H^{1/2} Q_2^T,
\label{eq:equationtosolveforlambda}
\ee
where $Q_1,Q_2$ are orthogonal matrices. To determine the $Q_{1,2}$, consider the combination
\be
\label{eq:equationforqs}
U_1^{-1} U_2 &= Q_1 H^{1/2} Q_2^T,
\ee
The right-hand side has the form of a singular value decomposition (SVD), 
so $Q_1,Q_2$ can be obtained by from an SVD of $U_1^{-1} U_2$. Finally, we solve for $\L'$ (and hence $\L$) using either equation in (\ref{eq:equationtosolveforlambda}).\footnote{It is easy to check that the number of unknowns $h_{[ij]_n}$ and $\l_{ij[kl]_n}$ always equals the total number of distinct entries in the matrices $M_0,M_0'$. Thus, we can solve for $\L$ using either equation in (\ref{eq:equationtosolveforlambda}) and we will get the same result.} Note that this procedure gives us $\l_{ij[kl]_n}$. To determine the actual OPE coefficients $f_{ij[kl]_n}$, we must multiply by Jacobian factors (\ref{eq:jacobianstuff}), which are different for each eigenvalue of the twist Hamiltonian $h_{[kl]_n}$.

\subsection{Choice of external states}
\label{sec:effectivesubspace}

We can understand the twist-Hamiltonian prescription as follows. The four-point function $\<\f_i(x_1)\f_j(x_2)\f_k(x_3)\f_l(x_4)\>$ is the amplitude for creating a state with $\f_i(x_1) \f_j(x_2)$ and annihilating it with $\f_k(x_3) \f_l(x_4)$.
States created by pairs of local operators are not eigenstates of the twist-Hamiltonian $H$.
Our task is to compute the change of basis between pair states $\f_i(x_1)\f_j(x_2)|0\>$ and $H$-eigenstates (the OPE coefficients $f_{ij[ab]_n}$), and to find the eigenvalues $h_{[ab]_n}$. For this, we need  matrix elements of $\bar y^H$ between enough states to span the Hilbert space. 

Although generically any eigenstate $\cO$ will appear in the span of $\f_i(x_1)\f_j(x_2)|0\>$ (when global charges allow it), it should be easier to study $\cO$ precisely if we use states that have large overlap with $\cO$.
Specifically, we expect to get a better picture of the $[\f_i \f_j]_n$ operators if we study matrix elements that include $\f_i(x_1)\f_j(x_2)|0\>$.
Similarly, one might learn about multi-twist operators $[\cO_1\cdots\cO_n]$ by performing very high-precision studies of four-point functions. However, it may be more efficient to study matrix elements of $\cO_1(x_1)\cdots \cO_n(x_n)|0\>$, i.e.\ to study higher-point correlators.
  
\subsection{Analogy with the renormalization group}

The difference between the twist-Hamiltonian approach and the approach of section~\ref{sec:generaldtfamilies} is analogous to the difference between RG-improved perturbation theory and fixed-order calculations. In fixed-order perturbation theory at $L$ loops, one finds powers of logarithms $\log^2 x, \dots, \log^L x$ (where $x$ is some kinematic variable) whose coefficients are related by exponentiation to coefficients at lower loop order. In RG-improved perturbation theory, we exploit this fact by choosing a scale $x_0$ and deriving a differential equation for the $x$-dependence near $x/x_0=1$.  The $\log^1 x/x_0$ terms at $L$-loops give $L$-th order corrections to anomalous dimensions, beta functions, etc..

In the context of large-spin operators, the role of $L$-loops is played by $L$-twist operators in the crossed-channel. To see exponentiation of anomalous dimensions, we must in principle  sum all multi-twist operators. Instead, in analogy with RG-improved perturbation theory, we assume exponentiation works and find anomalous dimensions by working at some scale $\bar y_0$. $L$-twist operators also give corrections to anomalous dimensions, given by the Casimir-regular terms after summing their conformal blocks. These are analogous to $\log^1 x/x_0$ terms in $L$-loop perturbation theory. To compute them, we must understand the detailed structure of the $L$-twist operators.

\subsection{Crossing symmetry for the twist Hamilonian}

To compute $M(\bar y, \bar h)$, we need the following lemma.
\begin{mylemma}
\label{lem:coperator}
If an infinite sum of $\SO(d,2)$ blocks has Casimir-singular part $f(\bar y, y)$,\footnote{We assume $p(\bar h)$ and $h(\bar h)$ depend nicely on $\bar h$, and $\bar h(\ell)$ is the solution to $\bar h(\ell) - h(\bar h(\ell)) = \ell$.}
\be
\label{eq:sodsum}
\sum_\ell \pdr{\bar h}{\ell} p(\bar h) G_{h(\bar h), \bar h}(\bar z, 1-z) &= f(\bar y, y) + [\dots]_y,
\ee
then the asymptotic density of $p(\bar h) \bar y^{h(\bar h)}$ is given by
\be
p(\bar h) \bar y^{h(\bar h)} &\sim (\cC f)(\bar y, \bar h),
\ee
where the operator $\cC$ is defined as follows. Let
\be
\cC: \bar y^{h_0} y^a &\mto \sum_{n=0}^\oo \bar y^{h_0+n} C_a^{(n)}(h_0,\bar h),
\ee
and extend $\cC$ linearly to arbitrary sums of powers and logs of $y,\bar y$. Here, $C^{(n)}_a(h_0,\bar h)$ are the coefficients defined in section~\ref{sec:sumovernandell}.
\end{mylemma}

\begin{proof}
By linearity, it suffices to consider $f(\bar y,y)=c(\bar y) y^a$ for some function $c(\bar y)$.  Let us assume
\be
p(\bar h) \bar y^{h(\bar h)} &\sim \sum_{n=0}^\oo \bar y^n C_a^{(n)}\p{\pdr{}{\log \bar y}, \bar h} c(\bar y),
\label{eq:weclaim}
\ee
and show that the sum (\ref{eq:sodsum}) has Casimir-singular part $c(\bar y) y^a$.
 Since Casimir-singular terms uniquely determine an asymptotic $\bar h$-expansion for coefficients of blocks, the claim follows.

As before, let $\bar\ptl=\pdr{}{\log \bar y}$. The $\SO(d,2)$ blocks have expansion
\be
G_{h,\bar h}(\bar z, 1-z) &= \p{\sum_{m=0}^\oo \sum_{j=-m}^m\bar y^m A_{m,j}(\bar \ptl,\bar h)k_{2(\bar h+j)}(1-z)} \bar y^{h}.
\ee
Applying the differential operator in parentheses to (\ref{eq:weclaim}), we get
\be
p(\bar h) G_{h(\bar h), \bar h}(\bar z, 1-z) &=
\sum_{n=0}^\oo\sum_{m=0}^\oo \sum_{j=-m}^m \bar y^m A_{m,j}(\bar\ptl,\bar h)  \bar y^n C^{(n)}_a(\bar \ptl,\bar h) c(\bar y) k_{2(\bar h + j)}(1-z)\nn\\
&= \sum_{n=0}^\oo\sum_{m=0}^\oo \sum_{j=-m}^m \bar y^{n+m} A_{m,j}(\bar\ptl + n, \bar h) C_a^{(n)}(\bar \ptl, \bar h) c(\bar y) k_{2(\bar h+j)}(1-z)\nn\\
&\sim \sum_{n=0}^\oo \bar y^{n}\sum_{m=0}^n \sum_{j=-m}^m  A_{m,j}(\bar\ptl + n-m, \bar h-j) C_a^{(n-m)}(\bar \ptl, \bar h-j) c(\bar y) k_{2\bar h}(1-z).
\ee
In the last line, ``$\sim$" indicates that the two sides give the same Casimir-singular part when summed over $\bar h$ (since shifting $\bar h\to \bar h - j$ only affects Casimir-regular terms). Finally, applying the recursion relation (\ref{eq:recursionrelationforCs}) with $h_0=\bar \ptl$ we get
\be
p(\bar h) G_{h(\bar h), \bar h}(\bar z, 1-z) &\sim S_a(\bar h) c(\bar y) k_{2\bar h}(1-z).
\ee
Summing over $\bar h$ gives the desired result.
\end{proof}

Lemma~\ref{lem:coperator} generalizes trivially to the case of mixed blocks, where we must use the operator
\be
\cC^{r,s}: \bar y^{h_0} y^a &\mto \sum_{n=0}^\oo \bar y^{h_0+n} C_a^{(n)r,s}(h_0,\bar h).
\ee
Applying $\cC^{h_{12},h_{34}}$ to the left-hand side of the crossing equation (\ref{eq:mixedcrossingequation}), we obtain
\be
\label{eq:sumoffamilies}
M_{1234}(\bar y, \bar h) &=
\sum_i \l_{12\cO_i}(\bar h) \l_{43\cO_i}(\bar h) \bar y^{h_i(\bar h)}\nn\\
&\sim
\cC^{h_{12},h_{34}}\p{
\bar y^{h_1+h_3} y^{-h_1-h_3} G_{3214}(z,1-\bar z)
} + (-1)^\ell (3\leftrightarrow 4),\\
G_{3214}(z,\bar z) &\equiv \sum_{\cO} f_{32\cO}f_{41\cO} G_{h_\cO,\bar h_\cO}^{h_{32},h_{14}}(z,\bar z),
\ee
where $i$ runs over twist families in the $1\x2$ and $3\x 4$ OPEs.  As in section~\ref{sec:matchingtosides}, we must add $(-1)^\ell (3\leftrightarrow 4)$  for consistency with the symmetry properties of $\l_{43\cO_i}$.

The contribution of an individual block to (\ref{eq:sumoffamilies}) is,
\be
&\cC^{h_{12},h_{34}} \p{
\bar y^{h_1+h_3} y^{-h_1-h_3} G_{h_\cO,\bar h_\cO}^{h_{32},h_{14}}(z,1-\bar z)
}=
\sum_{m=0}^\oo U^{1234}_{\cO,m}(\bar y, \bar h),
\ee
where
\be
U^{1234}_{\cO,m}(\bar y, \bar h) &\equiv
\sum_{n=0}^\oo \p{U^{(n)1234}_{\cO,m}(\bar h)\bar y^{n+h_1+h_2} + U^{(n)3412}_{\cO,m}(\bar h)\bar y^{n+h_3+h_4}}.
\ee
Using~(\ref{eq:implicitdefofVW}), this has a smooth limit as $h_1+h_2\to h_3+h_4$,
\be
U^{1234}_{\cO,m}(\bar y, \bar h) &= \sum_{n=0}^\oo \bar y^{n+h_1+h_2}\p{V^{(n)1234}_{\cO,m}(\bar h) \log \bar y + W^{(n)1234}_{\cO,m}(\bar h)} \qquad (h_1+h_2=h_3+h_4).
\ee
As a special case, the unit operator contributes
\be
U^{1221}_{\mathbf{1}}(\bar y, \bar h) &= \cC^{h_{12},h_{21}}(\bar y^{h_1+h_2} y^{-h_1-h_2}) = \sum_{n=0}^\oo \bar y^{n+h_1+h_2} C^{(n)h_{12},h_{21}}_{-h_1-h_2}(h_1+h_2,\bar h)\nn\\
&= \sum_{n=0}^\oo \bar y^{n+h_1+h_2} C^\mathrm{MFT}_{n,\bar h-h_1-h_2-n}(2h_1,2h_2).
\ee

\subsection{Application to $[\e\e]_0$ and $[\s\s]_1$}
\label{sec:eezerosigsigone}

\subsubsection{Why large mixing?}

Before computing the Hamiltonian for $[\e\e]_0$ and $[\s\s]_1$, let us explain intuitively why the two families exhibit large mixing at intermediate values of $\bar h$. At very large $\bar h$, the dominant contributions to the anomalous dimensions of $[\e\e]_0$ and $[\s\s]_1$ come from exchange of the stress tensor $T_{\mu\nu}$, and mixing is negligible. However, the operators $[\s\s]_0$ have twist only slightly larger than $T_{\mu\nu}$, so all of their contributions become important at slightly smaller $\bar h$.\footnote{In a weakly-coupled theory, there is no regime where the stress-tensor completely dominates over the first higher-spin family.}

\begin{figure}
\begin{center}
\be
\begin{array}{ccc}
\begin{tikzpicture}[xscale=0.6,yscale=0.6,baseline=(B.base)]
\draw[thick,dashed] (0,0) -- (0,1.5);
\draw[thick] (0,1.5) -- (0,5.5);
\draw[thick,dashed] (0,5.5) -- (0,7);
\draw[thick,dashed] (5,0) -- (5,1.5);
\draw[thick] (5,1.5) -- (5,5.5);
\draw[thick,dashed] (5,5.5) -- (5,7);
\draw[thick] (0,1.5) -- (5,1.5);
\draw[thick] (0,5.5) -- (5,5.5);
\node[below] at (0,0) {$\e$};
\node[below] at (5,0) {$\e$};
\node[above] at (0,7) {$\e$};
\node[above] at (5,7) {$\e$};
\node[below] at (2.5,1.5) {$\s$};
\node[right] at (5,3.5) {$\s$};
\node[above] at (2.5,5.5) {$\s$};
\node[left] at (0,3.5) {$\s$};
\node (B) [left] at (0,3.4) {};
\end{tikzpicture}
&
\qquad\longrightarrow\qquad
&
\exp
\left(
\begin{array}{cc}
&
\begin{tikzpicture}[xscale=0.3,yscale=0.3,baseline=(B.base)]
\draw[thick] (0,0) -- (0,1.5);
\draw[thick,dashed] (0,1.5) -- (0,3.5);
\draw[thick] (5,0) -- (5,1.5);
\draw[thick,dashed] (5,1.5) -- (5,3.5);
\draw[thick] (0,1.5) -- (5,1.5);
\node[below] at (0,0) {$\s$};
\node[below] at (5,0) {$\s$};
\node[above] at (2.5,1.5) {$\s$};
\node[above] at (5,3.5) {$\e$};
\node[above] at (0,3.5) {$\e$};
\node (B) [left] at (0,3.4) {};
\end{tikzpicture}
\\
\begin{tikzpicture}[xscale=0.3,yscale=0.3,baseline=(B.base)]
\draw[thick,dashed] (0,0) -- (0,1.5);
\draw[thick] (0,1.5) -- (0,3.5);
\draw[thick,dashed] (5,0) -- (5,1.5);
\draw[thick] (5,1.5) -- (5,3.5);
\draw[thick] (0,1.5) -- (5,1.5);
\node[below] at (0,0) {$\e$};
\node[below] at (5,0) {$\e$};
\node[above] at (2.5,1.5) {$\s$};
\node[above] at (5,3.5) {$\s$};
\node[above] at (0,3.5) {$\s$};
\node (B) [left] at (0,3.4) {};
\end{tikzpicture}
&
\end{array}
\right)
\end{array}
\ee
\end{center}
\caption{
Exchange of large-spin $[\s\s]_0$ operators looks like the exponentiation of a Hamiltonian that mixes $[\s\s]$ and $[\e\e]$.
}
\label{fig:mixingillustration}
\end{figure}

As illustrated in figure~\ref{fig:mixingillustration}, exchange of large-spin $[\s\s]_0$ operators (namely operators where the vertical distance between $\s$ lines in figure~\ref{fig:mixingillustration} is large) looks like a product of off-diagonal terms that transition between $[\e\e]$ and $[\s\s]$, coming from $\s$-exchange in the $\<\s\e\s\e\>$ four-point function. This is part of the exponentiation of a twist Hamiltonian with structure
\be
H(\bar h) &\approx \begin{pmatrix}
2h_\e + \bar h^{-\tau_T} & \bar h^{-\Delta_\s} \\
\bar h^{-\Delta_\s} & (2h_\s + 1) + \bar h^{-\tau_T}
\end{pmatrix}.
\ee
The off-diagonal terms are unimportant at very large $\bar h$. (We should compare the square of the off-diagonal terms to the diagonal terms.) However, they become important at slightly smaller $\bar h$. In fact, because $2h_\e\approx 2h_\s+1$, they cause the eigenvalues to repel significantly.

\subsubsection{Computing the twist Hamiltonian}

To find the twist Hamiltonian for $[\e\e]_0$ and $[\s\s]_1$, we must compute $M_{\s\s\s\s}$, $M_{\e\e\e\e}$, and $M_{\s\s\e\e}$. For example,
\be
M_{\s\s\s\s}(\bar y, \bar h) &\sim 2\,\cC\p{\bar y^{2h_\s} y^{-2h_\s}G_{\s\s\s\s}(z,1-\bar z)}.
\ee
We will take the first few terms in an asymptotic expansion in large-$\bar h$, so we should truncate powers of $y$ (so that only low-twist operators contribute) before applying $\cC$.  In the $G_{\s\s\s\s}$ correlator, we will include terms up to order $y^{h_\e}$. Let us describe the low-twist part of the correlators $G_{\s\s\s\s}$, $G_{\e\e\e\e}$, and $G_{\e\s\s\e}$ in more detail.

\subsubsection{$G_{\s\s\s\s}$}

We have
\be
\label{eq:blocksssss}
G_{\s\s\s\s}(z,1-\bar z) &= 1 + \sum_{\substack{\cO=[\s\s]_{0,\ell} \\ \ell = 2, 4, \dots}} f_{\s\s\cO}^2 y^{h_{\cO}} k_{2 \bar h_{\cO}}(1-\bar z) + f_{\s\s\e}^2 y^{h_\e}k_{2\bar h_\e}(1-\bar z) 
 + \dots,
\ee
where ``$\dots$" represents terms of higher order than $y^{h_\e}$.\footnote{Here, we assume that no $\Z_2$-even operators other than the ones written have twist less than $\De_\e$. This is supported by numerics but we cannot prove it.} 
Let us split the sum over $[\s\s]_0$ into a finite part which we treat exactly and an infinite part which we expand in small anomalous dimensions $\de_{[\s\s]_0}$,
\be
\label{eq:splitfiniteinfinite}
\sum_{\substack{\cO=[\s\s]_{0,\ell} \\ \ell = 2, 4, \dots}} f_{\s\s\cO}^2 y^{h_{\cO}} k_{2 \bar h_{\cO}}(1-\bar z)
&=
\p{\sum_{\substack{\cO=[\s\s]_{0,\ell} \\ \ell = 2,4,\dots,\ell_0-2}} + \sum_{\substack{\cO=[\s\s]_{0,\ell} \\ \ell = \ell_0,\ell_0+2\dots}}} f_{\s\s\cO}^2 y^{h_{\cO}} k_{2 \bar h_{\cO}}(1-\bar z).
\ee
We can make $\de_{[\s\s]_0}$ arbitrarily small by choosing $\ell_0$ large enough. Taking $\ell_0=6$ will be sufficient for our purposes. Thus, the finite sum in (\ref{eq:splitfiniteinfinite}) will contain the stress tensor and the spin-4 operator $[\s\s]_{0,4}$.
For these contributions, we use the expansion of $k_{2\bar h}(1-\bar z)$ up to first order in $\bar y$,
\be
y^{h_\cO} k_{2\bar h_\cO}(1-\bar z) &\approx y^{h_\cO}\sum_{k=0}^1 \pdr{}{k}\p{-\frac{\G(2\bar h_\cO)}{\G(\bar h_\cO)^2}T_{-k-1}(\bar h_\cO) \bar y^k}.
\ee

Meanwhile, expanding the infinite sum in $\de_{[\s\s]_0}$, we obtain
\be
\sum_{\substack{\cO=[\s\s]_{0,\ell} \\ \ell = \ell_0, \ell_0+2, \dots}} f_{\s\s\cO}^2 y^{h_{\cO}} k_{2 \bar h_{\cO}}(1-\bar z)
&=\sum_{m=0}^\oo 
y^{2h_\s}\log^m y \sum_{\ell=\ell_0,\ell_0+2,\dots} \pdr{\bar h}{\ell}
\frac{\de_{[\s\s]_0}(\bar h)^m}{m!} \l_{\s\s[\s\s]_0}(\bar h)^2 k_{2\bar h}(1-\bar z).
\ee
The quantities $\l_{\s\s[\s\s]_0}$ and $\de_{[\s\s]_0}$ are given in~(\ref{eq:sigsigfit0}) and (\ref{eq:sigsigfit1}).
We can compute the sums over $\ell$ using the methods of appendix~\ref{app:nonintegerspacing}, 
\be
&\sum_{\ell=\ell_0,\ell_0+2,\dots} \pdr{\bar h}{\ell}
\frac{\de_{[\s\s]_0}(\bar h)^m}{m!} \l_{\s\s[\s\s]_0}(\bar h)^2 k_{2\bar h}(1-\bar z)\nn\\
&=\sum_{a\in A_m} \frac 1 2 c_a^{(m)} \bar y^a + \sum_{k=0}^\oo \pdr{}{k}\p{\bar y^k\a_k^\mathrm{even}\left[\frac{\de_{[\s\s]_0}(\bar h)^m}{m!} \l_{\s\s[\s\s]_0}(\bar h)^2, \de_{[\s\s]_0}\right](2h_\s+\ell_0)},
\ee
where $c^{(m)}_a$ are coefficients in the asymptotic expansion
\be
\frac 1 {m!}\de_{[\s\s]_0}(\bar h)^m \l_{\s\s[\s\s]_0}(\bar h)^2 &\sim \sum_{a\in A_m} c_a^{(m)} S_a(\bar h).
\ee
Only the $m\geq 2$ cases will survive the $\cC$ operation (because $\log^m y$ is Casimir-regular for $m=0,1$). However the $m=2$ term is already quite small, so it will be sufficient to truncate the series here.  The first few $c_a^{(2)}$ are
\be
\label{eq:examplec2coeffs}
\frac{1}{2}\de_{[\s\s]_0}^2 \l_{\s\s[\s\s]_0}^2  &\sim \sum_{\cO_1,\cO_2\in\{\e,T\}} f_{\s\s\cO_1}^2 f_{\s\s\cO_2}^2
\frac{\Gamma (2 \bar h_1) \Gamma (2 \bar h_2) \Gamma (2 h_\s)^2 \Gamma (2 h_\s-h_1-h_2)^2}{\Gamma (\bar h_1)^2 \Gamma (\bar h_2)^2 \Gamma (2 h_\s-h_1)^2 \Gamma (2 h_\s-h_2)^2} S_{h_1+h_2-2h_\s}(\bar h) +\dots\nn\\
&= \sum_{\cO_1,\cO_2\in\{\e,T\}} c_{h_1+h_2-2h_\s}^{(2)} S_{h_1+h_2-2h_\s}(\bar h) + \dots.
\ee
The $S_{2h_\e-2h_\s}$ term is important because it gives a contribution to $M_{\s\s\s\s}(\bar y, \bar h)$ proportional to $\bar y^{2h_\e}$, which contributes to mixing with $[\e\e]_0$. In general, we find terms of the form $S_{h_1+\dots + h_n + k - 2h_\s}(\bar h)$ where $h_i \in \{h_T, h_\e\}$ and $k\in \Z_{\geq 0}$.

\begin{figure}
\begin{center}
\begin{tikzpicture}[xscale=0.6,yscale=0.6]
\draw[thick] (0,0) -- (0,1.5);
\draw[thick, dashed] (0,1.5) -- (0,5.5);
\draw[thick] (0,5.5) -- (0,7);
\draw[thick] (5,0) -- (5,1.5);
\draw[thick, dashed] (5,1.5) -- (5,5.5);
\draw[thick, dashed] (1.2,1.5) -- (1.2,5.5);
\draw[thick, dashed] (2.4,1.5) -- (2.4,5.5);
\draw[thick] (5,5.5) -- (5,7);
\draw[thick] (0,1.5) -- (5,1.5);
\draw[thick] (0,5.5) -- (5,5.5);
\node[below] at (0,0) {$\s$};
\node[below] at (5,0) {$\s$};
\node[above] at (0,7) {$\s$};
\node[above] at (5,7) {$\s$};
\node[below] at (2.5,1.5) {$\s$};
\node[above] at (2.5,5.5) {$\s$};
\node[left] at (0.1,3.5) {$h_1$};
\node[] at (0.7,3.5) {$h_2$};
\node[] at (1.9,3.5) {$h_3$};
\node[] at (3.8,3.5) {$\dots$};
\end{tikzpicture}
\end{center}
\caption{The operators $[\s\s]_0$ give contributions to $M(\bar y,\bar h)$ of the form $\bar y^{h_1+\cdots+h_i+k}$. At large-$\bar h$, these must be matched by multi-twist operators $[\cO_1\cdots \cO_i]_k$. In the picture, we can see how multi-twist operators must appear in the $\s\x\s$ OPE (bottom-to-top channel) because they come from the exponentiation of the anomalous dimensions of $[\s\s]$ in the left-to-right channel.}
\label{fig:illustrationofmultitwist}
\end{figure}
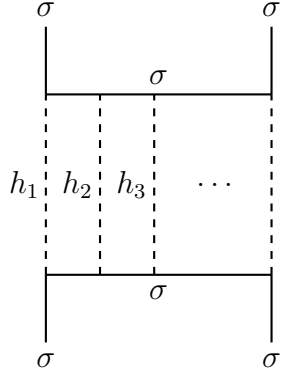

Here, we can see the exponentiation discussed in section~\ref{sec:whatabout} at work. Summing over the family $[\s\s]_0$ gives terms of the form $\bar y^{h_1 + \dots + h_i + k}$, where $h_i$ are half-twists of other operators in the theory. These give contributions to the twist Hamiltonian proportional to $h_1+\dots+h_i+k$, which must be matched by multi-twist operators $[\cO_1\dots\cO_i]_k$. This is illustrated in figure~\ref{fig:illustrationofmultitwist}.

\newcommand\lb{\nn\\&\quad}
\newcommand\parenspace{\quad\quad\quad\quad}

Plugging in the values (\ref{eq:actualvalues}) and multiplying by $\bar y^{2h_\s} y^{-2h_\s}$, the infinite sum is
\be
\label{eq:currentsumsigsigsigsig}
&\bar{y}^{2h_\s}y^{-2h_\s}\sum_{\substack{\cO=[\s\s]_{0,\ell} \\ \ell = 6, 8, \dots}} f_{\s\s\cO}^2 y^{h_{\cO}} k_{2 \bar h_{\cO}}(1-\bar z)
\nn\\
&=
1
- 3.69919 \bar{y}^{h_T}
+ (3.37064 \log \bar y+11.6413) \bar{y}^{2 h_{\sigma }}
\lb
+ 0.739023 \bar{y}^{h_{\epsilon }}
+ (28.1977 \log \bar y+44.2112) \bar{y}^{2 h_{\sigma }+1}
+\dots
\lb
+
\log y \Big(
-1.44458 \bar{y}^{h_T}
+ (0.0173629 \log \bar y+1.88281) \bar{y}^{2 h_{\sigma }}
\lb \parenspace
- 0.591176 \bar{y}^{h_{\epsilon }}
+ (-0.267215 \log \bar y-0.278914) \bar{y}^{2 h_{\sigma }+1}
+\dots
\Big)
\lb
+
\log^2y \Big(
(-0.0000261014 \log \bar y-0.000146056) \bar{y}^{2 h_{\sigma }}
\lb \parenspace
+ 4.36051\!\.\!10^{-6} \bar{y}
+ 0.000391581 \bar{y}^{h_T+h_{\epsilon }}
+ 0.0369549 \bar{y}^{2 h_{\epsilon }}
\lb \parenspace
+ 3.88489\!\.\!10^{-6} \bar{y}^{h_T+1}
+ (0.00506123 \log \bar y-0.0347285) \bar{y}^{2 h_{\sigma }+1}
\lb \parenspace
+ 1.64132\!\.\!10^{-6} \bar{y}^{h_{\epsilon }+1}
+ 2.26961\!\.\!10^{-7} \bar{y}^{h_T+2 h_{\epsilon }}
- 9.74836\!\.\!10^{-7} \bar{y}^2
+ \dots
\Big)
\lb
+\dots
\ee
where ``$\dots$" represent terms higher order in $\bar y$ or $\log y$. We stress that while we have written the above coefficients numerically for brevity, they all have analytic formulae.  For example, the coefficient of $\log^2 y\, \bar y^{2h_\e}$ is given by $\frac 1 2 c_{2h_\e-2h_\s}^{(2)}$ in (\ref{eq:examplec2coeffs}).

We have written ``$\approx$" instead of ``$=$" because the above formula is based on the approximations~(\ref{eq:sigsigfit0}) and (\ref{eq:sigsigfit1}) for $\l_{\s\s[\s\s]_0}$ and $\de_{[\s\s]_0}$. Because those formulae match the numerical data to high accuracy, the same is true of (\ref{eq:currentsumsigsigsigsig}). However, a more sophisticated approximation for $\l_{\s\s[\s\s]_0}$, $\de_{[\s\s]_0}$ would include contributions from operators $\cO_i$ other than $\e,T$, giving rise to additional terms like $\bar y^{h_i}$ in (\ref{eq:currentsumsigsigsigsig}).\footnote{Including the contribution of the whole family $[\s\s]_0$ to itself would give $\log^m y \log^n \bar y$ terms, coming in at order $m+n+3$ in the expansion in the small parameter $2h_\s-\nu$. Such terms have been discussed in~\cite{Alday:2015ota}.}

\subsubsection{$G_{\e\e\e\e}$}

The computation of $G_{\e\e\e\e}(z,1-\bar z)$ proceeds similarly. We have
\be
M_{\e\e\e\e}(\bar y, \bar h) &= 2\cC\p{\bar y^{2h_\e} y^{-2h_\e} G_{\e\e\e\e}(z,1-\bar z) },\nn\\
\label{eq:blockeeee}
G_{\e\e\e\e}(z,1-\bar z) &= 1 + \sum_{\substack{\cO=[\s\s]_{0,\ell} \\ \ell = 2, 4, \dots}} f_{\e\e\cO}^2 y^{h_{\cO}} k_{2 \bar h_{\cO}}(1-\bar z) + f_{\e\e\e}^2 y^{h_\e}k_{2\bar h_\e}(1-\bar z) 
 + \dots.
\ee
The coefficient $f_{\e\e\e}$ is given by (\ref{eq:isingisland}).
We split the sum over $[\s\s]_0$ into a finite part ($\ell<6$) and an infinite part $(\ell \geq 6$) and expand the infinite part in small $\de_{[\s\s]_0}$, up to order $m=2$. This time all the terms $m=0,1,2$ contribute nontrivially after the $\cC$ operation.
The OPE coefficients $\l_{\e\e[\s\s]_0}$ can be obtained from (\ref{eq:betterfitforfeess0}). The infinite sum is
\be
\label{eq:currentsumepsepsepseps}
&\bar y^{2h_\e} y^{-2h_\e} \sum_{\substack{\cO=[\s\s]_{0,\ell} \\ \ell = 6, 8, \dots}} f_{\e\e\cO}^2 y^{h_{\cO}} k_{2 \bar h_{\cO}}(1-\bar z)
\nn\\
&\approx
y^{2h_\s-2h_\e}
\Big(
1.83831 \bar{y}^{2 h_{\sigma }}
+ 0.0294478 \bar{y}^{2 h_{\sigma }+h_T}
- 11.8305 \bar{y}^{2 h_{\sigma }+h_{\epsilon }}
\lb \parenspace
+ (23.1945 \log \bar y+54.9846) \bar{y}^{2 h_{\epsilon }}
+ 57.4846 \bar{y}^{2 h_{\sigma }+1}
\lb \parenspace
- 0.038081 \bar{y}^{2 h_{\sigma }+h_T+h_{\epsilon }}
+ 0.609036 \bar{y}^{2 h_{\sigma }+2 h_{\epsilon }}
+ \dots
\Big)
\lb
+
y^{2h_\s-2h_\e} \log y \Big(
-0.0114997 \bar{y}^{2 h_{\sigma }+h_T}
- 1.77142 \bar{y}^{2 h_{\sigma }+h_{\epsilon }}
\lb \parenspace \parenspace
+ (0.604068 -0.746285 \log \bar y) \bar{y}^{2 h_{\epsilon }}
- 0.00171201 \bar{y}^{2 h_{\sigma }+1}
\lb \parenspace \parenspace
+ 0.00526509 \bar{y}^{2 h_{\sigma }+h_T+h_{\epsilon }}
+ 0.341187 \bar{y}^{2 h_{\sigma }+2 h_{\epsilon }}
+ \dots
\Big)
\lb
+
y^{2h_\s-2h_\e} \log^2y \Big(
(-0.000745198 \log \bar y-0.00496677) \bar{y}^{2 h_{\epsilon }}
+ 0.00016714 \bar{y}^{2 h_{\sigma }+1}
\lb \parenspace \parenspace
+ 0.00187562 \bar{y}^{2 h_{\sigma }+h_T+h_{\epsilon }}
+ 0.0215162 \bar{y}^{2 h_{\sigma }+2 h_{\epsilon }}
+ \dots
\Big)
\lb
+ \dots
\ee

\subsubsection{$G_{\s\e\e\s}$}

For $M_{\s\s\e\e}$, we have
\be
M_{\s\s\e\e}(\bar y, \bar h) &\sim 2 \cC \p{\bar y^{h_\s+h_\e} y^{-h_\s-h_\e} G_{\e\s\s\e}(z,1-\bar z) }, \nn\\
G_{\e\s\s\e}(z,1-\bar z) &=
f_{\s\s\e}^2 G^{h_{\e\s},h_{\s\e}}_{h_\s,h_\s}(z,1-\bar z)
+ \sum_{\substack{\cO=[\s\e]_{0,\ell} \\ \ell=2,3,\dots}} f_{\s\e\cO}^2 y^{h_\cO} k^{h_{\e\s},h_{\s\e}}_{2\bar h_\cO}(1-\bar z) + \dots.
\ee
Here, ``$\dots$" represents higher order terms in $y$.  We keep the terms of order $y^{h_\s}$ and $y^{h_\s+1}$ in the conformal block for $\s$.  The sum over $[\s\e]_0$ can be performed as before, by splitting it into a finite part $\ell<\ell_0$ that we treat exactly and an infinite part $\ell\geq \ell_0$ that we expand in the anomalous dimension $\de_{[\s\e]_0}$. The quantities $\l_{\s\e[\s\e]_0}^2$ and $\de_{[\s\e]_0}$ are given in (\ref{eq:secondapproximationfse}) and (\ref{eq:secondapproximationdeltase}).
We expand to fifth order in $\de_{[\s\e]_0}$ and take $\ell_0=6$.  The final result for $M_{\s\s\e\e}(\bar y, \bar h)$ is independent of $\ell_0$ to high precision. The infinite sum is
\be
\label{eq:currentsumsigepsepssig}
&\bar y^{h_\e+h_\s} y^{-h_\e-h_\s} \sum_{\substack{\cO=[\s\e]_{0,\ell} \\ \ell = 6, 7, \dots}} f_{\s\e\cO}^2 y^{h_{\cO}} k^{h_{\e\s},h_{\s\e}}_{2 \bar h_{\cO}}(1-\bar z)
\nn\\
&\approx
1
- 10.0851 \bar{y}^{h_T}
+ (1.02429 -0.234943 \log \bar y) \bar{y}^{2 h_{\sigma }}
\lb
+ 1.07667 \bar{y}^{h_{\epsilon }}
+ 650.249 \bar{y}^{2 h_{\epsilon }}
- 884.116 \bar{y}^{2 h_{\sigma }+1}
+\dots
\lb
+
\log y \Big(
-3.93834 \bar{y}^{h_T}
+ (4.06924 -0.123248 \log \bar y) \bar{y}^{2 h_{\sigma }}
\lb \parenspace
- 0.861278 \bar{y}^{h_{\epsilon }}
+ 18.8077 \bar{y}^{2 h_{\epsilon }}
- 22.0514 \bar{y}^{2 h_{\sigma }+1}
+\dots
\Big)
\lb
+
\log^2y \Big(
(-0.0170791 \log \bar y-0.0824236) \bar{y}^{2 h_{\sigma }}
- 0.0269513 \bar{y}
- 0.0525874 \bar{y}^{2 h_{\sigma }+h_T}
\lb\parenspace
+ 0.0409684 \bar{y}^{4 h_{\sigma }}
+ 0.209787 \bar{y}^{h_T+h_{\epsilon }}
+ 0.010682 \bar{y}^{2 h_{\sigma }+h_{\epsilon }}
\lb\parenspace
+ (0.0911374 \log \bar y+4.34364) \bar{y}^{2 h_{\epsilon }}
+ 0.629932 \bar{y}^{h_T+1}
\lb\parenspace
+ (-0.101919 \log \bar y-8.8983) \bar{y}^{2 h_{\sigma }+1}
+ 2.31403 \bar{y}^{4 h_{\sigma }+h_T}
\lb\parenspace
- 0.0830807 \bar{y}^{6 h_{\sigma }}
+ 0.0512032 \bar{y}^{h_{\epsilon }+1}
+ 0.165984 \bar{y}^{2 h_{\sigma }+h_T+h_{\epsilon }}
- 0.0762239 \bar{y}^{4 h_{\sigma }+h_{\epsilon }}
\lb\parenspace
- 0.000540742 \bar{y}^{h_T+2 h_{\epsilon }}
- 0.00322981 \bar{y}^{2 h_{\sigma }+2 h_{\epsilon }}
+ 0.00610087 \bar{y}^2
+\dots
\Big)
\lb
+
\log^3y \Big(
-0.0000360366 \bar{y}^{2 h_{\sigma }}
+ 0.0131415 \bar{y}^{2 h_{\sigma }+h_T}
- 0.00465426 \bar{y}^{4 h_{\sigma }}
\lb\parenspace
- 0.0559877 \bar{y}^{2 h_{\sigma }+h_{\epsilon }}
+ 0.217299 \bar{y}^{2 h_{\epsilon }}
- 0.0819986 \bar{y}^{h_T+1}
\lb\parenspace
+ (0.0176415 \log \bar y+0.14041) \bar{y}^{2 h_{\sigma }+1}
- 0.126862 \bar{y}^{4 h_{\sigma }+h_T}
- 0.0592894 \bar{y}^{6 h_{\sigma }}
\lb\parenspace
- 0.013226 \bar{y}^{h_{\epsilon }+1}
- 0.0691064 \bar{y}^{2 h_{\sigma }+h_T+h_{\epsilon }}
+ 0.019903 \bar{y}^{4 h_{\sigma }+h_{\epsilon }}
\lb\parenspace
+ 0.00892839 \bar{y}^{h_T+2 h_{\epsilon }}
+ 0.0103174 \bar{y}^{2 h_{\sigma }+2 h_{\epsilon }}
+ 0.000566515 \bar{y}^2
+\dots
\Big)
\lb
+
\log^4y \Big(
0.0000223151 \bar{y}^{2 h_{\sigma }}
- 0.000536879 \bar{y}^{4 h_{\sigma }}
+ 0.00019118 \bar{y}^{2 h_{\epsilon }}
\lb\parenspace
+ (0.0348547 -0.000867286 \log \bar y) \bar{y}^{2 h_{\sigma }+1}
- 0.0521425 \bar{y}^{4 h_{\sigma }+h_T}
\lb\parenspace
+ 0.0114817 \bar{y}^{6 h_{\sigma }}
+ 0.00543315 \bar{y}^{2 h_{\sigma }+h_T+h_{\epsilon }}
+ 0.000922452 \bar{y}^{4 h_{\sigma }+h_{\epsilon }}
\lb\parenspace
- 0.00202392 \bar{y}^{2 h_{\sigma }+2 h_{\epsilon }}
- 0.0000276538 \bar{y}^2
+\dots
\Big)
\lb
+
\log^5y \Big(
-4.06997\!\.\!10^{-8} \bar{y}^{2 h_{\sigma }}
+ 0.000231809 \bar{y}^{2 h_{\epsilon }}
- 0.00316706 \bar{y}^{2 h_{\sigma }+1}
\lb\parenspace
+ 0.00319354 \bar{y}^{4 h_{\sigma }+h_T}
- 0.0000952085 \bar{y}^{6 h_{\sigma }}
- 0.000188447 \bar{y}^{4 h_{\sigma }+h_{\epsilon }}
+\dots
\Big)
\lb
+\dots
\ee

\subsubsection{Choice of $\bar y_0$}

After computing $M(\bar y, \bar h)$, we must choose a value $\bar y_0$ at which to evaluate the twist Hamiltonian (\ref{eq:twisthamiltoniansummary}).
This presents a trade-off. Small $\bar y_0$ is good because higher-twist operators are exponentially suppressed.\footnote{Additionally, we truncate $M(\bar y, \bar h)$ at order $\bar y^2$, which also removes the effects of higher twist families.}

However, very small $\bar y_0$ is bad for the following reason. Consider the expansion
\be
\label{eq:exponentiatedseries}
\bar y^{\de} &= 1 + \de \log \bar y + \frac 1 2 \de^2 \log^2 \bar y + \frac 1 6 \de^3 \log^3 \bar y+\dots.
\ee
As explained in section~\ref{sec:whatabout}, if the $\log \bar y$ term gets a contribution from exchange of $\cO$ in the crossed-channel, then $\log^2 \bar y$ comes from the exchange of double-twist operators $[\cO\cO]$. Similarly, $\log^3 \bar y$ comes from the exchange of twiple-twist operators $[\cO\cO\cO]$, and so on. If we only include operators with bounded twist in the crossed-channel, we truncate the series (\ref{eq:exponentiatedseries}) and lose exponentiation. This becomes a problem when $\de \log \bar y$ is large. In other words, when $|\!\log \bar y| \gtrsim 1/\de$ there are large logarithms that have not been correctly re-summed because we have not included arbitrary multi-twist operators $[\cO\cdots\cO]$ in the crossed-channel.

In our case, we have included double-twist operators built out of $\s$'s and $\e$'s, so we expect to find errors that go like $\log^3 \bar y\, h^{-2h_\s-h_\e}$ and $\log^4\bar y\, \bar h^{-4h_\s}$, coming from $[\s\s\e]$ and $[\s\s\s\s]$. For small spins, the anomalous dimensions of $[\s\s]_1$ and $[\e\e]_0$ grow to $\sim 0.5$, suggesting we should not take $\bar y_0$ much smaller than $e^{-1/0.5}\sim 0.1$.\footnote{It should be possible to surmount these difficulties with a more sophisticated analysis. If we include higher-twist families $[\s\s]_{n\geq 1}$ and $[\e\e]_{n\geq 0}$ in the twist Hamiltonian, there is less downside to working at larger $\bar y_0$. On the other hand, we expect these higher families to have larger mixing with other families like $[TT]$, $[T\e]$, etc.. So it may be necessary to study a larger system of correlators at the same time. Alternatively, we might try to restore exponentiation of (\ref{eq:exponentiatedseries}) by approximating the contribution of multi-twist operators $[\cO\cdots\cO]$ in some way.}

\subsubsection{Comparison to numerics}

\begin{figure}[ht!]
\begin{center}
\includegraphics[width=0.9\textwidth]{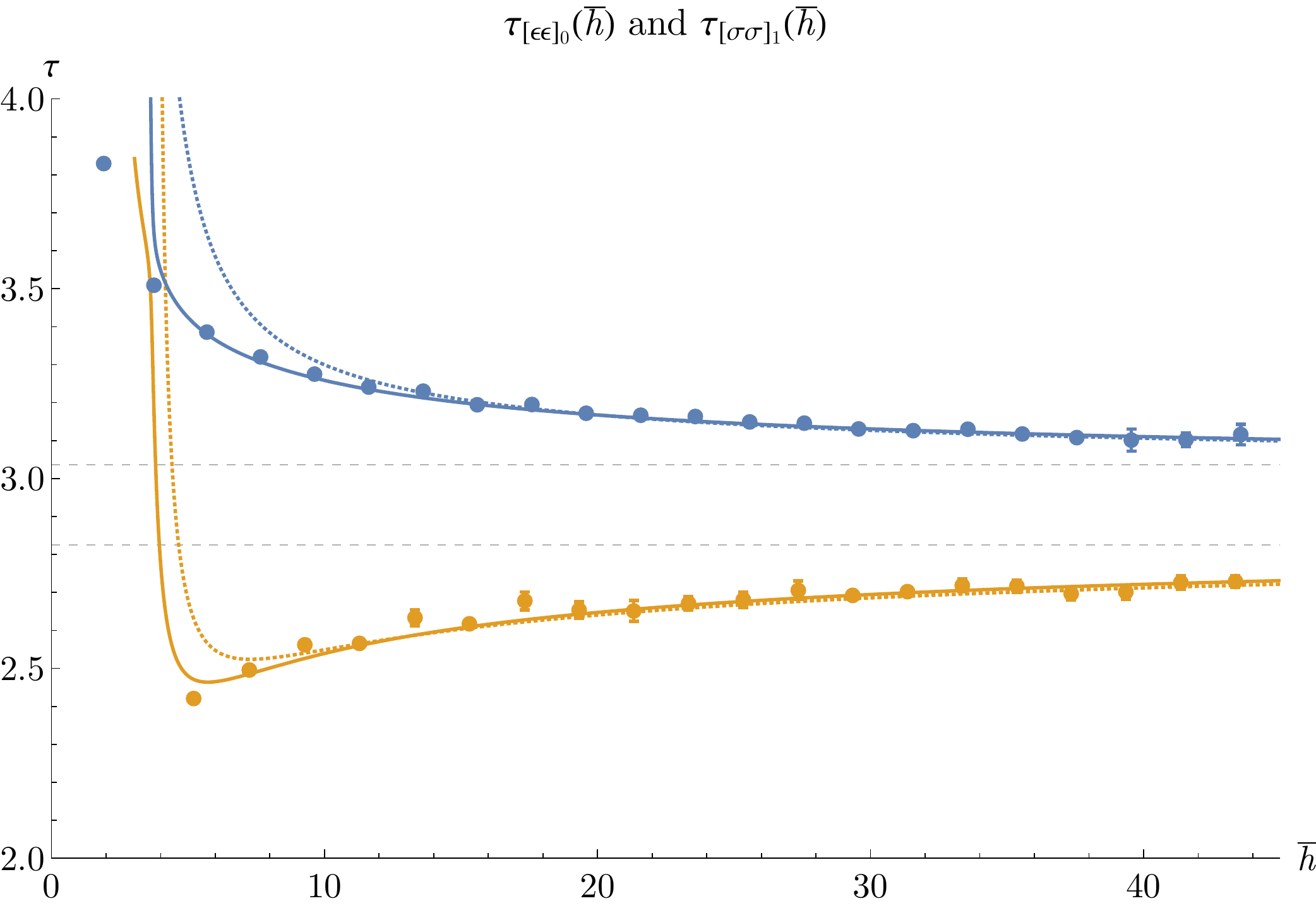}
\end{center}
\caption{Comparison between numerical data and analytical predictions for $\tau_{[\s\s]_1}$ (blue) and $\tau_{[\e\e]_0}$ (orange). Solid lines correspond to $\bar y_0=0.1$, and dotted lines correspond to $\bar y_0=0.02$. The orange curve ramps up sharply (moving from right to left) near $\bar h\approx 3.4$ because $M(\bar y_0,\bar h)$ becomes degenerate there. This coincides with the lower end of the family $[\e\e]_0$.}
\label{fig:tauEpsEps0AndSigSig1}
\end{figure}

\begin{figure}[ht!]
\begin{center}
\includegraphics[width=0.9\textwidth]{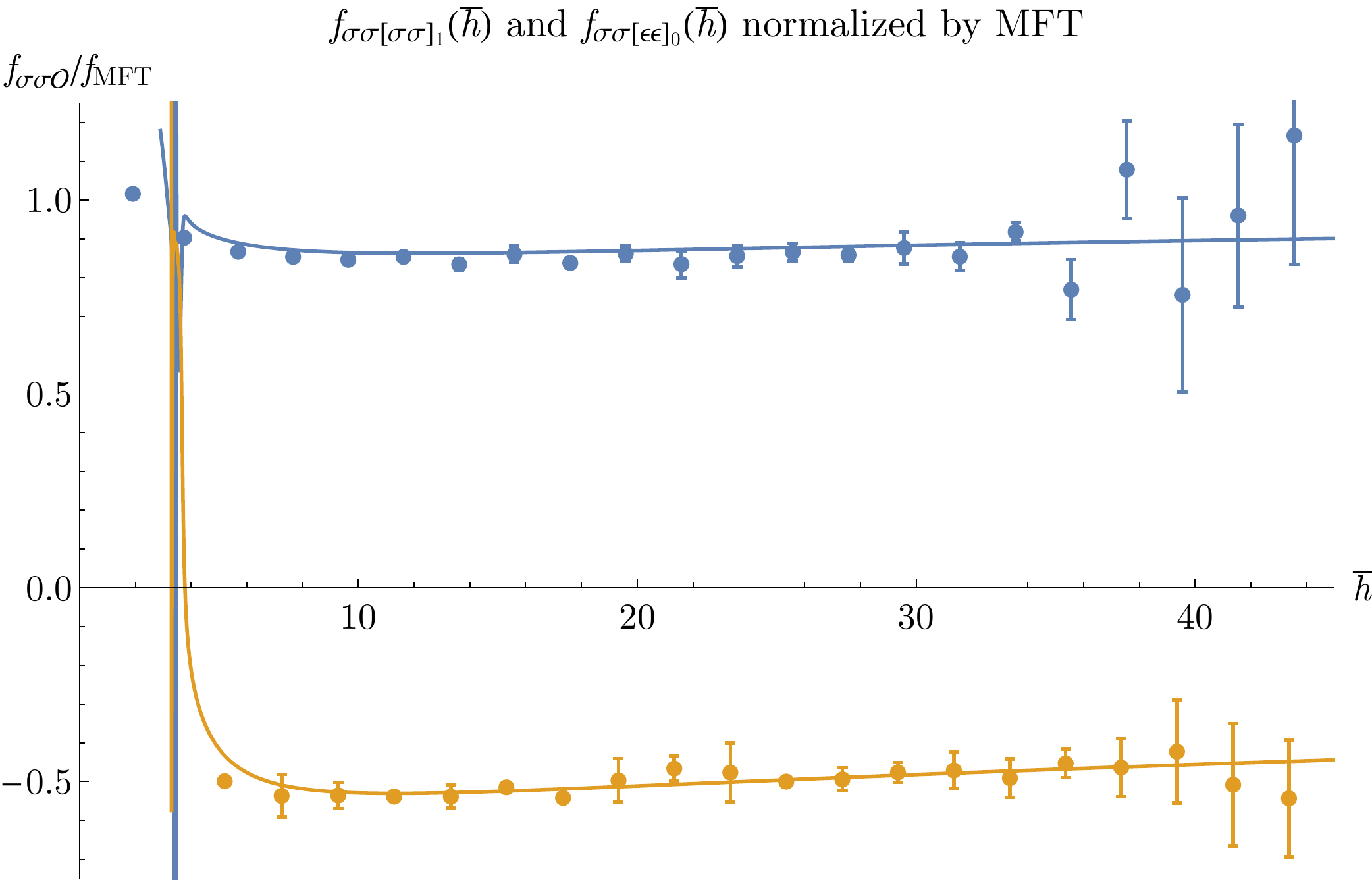}
\end{center}
\caption{Comparison between numerical data and analytical predictions for $f_{\s\s[\s\s]_1}$ (blue) and $f_{\s\s[\e\e]_0}$ (orange), both divided by the Mean Field Theory coefficient $f_\mathrm{MFT} = (2C^{(1)}_{-2h_\s}(2h_\s,\bar h))^{1/2}$. We fix the signs of $[\s\s]_n$ and $[\e\e]_n$ so that $f_{\s\s[\s\s]_n}$ and $f_{\e\e[\e\e]_n}$ are positive. With these conventions, $f_{\s\s[\e\e]_0}$ is negative.}
\label{fig:fSigSigSigSig1AndfSigSigEpsEps0}
\end{figure}

\begin{figure}[ht!]
\begin{center}
\includegraphics[width=0.9\textwidth]{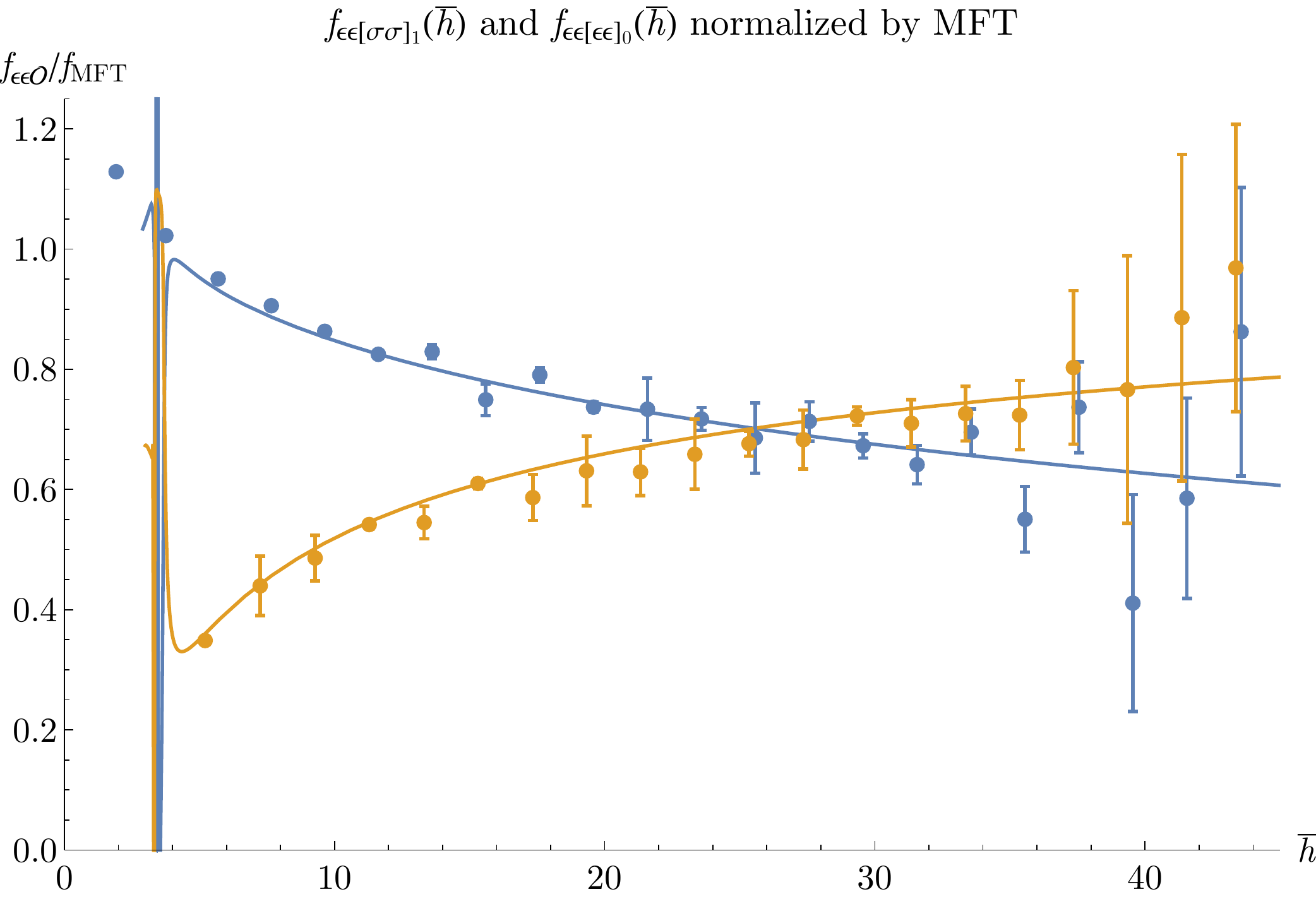}
\end{center}
\caption{Comparison between numerical data and analytical predictions for $f_{\e\e[\s\s]_1}$ (blue) and $f_{\e\e[\e\e]_0}$ (orange), both divided by the Mean Field Theory coefficient $f_\mathrm{MFT} = (2C^{(0)}_{-2h_\e}(2h_\e,\bar h))^{1/2}$. Note that $f_{\e\e[\s\s]_1}$ is larger than $f_{\e\e[\e\e]_0}$ for spins $\ell \lesssim 26$.}
\label{fig:fEpsEpsSigSig1AndfEpsEpsEpsEps0}
\end{figure}

 We compare analytics to numerics in figures~\ref{fig:tauEpsEps0AndSigSig1}, \ref{fig:fSigSigSigSig1AndfSigSigEpsEps0}, and~\ref{fig:fEpsEpsSigSig1AndfEpsEpsEpsEps0}. In figure~\ref{fig:tauEpsEps0AndSigSig1}, we show two sets of curves: the solid lines correspond to $\bar y_0=0.1$, and the dotted lines correspond to $\bar y_0 = 0.02$. As expected, the smaller value of $\bar y_0$ introduces errors that behave approximately like $\log^4 \bar y \bar h^{-4h_\s}$. The value $\bar y_0=0.1$ gives beautiful agreement with numerics for all spins $\ell\gtrsim 2$, so we take $\bar y_0=0.1$ in the remaining plots.
 
The results show several interesting features.  Firstly, we have correctly modeled the large mixing between the two families.  For example, the fact that $f_{\e\e[\s\s]_1}$ is larger than $f_{\e\e[\e\e]_0}$ for $\ell \lesssim 26$ is reproduced nicely.

We also find that $M(\bar y_0, \bar h)$ ceases to be positive-definite at $\bar h\approx 3.4$. This suggests that we cannot continue one of the twist families below this value. Indeed, in the numerical data, the family $[\e\e]_0$ ends at spin 4, which is the lowest spin such that $\bar h > 3.4$. It is surprising that one can predict such a detailed fact about the low-spin spectrum using the first few terms in an asymptotic expansion at high spin. It may be a happy coincidence. Zeros in the determinant of $M(\bar y_0,\bar h)$ are responsible for the rapid oscillations and poles at the leftmost edges of figures~\ref{fig:tauEpsEps0AndSigSig1}, \ref{fig:fSigSigSigSig1AndfSigSigEpsEps0}, and~\ref{fig:fEpsEpsSigSig1AndfEpsEpsEpsEps0}.

\section{Tying the knot}
\label{sec:knot}

\subsection{Where's the magic?}

By matching Casimir-singular terms on one side of the crossing equation to asymptotic large-$\bar h$ expansions on the other, we can systematically solve the crossing equations order-by-order in $y,\bar y$. In particular, we can reproduce a conformal block on one side with a particular large-$\bar h$ expansion on the other side. Our techniques for summing over twist families remove the difficulties associated with accumulation points in twist space.\footnote{See~\cite{Alday:2016njk,Alday:2016jfr} for another approach to this problem.} If this order-by-order solution to crossing is systematic, where are the nontrivial constraints on the spectrum?

Note that the asymptotic large-$\bar h$ expansion misses terms that are Casimir-regular in both channels. That is, terms that are both Casimir-regular in $y$ and Casimir-regular in $\bar y$.  If we write the crossing equation as
\be
\label{eq:writethecrossingeqas}
y^{-2h_\s}\sum_{\cO} f_{\s\s\cO}^2 G_{h_\cO,\bar h_\cO}(z,1-\bar z) &= z\leftrightarrow \bar z,
\ee
then these are terms of the form $y^m \bar y^n \log^p y \log^q \bar y$ with $p,q\leq 1$. We call such terms ``biregular."

We have already seen examples of biregular terms in computations: for example, the $\bar y^{2h_\s} \log \bar y$ and $\bar y^{2h_\s} \log y$ terms in the sum over $[\s\s]_0$ in (\ref{eq:currentsumsigsigsigsig}) are bi-regular, as we can see by multiplying by $\bar y^{-2h_\s}$ as on the right-hand side of (\ref{eq:writethecrossingeqas}). These are certainly nonzero, but they map to zero in the large-$\bar h$ expansion in either channel because $S_a(\bar h)$ has a double zero at $a=0$.

It is somewhat subtle to define the biregular part of a correlator separately from the Casimir-singular part. For example, $y^\de$ is Casimir-singular, but its expansion in small $\de$ contains nonzero Casimir-regular terms ($1$ and $\de \log y$). Indeed, no individual term in the sum (\ref{eq:writethecrossingeqas}) is biregular. However, biregular terms can appear when we evaluate the sum by expanding in the anomalous dimensions of double-twist operators.

To make sense of this, we propose the following prescription. Let us define an ``asymptotic solution" to crossing symmetry as a set of CFT data where the dimensions and OPE coefficients of multi-twist operators have the correct asymptotic large-$\bar h$ behavior to reproduce all Casimir-singular terms on the other side of the crossing equation. Given $S=\{f_{\s\s\cO},\De_\cO,\ell_\cO\}$, define the difference
\be
F_S(y,\bar y) &\equiv y^{-2h_\s} \sum_\cO f_{\s\s\cO}^2 G_{h_\cO,\bar h_\cO}(z,1-\bar z) - (y\leftrightarrow \bar y).
\ee
\begin{claim}
If $S$ is an asymptotic solution to crossing, then the ``biregular limit"
\be
\label{eq:magiclimit}
L_S=\lim_{y\to 0} \left.\p{\pdr{}{\log y}-\pdr{}{\log \bar y}} F_S(y,\bar y)\right|_{\bar y = y}
\ee
is finite. Furthermore if $S$ is a true (not just asymptotic) solution to crossing, then $L_S=0$.
\end{claim}
One can define similar biregular limits to extract biregular terms of the form $y^m \bar y^n \log^p y \log^q \bar y$ with $p,q\leq 1$. Demanding that biregular terms are crossing-symmetric gives nontrivial constraints on the spectrum.

\subsection{Constraints from low-twist operators}
\label{sec:biregularconstraints}

This suggests an interesting way to derive approximate constraints on the data of the 3d Ising CFT\@. From our work in sections~\ref{sec:asymptoticsappliedtoising} and~\ref{sec:eezerosigsigone}, we have approximate expressions for OPE coefficients and dimensions of the twist families $[\s\s]_0$, $[\s\s]_1$, $[\e\e]_0$, and $[\s\e]_0$ in terms of a finite set of initial data, namely $\{\De_\s, \De_\e, f_{\s\s\e}, f_{\e\e\e}, c_T\}$. By plugging these expressions back into the correlator and demanding that biregular limits vanish, we obtain constraints on the initial data.\footnote{In functional programming, defining a data structure in terms of itself is known as ``tying the knot" (\href{https://wiki.haskell.org/Tying_the_Knot}{https://wiki.haskell.org/Tying\_the\_Knot}).
}

Because we have not found exact asymptotic solutions to crossing, we must approximate the limits $L_S$ in some way.  We also do not have analytic approximations for the lowest spin members of the families $[\e\e]_0$ and $[\s\s]_1$, so we will restrict ourselves to limits involving $[\s\s]_0$ and $[\s\e]_0$.

In our expressions (\ref{eq:sigsigfit0}) and (\ref{eq:sigsigfit1}) for the dimensions and OPE coefficients of the $[\s\s]_0$ family, we treated the $\e$ and $T$ operators exactly. The biregular terms are approximately given by the $\log y, \log\bar y$ terms from expanding in small anomalous dimensions of the remaining operators $[\s\s]_{0,\ell\geq 4}$,  
\be
L_S&\approx 
2\p{\a_0^\mathrm{even}\left[\l_{\s\s[\s\s]_0}^2, \de_{[\s\s]_0}\right](2h_\s+4) -
\b_0^\mathrm{even}\left[\l_{\s\s[\s\s]_0}^2 \de_{[\s\s]_0}, \de_{[\s\s]_0}\right](2h_\s+4)},
\ee
where $\a_k[p,\de](\bar h_0)$ and $\b_k[p,\de](\bar h_0)$ are defined in appendix~\ref{app:nonintegerspacing} and $\l_{\s\s[\s\s]_0}$, $\de_{[\s\s]_0}$ are given by (\ref{eq:sigsigfit0}) and (\ref{eq:sigsigfit1}). Naively, these two quantities in parentheses have nothing to do with each other. However, plugging in the numerically-determined values of $\{\De_\s,\De_\e,f_{\s\s\e},f_{\e\e\e},c_T\}$, we find that they match to one part in $10^{-3}$,
\be
\a_0^\mathrm{even}\left[\l_{\s\s[\s\s]_0}^2, \de_{[\s\s]_0}\right](2h_\s+4) &\approx 1.92084,
\nn\\
\b_0^\mathrm{even}\left[\l_{\s\s[\s\s]_0}^2 \de_{[\s\s]_0}, \de_{[\s\s]_0}\right](2h_\s+4) &\approx 1.92280.
\ee

Similarly, by demanding that the leading biregular terms cancel in the sums over $[\s\s]_0$ and $[\s\e]_0$ in $\<\s\s\e\e\>$, we find the conditions
\be
L'_S &\equiv  A_1^{\s\s} - A_1^{\s\e} = 0,\nn\\
L''_S & \equiv A_2^{\s\s} - A_2^{\s\e}  = 0,
\ee
where
\be
A_1^{\s\s} &\equiv \bar{\b}_0^\mathrm{even}\left[\l_{\s\s[\s\s]_0} \l_{\e\e[\s\s]_0}, \de_{[\s\s]_0}\right](2h_\s+4)
    &\approx 6.89276, \nn\\
A_1^{\s\e} &\equiv
-\Gamma (2 h_{\epsilon\s}) \Gamma (1-h_{\e\s})^2 \bar{\alpha}_{-h_{\e\s}}\left[\l_{\s\e[\s\e]_0}^2,\de_{[\s\e]_0}\right](h_\s+h_\e+2)
&\approx 6.92499,
\ee
and
\be
A_2^{\s\s} &\equiv \bar \a^\mathrm{even}_0\left[\l_{\s\s[\s\s]_0} \l_{\e\e[\s\s]_0}, \de_{[\s\s]_0}\right](2h_\s+4)
&\approx 4.36510, \nn\\
A_2^{\s\e} &\equiv -\Gamma (2 h_{\epsilon\s}) \Gamma (1-h_{\e\s})^2 \bar{\alpha}_{-h_{\e\s}}\left[\l_{\s\e[\s\e]_0}^2 \de_{[\s\e]_0},\de_{[\s\e]_0}\right](h_\s+h_\e+2)
&\approx 4.35102.
\ee
(The regularized sums $\bar \a$ and $\bar \b$ are defined in appendix~\ref{app:regularizationofsums}.) On the right, we show the values of these quantities using the approximations in section~\ref{sec:asymptoticsappliedtoising} and the numerically-determined $\{\De_\s,\De_\e,f_{\s\s\e},f_{\e\e\e},c_T\}$. In all cases, the contributions to the limits $L_S,L'_S,L''_S$ cancel to reasonable precision.

The $L_S,L'_S, L''_S$ are interesting because their dominant contributions come from the lowest-twist operators in the theory, namely $[\s\s]_0$, $[\s\e]_0$, and indirectly $\s,\e,T$. This is based on our empirical observation that the contributions of these operators to the large-spin expansion give approximations that work well for {\it all\/} the operators in the twist families $[\s\s]_0$, $[\s\e]_0$. Thus, we can explore them without fully understanding the larger-twist spectrum.

By sampling values near the actual Ising point, we find that $L_S$ is much more sensitive to $c_T$ and $f_{\s\s\e}$ than the other quantities $\De_\s,\De_\e,f_{\e\e\e}$. The tangent plane to $L_S(c_T, f_{\s\s\e})$ at the Ising point is given by
\be
L_S &\approx -0.3999+1.599 c_T-1.061 f_{\s\s\e}.
\ee
Demanding that $L_S$ vanish gives a relationship between $c_T$ and $f_{\s\s\e}$.

\begin{figure}[ht!]
\begin{center}
\includegraphics[width=0.83\textwidth]{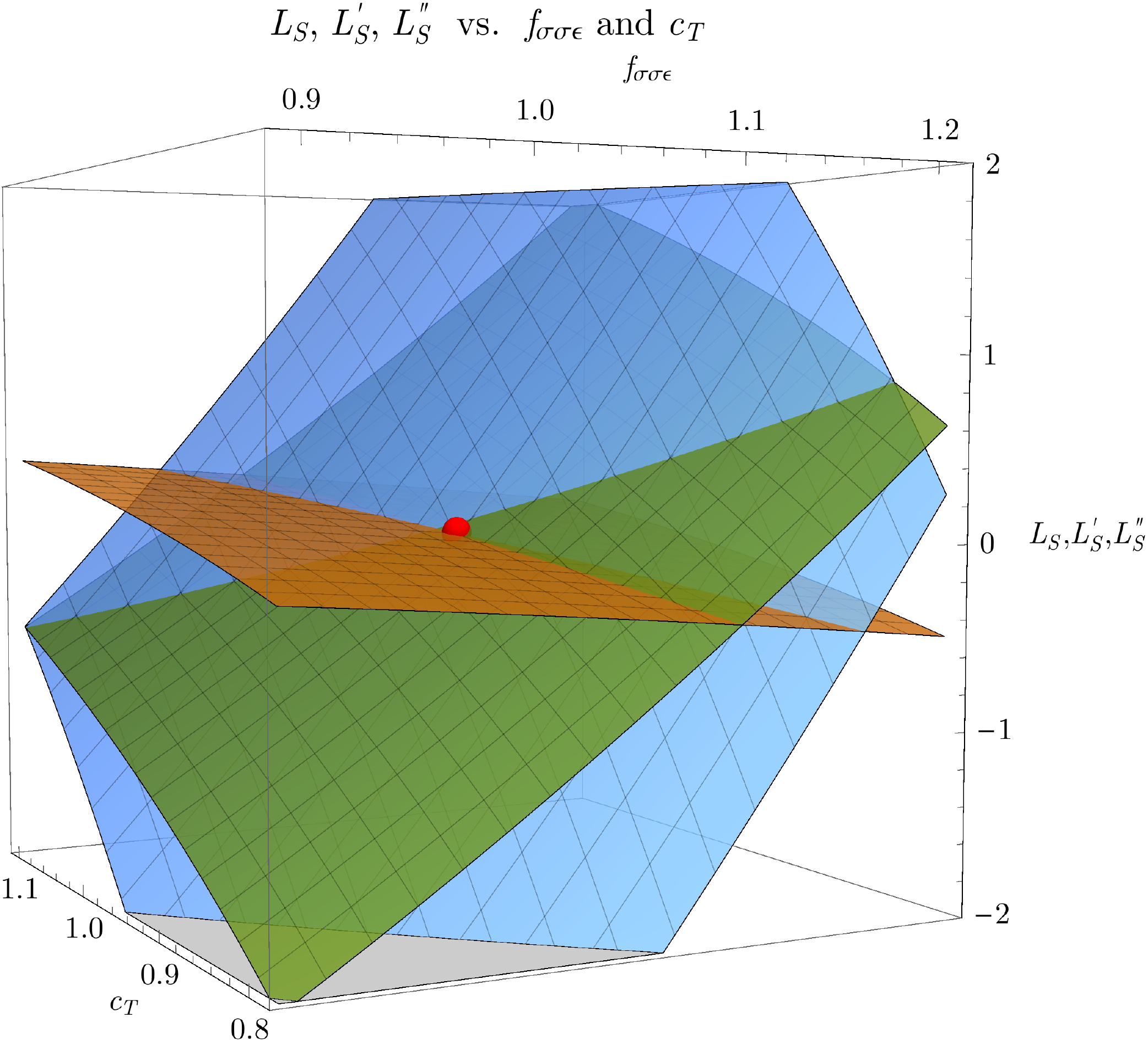}
\end{center}
\caption{The biregular limits $L_S$ (orange), $L'_S$ (blue), and $L''_S$ (green), plotted as a function of $f_{\s\s\e}$ and $c_T$, with $\De_\s,\De_\e,f_{\e\e\e}$ set to the values~(\ref{eq:isingisland}). The red sphere sits at the point expected for the 3d Ising CFT, $(f_{\s\s\e},c_T,L_S^*)=(1.0518539,0.946539,0)$.}
\label{fig:biregularLimits}
\end{figure}

In figure~\ref{fig:biregularLimits}, we plot all three limits $L_S,L_S',L_S''$ as a function of $c_T$ and $f_{\s\s\e}$, with the other quantities $\De_\s,\De_\e,f_{\e\e\e}$ held fixed at the values~(\ref{eq:isingisland}). The three quantities vanish nearly simultaneously at the correct values of $c_T$ and $f_{\s\s\e}$. Thus, requiring that $L_S,L_S',L_S''$ vanish gives a way to fix $c_T$ and $f_{\s\s\e}$ analytically in terms of $\De_\s,\De_\e,f_{\e\e\e}$, to accuracy $\sim 10^{-2}$-$10^{-3}$, using only the lightcone limit!

\section{Discussion}
\label{sec:discussion}

\subsection{Lessons for the numerical bootstrap}
\label{sec:lessonsfornumerics}

Traditional numerical bootstrap methods clearly probe the lightcone limit. This might explain why one must typically study a large number of derivatives around the Euclidean point $z=\bar z = \frac 1 2$ before the bounds saturate: many derivatives are needed to reach the lightcone limit, and the bounds may not saturate until the lightcone limit has been explored.

However, the Euclidean regime is also important. Because of the convergence properties of the conformal block expansion, the Euclidean regime effectively receives contributions from a small number of operators~\cite{Pappadopulo:2012jk,Hogervorst:2013sma,Rychkov:2015lca}, and one can make surprising progress by demanding that these contributions (almost) cancel among themselves~\cite{Gliozzi:2013ysa,Gliozzi:2014jsa,Gliozzi:2015qsa}.

This suggests the following hybrid analytical/numerical approach
\begin{enumerate}
\item First solve the lightcone limit analytically using the techniques in this work. The result will be an asymptotic expansion in $\bar h$, as a function of a small amount of initial data.

This step is likely easiest for theories with a relatively sparse spectrum in twist space. Since the spectrum becomes less sparse at high-twist, we expect mixing effects in the twist-Hamiltonian to become more important in this regime. We may not find an accurate picture of the high-twist spectrum without studying a large system of crossing equations (enough to build all the necessary multi-twist operators). However, figures~\ref{fig:evenspectrum} and~\ref{fig:oddspectrum} suggest that we may not need a perfectly accurate high-twist spectrum to make progress --- we just need {\it some\/} high-twist spectrum with approximately the right density in $\bar h$-space.
\label{step:one}

\item Choose some lower cutoff $\bar h \geq \bar h_0$, and compute the dimensions and OPE coefficients of multi-twist operators above this cutoff using the asymptotic expansion in step~\ref{step:one}. Larger $\bar h_0$ will mean more accurate expressions. However, smaller $\bar h_0$ will leave fewer operators to solve for in step~\ref{step:three}.
\label{step:two}

\item Plug the large-spin operators from step~\ref{step:two} back into the crossing equation and solve for the remaining operators in the Euclidean regime using traditional numerical bootstrap methods, the techniques of~\cite{Gliozzi:2013ysa,Gliozzi:2014jsa,Gliozzi:2015qsa},\footnote{The helpfulness of including higher spin $\Z_2$-odd operators in the ``severe truncation" method of~\cite{Gliozzi:2013ysa} has been observed previously~\cite{SlavaUnpublished2015}.} or some other method. We suspect that many fewer derivatives may be needed. It would also be interesting to see if this hybrid method reduces the need for  high-precision arithmetic.

\label{step:three}
\end{enumerate}

Unfortunately, this approach sacrifices the rigor of traditional numerical bootstrap methods because the large-$\bar h$ expansion is asymptotic. One must take $\bar h_0$ large enough that the results saturate. (Though working with larger $\bar h_0$ likely requires more derivatives.) It is encouraging that $\bar h_0\approx 4$ is good enough for most of the results in this work. Another disadvantage is that some theories might require a large amount of initial input to compute the large-spin spectrum. For example, in this work we used $\{\De_\s,\De_\e,f_{\s\s\e},f_{\e\e\e}, c_T\}$ to parameterize the large-spin spectrum of the 3d Ising CFT\@. We must scan over each of these parameters to explore the space of theories. In a larger system of crossing equations, we would have even more parameters.

On the other hand, the possibility of working with fewer derivatives, at lower precision, with larger systems of correlators, and perhaps without imposing unitarity (using the methods of~\cite{Gliozzi:2013ysa,Gliozzi:2014jsa,Gliozzi:2015qsa})\footnote{Not imposing unitarity could also help in studies of boundary and defect crossing equations, which in some cases haven't been formulated in a way that takes advantage of positivity (even in unitary theories)~\cite{Gliozzi:2015qsa,Liendo:2012hy,Gadde:2016fbj}.} makes this hybrid approach worth exploring.

\subsection{Moving towards analytics}

Although we have made progress in reverse-engineering a solution to crossing symmetry analytically, numerics were crucial throughout. Let us catalog the ways in which we used numerics and discuss whether/how they can be replaced with analytics.

\begin{itemize}
\item Because the large-$\bar h$ expansion is asymptotic, numerics were crucial in determining how many terms to keep in the expansion to get reasonable results. We could also see explicitly which operators were well-described by a truncated asymptotic expansion and which ones were not. For example, $[\s\s]_{0,\ell=4}$ fits well to the analytic predictions~(\ref{eq:sigsigfit0}, \ref{eq:sigsigfit1}), while $[\s\s]_{1,\ell=0}$ does not fit the prediction in figure~\ref{fig:tauEpsEps0AndSigSig1}. We used this information implicitly in several ways. For example, in section~\ref{sec:biregularconstraints}, we used that the analytic predictions~(\ref{eq:sigsigfit0}, \ref{eq:sigsigfit1}) fit well {\it all\/} the operators in the family $[\s\s]_0$. 
 
To understand these issues without numerics, it will be important to prove rigorous bounds on error terms in the large-$\bar h$ expansion. It would also be interesting to understand convergence properties of the twist expansion in a way analogous to the dimension expansion~\cite{Pappadopulo:2012jk,Hogervorst:2013sma,Rychkov:2015lca,Maldacena:2015iua}.

Ideally, perhaps there is a way to identify the correct representative of a given large-$\bar h$ equivalence class.
Consistency with causality and the chaos bound~\cite{Maldacena:2015waa} may be relevant, since it requires delicate cancellations between high-spin operators in a certain kinematic limit, see e.g.~\cite{Maldacena:2016hyu,Perlmutter:2016pkf}.\footnote{We thank Douglas Stanford for this suggestion.} It also implies bounds on dimensions of operators in the lowest twist-family~\cite{Hartman:2015lfa,Hartman:2016dxc,Hofman:2016awc,Hartman:2016lgu}.

\item We used numerics to discover that the contribution of other multi-twist families like $[TT]$ and $[\e T]$ to the four-point functions $\<\s\s\s\s\>$, $\<\s\s\e\e\>$, $\<\e\e\e\e\>$ is small. Consequently, we could ignore these families when diagonalizing the twist Hamiltonian for $[\s\s]_1$ and $[\e\e]_0$ in section~\ref{sec:eezerosigsigone}.

We might guess that $[\e T]_0$ and $[TT]_0$ should be unimportant because the anomalous dimension of $T$ is small, so only the leading term in the exponentiation of $\bar y^{h_T-2h_\s}$ matters. However, a better treatment of this issue would involve studying correlators with $T$ as an external operator in addition to $\s$ and $\e$. In fact, to get a full picture of the small-twist spectrum of the 3d Ising model, we should study correlators including all the operators in $[\s\s]_0$ as external operators. This will likely require new techniques, since the mixing matrices will be infinite-dimensional.

\item We also used numerics to help choose the value $\bar y_0$ at which to evaluate the twist Hamiltonian in section~\ref{sec:eezerosigsigone}. The results should become less sensitive to $\bar y_0$ when we study all the twist families that could contribute to $M(\bar y,\bar h)$. This includes additional double-twist families like $[\e T]_0$ and $[TT]_0$ discussed above, as well as higher-twist towers like $[\s\s]_2$ and $[\e\e]_1$. To completely recover exponentiation, we must also understand $n$-twist families with $n\geq 3$. Although this may be possible with four-point functions, in practice it might require studying higher-point functions, as discussed in section~\ref{sec:effectivesubspace}.

\item Although we parameterized most of the low-twist spectrum in terms of a small amount of initial data $\{\De_\s,\De_\e,f_{\s\s\e},f_{\e\e\e},c_T\}$, it would be difficult to fix this data in practice without already knowing the answer~(\ref{eq:isingisland}). The biregular limits in section~\ref{sec:knot} give constraints. It will be important to understand whether they can be solved systematically.

The Euclidean regime is also important and currently the best techniques for exploring it are numerical. Perhaps the hybrid approach suggested above can help. It may also be interesting to study how recent Mellin-space approaches to the bootstrap~\cite{Gopakumar:2016wkt,Gopakumar:2016cpb,Dey:2016mcs} interact with the results of this work.

\end{itemize}

\subsection{More future directions}

A central question is: why do the truncated large-$\bar h$ expansions for $[\s\s]_0$, $[\s\e]_0$, etc., work so well even at small $\bar h$? Perhaps our all-orders asymptotic solutions are close to an exact answer. 
Our work in section~\ref{sec:twisthamiltonian} suggests the following ansatz for the conformal block expansion:
\be
\oint_{-i\oo}^{i\oo} d\bar h\, \L(\bar h) \frac{\pi}{\tan\p{\pi(\bar h - H(\bar h))}} G_{H(\bar h), \bar h}(z,1-z) \L(\bar h)^T,
\ee
where $H(\bar h)$ is the twist-Hamiltonian and $\L(\bar h)$ is a matrix of OPE coefficients. We have shown how to compute the large-$\bar h$ asymptotics of
\be
\L(\bar h) \frac{\pi}{\tan\p{\pi(\bar h - H(\bar h))}} y^{H(\bar h)} \L(\bar h)^T &\sim \L(\bar h) y^{H(\bar h)} \L(\bar h)^T
\label{eq:thingwecomputeasymptoticsof}
\ee
using crossing symmetry. (Asymptotics as $\bar h\to \oo$ along the real axis are related to asymptotics as $\bar h\to\pm i\oo$ for the class of functions we consider.) However, perhaps one could compute the full function on the left-hand side of (\ref{eq:thingwecomputeasymptoticsof}) using a crossing kernel for $\SO(d,2)$ conformal blocks~\cite{CaronHuotFuture}. 
This would remove the difficulty of working with asymptotic expansions. 

One could then try to solve crossing symmetry via an iterative procedure:
\begin{enumerate}
\item Start with a few known operators like $\s$, $\e$, and $T$.
\item Compute $H(\bar h)$ and $\L(\bar h)$ and diagonalize $H(\bar h)$.
\item Plug the results back in to compute corrections to $H(\bar h)$ and $\L(\bar h)$ from multi-twist operators.
\item Repeat until the spectrum converges.
\end{enumerate}
It will be interesting to explore this program in the future.

While we have focused on multi-twist operators (in particular double-twist operators), it is also interesting to consider other types of operator families like logarithmic Regge trajectories in conformal gauge theories. Such trajectories can still be described using the techniques in this work, by writing $\log \ell=\pdr{}{\e}\ell^\e |_{\e=0}$.\footnote{This observation was also made in~\cite{Alday:2016njk}.} We expect that the techniques of section~\ref{sec:twisthamiltonian} give the right language for studying interesting phenomena like mixing between large-spin single- and double-trace operators in non-planar $\mathcal{N}=4$ SYM\@.

It would also be interesting to apply our techniques to large-$N$ theories. Summing up the effects of graviton exchange in the bulk is important for understanding the emergence of geometry in AdS/CFT\@. While Virasoro symmetry makes this relatively simple in 2d CFTs, it is a difficult task in $d>2$. Our all-orders results for large-spin operators may help make headway on this problem.

Finally, our new data for the 3d Ising CFT may have interesting applications to condensed matter and statistical physics. In~\cite{Rychkov:2016mrc}, we used the low-dimension operators in table~\ref{tab:lowestdim} to plot the Euclidean four-point function and check some inequalities from the lattice Ising model. (In this work, we can see explicitly some of the non-Gaussianity discussed in~\cite{Rychkov:2016mrc}, from the large mixing between $[\e\e]_0$ and $[\s\s]_1$.) It would be interesting if some of the new operator dimensions and OPE coefficients in this work could be checked with Monte-Carlo techniques, the $\e$-expansion, or experiment.

\section*{Acknowledgements}
I am grateful to Chris Beem, Simon Caron-Huot, Abhijit Gadde, Joanna Huey, Juan Maldacena, Eric Perlmutter, David Poland, Slava Rychkov, Douglas Stanford, and Balt van Rees for discussions. Thanks to Douglas Stanford and Slava Rychkov for comments on the draft.
I am supported by DOE grant DE-SC0009988, a William D. Loughlin Membership at the Institute for Advanced Study, and Simons Foundation grant 488657 (Simons Collaboration on the Non-perturbative Bootstrap).
The computations in this paper were run on the Hyperion computing cluster supported by the School of Natural Sciences Computing Staff at the Institute for Advanced Study.

\pagebreak

\appendix

\section{Numerical calculation of the 3d Ising spectrum}
\label{app:numerics}

\subsection{More details on the numerics}
\label{app:extremalfunctional}

As explained in section~\ref{sec:numericalpicture}, our strategy is to compute a partial spectrum $\cS_N(p)$ for several different points $p$ on the boundary of the allowed region $\cA_N$, and then choose the operators that are stable under varying $p$. To get to the boundary of $\cA_N$, we can minimize or maximize any quantity. It is not actually necessary that $(\De_\s,\De_\e,f_{\s\s\e},f_{\e\e\e})$ themselves lie on the boundary of (\ref{eq:isingisland}), as long as they don't lie outside (\ref{eq:isingisland}) and some other quantity is minimal or maximal.

Extremizing an OPE coefficient is technically easier than extremizing an operator dimension because it can be done in a single optimization step.\footnote{See~\cite{El-Showk:2016mxr} for recent progress on speeding up operator dimension extremization.} In~\cite{Komargodski:2016auf}, we chose to maximize $f_{\s\s\e}$. In this work, we minimize $c_T$ (equivalently maximize $f_{\s\s T}$) as in~\cite{El-Showk:2014dwa}. The answers are essentially identical.  We describe how to extract a partial spectrum by extremizing an OPE coefficient in section~\ref{app:extractingspectra}.

To get a sense of the errors in the extremal functional method, we must choose a variety of points on the boundary of $\cA_N$. The space of CFT data is infinite-dimensional, so random sampling is imposible. We must simply try different things, and hope some results will be invariant.

Let $\tan \th_{\s\e} = f_{\e\e\e}/f_{\s\s\e}$. We minimize $c_T$ at the following 20 points in $(\De_\s,\De_\e,\th_{\s\e})$-space:
\be
(\De_\s,\De_\e,\th_{\s\e}) &\in\nn\\
&
\begin{array}{llll}
\{ \!\!\!\!\!\!\!&  (0.51814898,1.4126250,0.9692610),
&  (0.51814937,1.4126306,0.9692662),
\\&(0.51814930,1.4126283,0.9692632),
&  (0.51814807,1.4126156,0.9692547),
\\&(0.51814978,1.4126348,0.9692687),
&  (0.51814893,1.4126251,0.9692611),
\\&(0.51814927,1.4126283,0.9692631),
&  (0.51814881,1.4126251,0.9692623),
\\&(0.51814835,1.4126192,0.9692574),
&  (0.51814880,1.4126253,0.9692632),
\\&(0.51814924,1.4126285,0.9692643),
&  (0.51814951,1.4126320,0.9692664),
\\&(0.51814865,1.4126215,0.9692583),
&  (0.51814945,1.4126323,0.9692678),
\\&(0.51814791,1.4126142,0.9692544),
&  (0.51814954,1.4126313,0.9692654),
\\&(0.51814819,1.4126180,0.9692574),
&  (0.51814856,1.4126210,0.9692591),
\\&(0.51814828,1.4126191,0.9692586),
&  (0.51814931,1.4126302,0.9692658)\  \}
\end{array}
\ee
We do not specify the norm $\sqrt{f_{\e\e\e}^2 + f_{\s\s\e}^2}$ --- it is an output of the spectrum computation, together with a list of other operators.

We assume that $\s$ and $\e$ are the only relevant scalars in the theory.  In addition, we impose gaps in the $\Z_2$-even scalar sector (above $\e$) and spin-2 sector (above $T_{\mu\nu}$) of the following form
\be
\De_\mathrm{\ell=0}^\textrm{min} \in \{3,3.5\}, \qquad
\De_\mathrm{\ell=2}^\textrm{min} \in \{3,4,5\}.
\ee
When we impose a gap in the spin-2 sector, we also impose the stress-tensor Ward identity $f_{\s\s T}/f_{\e\e T} = \De_\s/\De_\e$.

The resulting spectra are mostly independent of the gaps, with one exception: in the extremal functional method, spurious operators often appear at the gaps. Some examples are the higher spin operators at the unitarity bound discussed in section~\ref{sec:sigsig0section}. Similarly, the spectra computed using the above assumptions often (not always) have scalars of dimensions $3$ or $3.5$ or spin-2 operators of dimension $4$ or $5$. (In addition, there are occasionally $\Z_2$-odd scalars with dimension $3$ due to the gap in that sector.) By varying the gaps, these operators become ``unstable" in the sense that their dimensions depend on the boundary point $p$. Hence, in practice we can ignore them compared to the stable operators. Their OPE coefficients are usually quite small, so they don't affect the crossing equations much. We have removed these spurious operators by hand in figures~\ref{fig:evenspectrum} and~\ref{fig:oddspectrum}.

We minimize $c_T$ by setting up a semidefinite program and solving it with the solver \texttt{SDPB}~\cite{Simmons-Duffin:2015qma}. We work with $\Lambda=43$, corresponding to 1265 derivatives of the crossing equations, and our \texttt{SDPB} parameters are given in table~\ref{tab:parameters}.

\begin{table}[!ht]
\centering
{\small
\begin{tabular}{|l|c|c|c|c|}
\hline
$\Lambda$& 43 \\
$\kappa$ & 40 \\
spins & $S_{43}$\\
\texttt{precision} & 960\\
\texttt{findPrimalFeasible} & False\\
\texttt{findDualFeasible} & False\\
\texttt{detectPrimalFeasibleJump} & False\\
\texttt{detectDualFeasibleJump} & False\\
\texttt{dualityGapThreshold} & $10^{-60}$ \\
\texttt{primalErrorThreshold} & $10^{-75}$ \\
\texttt{dualErrorThreshold} & $10^{-75}$ \\
\texttt{initialMatrixScalePrimal} & $10^{60}$\\
\texttt{initialMatrixScaleDual} & $10^{60}$\\
\texttt{feasibleCenteringParameter} & 0.1 \\
\texttt{infeasibleCenteringParameter} & 0.3\\
\texttt{stepLengthReduction} & 0.7\\
\texttt{choleskyStabilizeThreshold} & $10^{-200}$ \\
\texttt{maxComplementarity} & $10^{200}$\\
\hline
\end{tabular}
}
\caption{\texttt{SDPB} parameters for the computations in this work. $S_{43}$ is given by $\{0,\dots,64\}\cup\{67, 68, 71, 72, 75, 76, 79, 80, 83, 84, 87, 88\}$.}
\label{tab:parameters}
\end{table}

\subsection{Extracting spectra and OPE coefficients from SDPs}
\label{app:extractingspectra}

A system of $k$ crossing equations for CFT four-point functions can be put in the form
\be
\label{crossingexample}
0 = \sum_{\Delta,\ell,R} \vec{\lambda}_{\Delta,\ell,R}^T F^i_{\Delta,\ell,R}(z,\bar z) \vec{\lambda}_{\Delta,\ell,R},\qquad (i=1,\dots,k).
\ee
Here, $\Delta,\ell,R$ run over the dimension, spin, and symmetry representations of exchanged operators. Each $F^i_{\Delta,\ell,R}(z,\bar z)$ is a matrix whose entries are combinations of conformal blocks, and the $\vec{\lambda}_{\Delta,\ell,R}$ are vectors of OPE coefficients. For example, for a four-point function of identical scalars $\<\f\f\f\f\>$, $k=1$ and $F^1_{\De,\ell}(z,\bar z)$ is a $1\x1$ matrix with entry $v^{\De_\f} g_{\De,\ell}(u,v)-u^{\De_\f} g_{\De,\ell}(v,u)$.

For simplicity, we first consider minimizing the $0$ function with respect to the constraints (\ref{crossingexample}). Thus, we are simply asking when it is possible to find real $\vec \l_{\De,\ell,R}$ such that (\ref{crossingexample}) is true.  (We comment about the case where we minimize something nontrivial later.)
 In~\cite{Kos:2014bka,Simmons-Duffin:2015qma}, it was shown how to reformulate this question as a Polynomial Matrix Program (a special type of semidefinite program) of the following form:\footnote{We follow the notation of~\cite{Simmons-Duffin:2015qma} for the rest of this appendix.}  Find $y\in \R^N$ such that
\be
\label{pmp}
\sum_{n=0}^N \alpha_n M^n_j(x) \succeq 0\qquad{\rm for\ all\ }x\in[0,\oo),\ j=1,\dots,J.
\ee
where $\a=(1,y)\in \R^{N+1}$.
The notation ``$M\succeq 0$" means ``$M$ is positive semidefinite."  The $M^n_j(x)$ are matrices with polynomial entries
\be
\label{polynomialmatrices}
M_j^n(x) = 
\begin{pmatrix}
P_{j,11}^n(x) & \cdots & P^n_{j,1m_j}(x) \\
\vdots & \ddots & \vdots \\
P_{j,m_j1}^n(x) & \cdots & P^n_{j,m_jm_j}(x)
\end{pmatrix}.
\ee
In our case, the dual objective function $b$ vanishes because we are minimizing the 0 function.

The relation between the matrices $M_j^n(x)$ and the functions $F^i_{\Delta,\ell,R}(z,\bar z)$ entering the crossing equation is as follows.  Firstly, $j$ corresponds to tuples $(\ell,R)$.  Secondly, $n$ corresponds to tuples $(i,a,b)$, where $i$ labels crossing equations and $a,b$ are positive integers labeling derivatives. We have
\be
\label{matrixintermsofblocks}
M_j^n(x) &\approx \chi_j(\Delta(x))^{-1} \left.\partial_z^a \partial_{\bar z}^b F_{\Delta(x), \ell, R}(z,\bar z)\right|_{z=\bar z = \frac 1 2},\cr
\Delta(x) &= \Delta_{j,{\rm min}} + x,
\ee
where $\Delta_{j,{\rm min}}$ is the minimum dimension for $j=(\ell,R)$ (e.g., the unitarity bound for an operator with spin $\ell$ and representation $R$).  $\chi_j(\Delta)$ is a positive function of $\Delta$, written explicitly in~\cite{Simmons-Duffin:2015qma}. The accuracy of the approximation (\ref{matrixintermsofblocks}) can be made arbitrarily good by increasing the polynomial degree of $M_j^n(x)$.

Consequently, the first $N+1$ derivatives of the crossing equations (\ref{crossingexample}) are (approximately) equivalent to
\be
\label{matrixcrossingeq}
0 =\sum_{j}\sum_{\tau} \vec{v}_{j,\tau}^T M_j^n(\tau) \vec{v}_{j,\tau},
\ee
where 
\be
\label{translation}
\vec{\lambda}_{\Delta,\ell,R} &= \chi_j(\Delta_{j,{\rm min}}+\tau) \vec{v}_{j,\tau},\cr
\Delta &= \Delta_{j,{\rm min}}+\tau.
\ee
For each $j$, the $\tau$-sum in (\ref{matrixcrossingeq}) ranges over a discrete set of nonnegative real numbers.
Equation (\ref{matrixcrossingeq}) can be rewritten as
\be
\label{matrixcrossingeqtwo}
0 = \sum_{j}\sum_\tau {\rm Tr}(V_{j,\tau} M_j^n(\tau)),\qquad V_{j,\tau}\succeq 0,
\ee
where $V_{j,\tau}$ is a sum of outer products of $\vec v$'s and is thus positive semidefinite.
The vectors $\vec{v}_{j,\tau}$ can be recovered from $V_{j,\tau}$ via Cholesky decomposition.\footnote{The matrix $V_{j,\tau}$ typically has low rank, which means that numerically it may have very small negative eigenvalues.  Thus, instead of using a Cholesky decomposition, we simply compute its eigenvectors and throw out those with very small eigenvalues.  For computations in this work, it suffices to keep only the first eigenvector. We expect low rank matrices because a higher-rank matrix would mean that the determinant of the functional has a higher-order zero at a fixed $\Delta$, which is non-generic. An exception occurs if an operator is isolated in $\De$-space, which is why one can obtain stronger constraints by imposing that the matrix associated to the $\e$ operator is rank-$1$ as in~\cite{Kos:2016ysd}. We thank Slava Rychkov and Alessandro Vichi for discussions on this point.}

Thus, if we can find $\tau$'s and $V_{j,\tau}$'s such that (\ref{matrixcrossingeqtwo}) holds, then (\ref{translation}) gives a set of dimensions and OPE coefficients that solve $N+1$ derivatives of the crossing equations and are consistent with unitarity.

It is simple to find the appropriate $\tau$'s.  The solver {\tt SDPB} returns a vector $y\in \R^N$ which can be assembled into $\alpha=(1,y)\in \R^{N+1}$ satisfying the constraints (\ref{matrixcrossingeq}).  Taking the inner product of (\ref{matrixcrossingeqtwo}) with $\alpha$, we find 
\be
\label{innerprodwithalpha}
0 = \sum_{j,\tau} {\rm Tr}(V_{j,\tau}(\alpha\cdot M_j(\tau))).
\ee
This implies that each term in (\ref{innerprodwithalpha}) vanishes individually, since each term is nonnegative.  However, this is only possible if $\alpha\cdot M_j(\tau)$ is a degenerate $m_j\times m_j$ matrix.  Thus, $\tau$ must be a nonnegative zero of $\det(\alpha\cdot M_j(x))$.

The function $\det(\alpha \cdot M_j(x))$ is constrained to be positive for $x\in[0,\infty)$.  Thus, its positive zeros for must be double zeros.  In numerical computations, it never actually attains the value zero, but instead dips very close to to the $x$-axis.  Thus, it's more convenient to compute the $\tau$'s as local minima of $\det(\alpha \cdot M_j(x))$, which can be computed as zeros of its derivative (together with a possible zero at $x=0$, which must be checked separately).

The matrices $V_{j,\tau}$ can be obtained by solving the linear algebra problem (\ref{matrixcrossingeqtwo}).  However, they are also already encoded in the primal solution computed by {\tt SDPB}.  
Let $d_j=\max_{n}[\deg(M_j^n(x))]$.  The primal solution is a vector $u\in \R^P$ where $P$ counts the number of tuples $(j,r,s,k)$ with $1\leq r\leq s\leq m_j$, $k=0,\dots,d_j$.\footnote{We use the notation $u$ instead of~\cite{Simmons-Duffin:2015qma}'s $x$ to avoid confusion with the sample points $x_k^{(j)}$.}  We can assemble $u$ into symmetric matrices
\be
\label{umatrices}
U_{j,k} &= \chi_j(x_k^{(j)})\sum_{r,s} u_{(j,r,s,k)} E^{rs},
\ee
where $E^{rs}$ is a symmetrized unit matrix with components
\be
\label{symmetrizedunitmatrix}
(E^{rs})_{ab} \equiv \half(\delta^r_a \delta^s_b + \delta^r_b\delta^s_a).
\ee

The primal solution satisfies the constraint\footnote{This is a combination of the equations $B^T u=0$ and $c\cdot u=0$ in~\cite{Simmons-Duffin:2015qma}.  The equation $c\cdot u=0$ follows from equality of the primal and dual objective functions at a solution of the SDP.}
\be
\label{primalconstraint}
0 = \sum_{j,k}{\rm Tr}(U_{j,k} M^n_j(x_k^{(j)})),
\ee
where the $x_k^{(j)}$, $k=0,\dots,d_j$ are ``sample points" provided as input to {\tt SDPB}.  This is almost the desired result (\ref{matrixcrossingeqtwo}), except that $M_j^n(x)$ is evaluated at the sample points $x_k^{(j)}$ instead of the $\tau$'s.  However, since since $M_j^n(x)$ is a polynomial of degree $d_j$, its value at $\tau$ is a linear combination of its value at the sample points,
\be
\label{lagrangeinterpolation}
M_j^n(\tau) = \sum_{k=0}^{d_j} L(\tau,x_k^{(j)}) M_j^n(x_k^{(j)}),
\ee
where $L(\tau,x_k^{(j)})$ are Lagrange interpolation coefficients.  Thus, we should solve the linear algebra problem
\be
\label{vfromu}
U_{j,k} = \sum_\tau  V_{j,\tau} L(\tau,x_k^{(j)}).
\ee
This is usually an overdetermined system, since the number of positive real zeros of $\det(\alpha\cdot M_j(x))$ is typically smaller than the number of sample points $d_j+1$. We solve it with a least-squares fit, using the singular value decomposition of $L(\tau,x_k^{(j)})$.  The validity of the fit can be verified by checking the crossing equation (\ref{matrixcrossingeq}).

We are sometimes interested in solving a program with a nonzero objective function $b\in \R^N$.  When this objective function is a linear combination of contributions of operators to the crossing equation, we must simply include those operators in the resulting spectrum. Their OPE coefficients should be multiplied by the square root of the absolute value of the objective function at the solution.

An implementation of the algorithm described in this section is available at \\ \href{https://gitlab.com/bootstrapcollaboration/spectrum-extraction}{\tt https://gitlab.com/bootstrapcollaboration/spectrum-extraction}.

\subsection{Several operators in the 3d Ising CFT}
\label{app:severaloperators}

In this section, we list dimensions and OPE coefficients of 112 stable operators obtained from the calculation described in section~\ref{app:extremalfunctional} (and plotted in figures~\ref{fig:evenspectrum} and~\ref{fig:oddspectrum}).
Most of the stable operators fall into the families $[\s\s]_0$ (table~\ref{tab:sigsig0}), $[\e\e]_0$ (table~\ref{tab:epseps0}), $[\s\s]_1$ (table~\ref{tab:sigsig1}), and $[\s\e]_0$ (table~\ref{tab:sigeps0}).\footnote{A unique assignment of an operator to a twist family is not actually well-defined, due to the fact that the large-spin expansion is asymptotic, and the possibility of mixing between twist families. However, for all the operators here, there is only one reasonable choice.} The rest include $\s$ and $\e$, and a few low-dimension stable operators that are not obviously part of any twist family (table~\ref{tab:sporadic}). For convenience, we also list all stable operators with dimension $\De\leq 8$ in table~\ref{tab:lowestdim}.

We estimate errors as standard deviations in our sample set of $60$ partial spectra. It is important to remember that these error estimates are non-rigorous (in contrast to the bounds on $\De_\s,\De_\e,f_{\s\s\e}$, and $f_{\e\e\e}$, in (\ref{eq:isingisland}), which are rigorous).
In fact, the tables show that this method of assigning errors is imperfect. In table~\ref{tab:sigeps0}, for example, the precision of the OPE coefficients $f_{\s\e\cO}$ varies significantly at large $\ell$.  For instance, the error for the OPE coefficient of $[\s\e]_{0,\ell=27}$ is $0.3\%$, while the error for $[\s\e]_{0,\ell=28}$ is $2\%$. It is surprising that these should be so different. Perhaps a wider scan of the boundary of the allowed region $\cA_N$ would equalize the errors somewhat. Regardless, the reader should take the error estimates with a grain of salt.

Because we have chosen different conventions for conformal blocks, our OPE coefficients are normalized differently from those in~\cite{Komargodski:2016auf}.  Specifically, the leading terms in the conformal block expansion are given by
\be
f_{12\cO} f_{43\cO} G_{h=\frac{\De-\ell}{2},\bar h = \frac{\De+\ell}{2}}^{r=\frac{\De_{12}}{2},s=\frac{\De_{34}}{2}}(z,\bar z) &= f_{12\cO} f_{43\cO} z^{\frac{\De-\ell}{2}} \bar z^{\frac{\De+\ell}{2}} + \ldots, & (\textrm{here}),\nn\\
f_{12\cO} f_{34\cO} g_{\De,\ell}^{\De_{12},\De_{34}}(z,\bar z) &= f_{12\cO} f_{34\cO} (-1)^\ell \frac{(\nu)_\ell}{(2\nu)_\ell} z^{\frac{\De-\ell}{2}} \bar z^{\frac{\De+\ell}{2}} + \ldots &(\textrm{in~\cite{Komargodski:2016auf}}),
\ee
where $\nu=\frac{d-2}{2}=\frac 1 2$.
Using $f_{43\cO}=(-1)^{\ell_\cO} f_{34\cO}$, we find
\be
\left[
f_{12\cO} f_{34\cO}
\right]_\textrm{here} &=
\left[
f_{12\cO} f_{34\cO} \frac{(\nu)_{\ell_\cO}}{(2\nu)_{\ell_\cO}}
\right]_\textrm{in~\cite{Komargodski:2016auf}}.
\label{eq:operelation}
\ee

\begin{table}
\begin{center}
{\small
\begin{tabular}{|c|c|l|l|l|l|l|}
\hline
$\cO$ & $\Z_2$ & $\ell$ & $\De$ & $\tau=\De-\ell$ & $f_{\s\s\cO}$ & $f_{\e\e\cO}$ \\
\hline
$\e$ & $+$ & 0 & $1.412625{\bf\boldsymbol(10\boldsymbol)}$ & $1.412625{\bf\boldsymbol(10\boldsymbol)}$ & $1.0518537{\bf\boldsymbol(41\boldsymbol)}$ & $1.532435{\bf\boldsymbol(19\boldsymbol)}$ \\
$\e'$ & $+$ & 0 & $3.82968(23)$ & $3.82968(23)$ & $0.053012(55)$ & $1.5360(16)$ \\
& $+$ & 0 & $6.8956(43)$ & $6.8956(43)$ & $0.0007338(31)$ & $0.1279(17)$ \\
& $+$ & 0 & $7.2535(51)$ & $7.2535(51)$ & $0.000162(12)$ & $0.1874(31)$ \\
$T_{\mu\nu}$ & $+$ & 2 & $3$ & $1$ & $0.32613776(45)$ & $0.8891471(40)$ \\
$T'_{\mu\nu}$ & $+$ & 2 & $5.50915(44)$ & $3.50915(44)$ & $0.0105745(42)$ & $0.69023(49)$ \\
& $+$ & 2 & $7.0758(58)$ & $5.0758(58)$ & $0.0004773(62)$ & $0.21882(73)$ \\
$C_{\mu\nu\rho\s}$ & $+$ & 4 & $5.022665(28)$ & $1.022665(28)$ & $0.069076(43)$ & $0.24792(20)$ \\
& $+$ & 4 & $6.42065(64)$ & $2.42065(64)$ & $0.0019552(12)$ & $-0.110247(54)$ \\
& $+$ & 4 & $7.38568(28)$ & $3.38568(28)$ & $0.00237745(44)$ & $0.22975(10)$ \\
& $+$ & 6 & $7.028488(16)$ & $1.028488(16)$ & $0.0157416(41)$ & $0.066136(36)$ \\
\hline
\hline
$\cO$ & $\Z_2$ & $\ell$ & $\Delta$ & $\tau=\De-\ell$ & $f_{\s\e\cO}$ & -  \\
\hline
$\s$ & $-$ & 0 & $0.5181489{\bf\boldsymbol(10\boldsymbol)}$ & $0.5181489{\bf\boldsymbol(10\boldsymbol)}$ & $1.0518537{\bf\boldsymbol(41\boldsymbol)}$ & \\
$\s'$ & $-$ & 0 & $5.2906(11)$ & $5.2906(11)$ & $0.057235(20)$  &\\
& $-$ & 2 & $4.180305(18)$ & $2.180305(18)$ & $0.38915941(81)$ & \\
& $-$ & 2 & $6.9873(53)$ & $4.9873(53)$ & $0.017413(73)$ & \\
& $-$ & 3 & $4.63804(88)$ & $1.63804(88)$ & $0.1385(34)$ & \\
& $-$ & 4 & $6.112674(19)$ & $2.112674(19)$ & $0.1077052(16)$ & \\
& $-$ & 5 & $6.709778(27)$ & $1.709778(27)$ & $0.04191549(88)$ & \\
\hline
\end{tabular}
}
\end{center}
\caption{
Stable operators with dimensions $\De\leq 8$. The leftmost column shows the names of the operators from~\cite{El-Showk:2014dwa}. Errors in bold are rigorous. All other errors are non-rigorous. Because we have chosen different conventions for conformal blocks, our normalization of OPE coefficients differs from those in~\cite{El-Showk:2014dwa,Komargodski:2016auf} by (\ref{eq:operelation}).
}
\label{tab:lowestdim}
\end{table}

\begin{table}
\begin{center}
{\small
\begin{tabular}{|c|l|l|l|l|l|}
\hline
$\Z_2$ & $\ell$  & $\De$ & $\tau=\De-\ell$ & $f_{\s\s\cO}$ & $f_{\e\e\cO}$ \\
\hline
$+$ & 2 & $3$ & $1$ & $0.32613776(45)$ & $0.8891471(40)$ \\
$+$ & 4 & $5.022665(28)$ & $1.022665(28)$ & $0.069076(43)$ & $0.24792(20)$ \\
$+$ & 6 & $7.028488(16)$ & $1.028488(16)$ & $0.0157416(41)$ & $0.066136(36)$ \\
$+$ & 8 & $9.031023(30)$ & $1.031023(30)$ & $0.0036850(54)$ & $0.017318(30)$ \\
$+$ & 10 & $11.0324141(99)$ & $1.0324141(99)$ & $0.00087562(13)$ & $0.0044811(15)$ \\
$+$ & 12 & $13.033286(12)$ & $1.033286(12)$ & $0.000209920(37)$ & $0.00115174(59)$ \\
$+$ & 14 & $15.033838(15)$ & $1.033838(15)$ & $0.000050650(99)$ & $0.00029484(56)$ \\
$+$ & 16 & $17.034258(34)$ & $1.034258(34)$ & $0.000012280(18)$ & $0.00007517(18)$ \\
$+$ & 18 & $19.034564(12)$ & $1.034564(12)$ & $2.98935(46)\.10^{-6}$ & $0.0000191408(89)$ \\
$+$ & 20 & $21.0347884(84)$ & $1.0347884(84)$ & $7.2954(10)\.10^{-7}$ & $4.8632(23)\.10^{-6}$ \\
$+$ & 22 & $23.034983(11)$ & $1.034983(11)$ & $1.78412(27)\.10^{-7}$ & $1.23201(72)\.10^{-6}$ \\
$+$ & 24 & $25.035122(11)$ & $1.035122(11)$ & $4.37261(60)\.10^{-8}$ & $3.1223(15)\.10^{-7}$ \\
$+$ & 26 & $27.035249(11)$ & $1.035249(11)$ & $1.07287(18)\.10^{-8}$ & $7.8948(42)\.10^{-8}$ \\
$+$ & 28 & $29.035344(19)$ & $1.035344(19)$ & $2.6409(19)\.10^{-9}$ & $1.9992(23)\.10^{-8}$ \\
$+$ & 30 & $31.035452(16)$ & $1.035452(16)$ & $6.447(24)\.10^{-10}$ & $5.003(20)\.10^{-9}$ \\
$+$ & 32 & $33.035473(28)$ & $1.035473(28)$ & $1.640(25)\.10^{-10}$ & $1.308(21)\.10^{-9}$ \\
$+$ & 34 & $35.035632(67)$ & $1.035632(67)$ & $3.58(22)\.10^{-11}$ & $2.90(19)\.10^{-10}$ \\
$+$ & 36 & $37.035610(41)$ & $1.035610(41)$ & $1.15(13)\.10^{-11}$ & $9.6(11)\.10^{-11}$ \\
$+$ & 38 & $39.035638(58)$ & $1.035638(58)$ & $2.26(71)\.10^{-12}$ & $1.93(60)\.10^{-11}$ \\
$+$ & 40 & $41.03564(13)$ & $1.03564(13)$ & $7.3(15)\.10^{-13}$ & $6.3(13)\.10^{-12}$ \\
\hline
\end{tabular}
}
\end{center}
\caption{
Operators in the family $[\s\s]_0$. The first line is the stress tensor $T_{\mu\nu}$.
}
\label{tab:sigsig0}
\end{table}

\begin{table}
\begin{center}
{\small
\begin{tabular}{|c|l|l|l|l|l|}
\hline
$\Z_2$ & $\ell$ & $\De$ & $\tau=\De-\ell$ & $f_{\s\s\cO}$ & $f_{\e\e\cO}$  \\
\hline
$+$ & 4 & $6.42065(64)$ & $2.42065(64)$ & $0.0019552(12)$ & $-0.110247(54)$ \\
$+$ & 6 & $8.4957(75)$ & $2.4957(75)$ & $0.000472(49)$ & $-0.0431(48)$ \\
$+$ & 8 & $10.562(12)$ & $2.562(12)$ & $0.0001084(69)$ & $-0.0139(11)$ \\
$+$ & 10 & $12.5659(57)$ & $2.5659(57)$ & $0.00002598(39)$ & $-0.004437(62)$ \\
$+$ & 12 & $14.633(21)$ & $2.633(21)$ & $6.10(33)\.10^{-6}$ & $-0.001224(60)$ \\
$+$ & 14 & $16.6174(75)$ & $2.6174(75)$ & $1.417(34)\.10^{-6}$ & $-0.0003791(54)$ \\
$+$ & 16 & $18.678(24)$ & $2.678(24)$ & $3.547(59)\.10^{-7}$ & $-0.0000972(64)$ \\
$+$ & 18 & $20.654(22)$ & $2.654(22)$ & $7.99(90)\.10^{-8}$ & $-0.0000284(26)$ \\
$+$ & 20 & $22.651(27)$ & $2.651(27)$ & $1.83(13)\.10^{-8}$ & $-7.58(47)\.10^{-6}$ \\
$+$ & 22 & $24.671(18)$ & $2.671(18)$ & $4.55(72)\.10^{-9}$ & $-2.09(19)\.10^{-6}$ \\
$+$ & 24 & $26.681(20)$ & $2.681(20)$ & $1.168(29)\.10^{-9}$ & $-5.67(17)\.10^{-7}$ \\
$+$ & 26 & $28.706(24)$ & $2.706(24)$ & $2.81(17)\.10^{-10}$ & $-1.49(11)\.10^{-7}$ \\
$+$ & 28 & $30.6923(81)$ & $2.6923(81)$ & $6.69(36)\.10^{-11}$ & $-4.162(88)\.10^{-8}$ \\
$+$ & 30 & $32.702(11)$ & $2.702(11)$ & $1.62(16)\.10^{-11}$ & $-1.066(59)\.10^{-8}$ \\
$+$ & 32 & $34.718(17)$ & $2.718(17)$ & $4.15(42)\.10^{-12}$ & $-2.83(18)\.10^{-9}$ \\
$+$ & 34 & $36.717(16)$ & $2.717(16)$ & $9.44(77)\.10^{-13}$ & $-7.33(59)\.10^{-10}$ \\
$+$ & 36 & $38.697(17)$ & $2.697(17)$ & $2.40(39)\.10^{-13}$ & $-2.12(34)\.10^{-10}$ \\
$+$ & 38 & $40.701(19)$ & $2.701(19)$ & $5.4(17)\.10^{-14}$ & $-5.2(15)\.10^{-11}$ \\
$+$ & 40 & $42.726(18)$ & $2.726(18)$ & $1.59(49)\.10^{-14}$ & $-1.55(48)\.10^{-11}$ \\
$+$ & 42 & $44.729(15)$ & $2.729(15)$ & $4.2(12)\.10^{-15}$ & $-4.4(11)\.10^{-12}$ \\
\hline
\end{tabular}
}
\end{center}
\caption{
Operators in the family $[\e\e]_0$.
}
\label{tab:epseps0}
\end{table}

\begin{table}
\begin{center}
{\small
\begin{tabular}{|c|l|l|l|l|l|}
\hline
$\Z_2$ & $\ell$ & $\De$ & $\tau=\De-\ell$ & $f_{\s\s\cO}$ & $f_{\e\e\cO}$ \\
\hline
$+$ & 0 & $3.82968(23)$ & $3.82968(23)$ & $0.053012(55)$ & $1.5360(16)$ \\
$+$ & 2 & $5.50915(44)$ & $3.50915(44)$ & $0.0105745(42)$ & $0.69023(49)$ \\
$+$ & 4 & $7.38568(28)$ & $3.38568(28)$ & $0.00237745(44)$ & $0.22975(10)$ \\
$+$ & 6 & $9.32032(34)$ & $3.32032(34)$ & $0.00055657(42)$ & $0.06949(11)$ \\
$+$ & 8 & $11.2751(24)$ & $3.2751(24)$ & $0.00013251(91)$ & $0.01980(15)$ \\
$+$ & 10 & $13.2410(10)$ & $3.2410(10)$ & $0.00003234(15)$ & $0.005459(39)$ \\
$+$ & 12 & $15.2301(64)$ & $3.2301(64)$ & $7.64(14)\.10^{-6}$ & $0.001538(22)$ \\
$+$ & 14 & $17.1944(55)$ & $3.1944(55)$ & $1.930(46)\.10^{-6}$ & $0.000386(14)$ \\
$+$ & 16 & $19.1950(62)$ & $3.1950(62)$ & $4.568(72)\.10^{-7}$ & $0.0001107(16)$ \\
$+$ & 18 & $21.1720(23)$ & $3.1720(23)$ & $1.153(27)\.10^{-7}$ & $0.00002798(33)$ \\
$+$ & 20 & $23.167(10)$ & $3.167(10)$ & $2.74(11)\.10^{-8}$ & $7.45(52)\.10^{-6}$ \\
$+$ & 22 & $25.163(10)$ & $3.163(10)$ & $6.88(22)\.10^{-9}$ & $1.937(51)\.10^{-6}$ \\
$+$ & 24 & $27.1491(82)$ & $3.1491(82)$ & $1.716(45)\.10^{-9}$ & $4.92(42)\.10^{-7}$ \\
$+$ & 26 & $29.1460(53)$ & $3.1460(53)$ & $4.183(78)\.10^{-10}$ & $1.347(62)\.10^{-7}$ \\
$+$ & 28 & $31.1306(52)$ & $3.1306(52)$ & $1.056(50)\.10^{-10}$ & $3.35(10)\.10^{-8}$ \\
$+$ & 30 & $33.126(12)$ & $3.126(12)$ & $2.54(10)\.10^{-11}$ & $8.35(42)\.10^{-9}$ \\
$+$ & 32 & $35.1299(77)$ & $3.1299(77)$ & $6.71(17)\.10^{-12}$ & $2.36(13)\.10^{-9}$ \\
$+$ & 34 & $37.1174(64)$ & $3.1174(64)$ & $1.39(14)\.10^{-12}$ & $4.87(48)\.10^{-10}$ \\
$+$ & 36 & $39.1079(78)$ & $3.1079(78)$ & $4.84(56)\.10^{-13}$ & $1.70(17)\.10^{-10}$ \\
$+$ & 38 & $41.101(29)$ & $3.101(29)$ & $8.4(28)\.10^{-14}$ & $2.5(11)\.10^{-11}$ \\
$+$ & 40 & $43.102(18)$ & $3.102(18)$ & $2.63(64)\.10^{-14}$ & $9.0(26)\.10^{-12}$ \\
$+$ & 42 & $45.116(27)$ & $3.116(27)$ & $7.9(22)\.10^{-15}$ & $3.42(95)\.10^{-12}$ \\
\hline
\end{tabular}
}
\end{center}
\caption{
Operators in the family $[\s\s]_1$.
}
\label{tab:sigsig1}
\end{table}

\begin{table}
\begin{center}
{\small
\begin{tabular}{|c|l|l|l|l|}
\hline
$\Z_2$ & $\ell$ & $\De$ & $\tau=\De-\ell$ & $f_{\s\e\cO}$  \\
\hline
$-$ & 2 & $4.180305(18)$ & $2.180305(18)$ & $0.38915941(81)$ \\
$-$ & 3 & $4.63804(88)$ & $1.63804(88)$ & $0.1385(34)$ \\
$-$ & 4 & $6.112674(19)$ & $2.112674(19)$ & $0.1077052(16)$ \\
$-$ & 5 & $6.709778(27)$ & $1.709778(27)$ & $0.04191549(88)$ \\
$-$ & 6 & $8.08097(25)$ & $2.08097(25)$ & $0.0286902(80)$ \\
$-$ & 7 & $8.747293(56)$ & $1.747293(56)$ & $0.01161255(13)$ \\
$-$ & 8 & $10.0623(29)$ & $2.0623(29)$ & $0.00745(21)$ \\
$-$ & 9 & $10.77075(36)$ & $1.77075(36)$ & $0.003115(12)$ \\
$-$ & 10 & $12.0492(18)$ & $2.0492(18)$ & $0.001940(19)$ \\
$-$ & 11 & $12.787668(92)$ & $1.787668(92)$ & $0.000823634(82)$ \\
$-$ & 12 & $14.0383(33)$ & $2.0383(33)$ & $0.0004983(88)$ \\
$-$ & 13 & $14.80006(51)$ & $1.80006(51)$ & $0.0002150(10)$ \\
$-$ & 14 & $16.0305(12)$ & $2.0305(12)$ & $0.0001291(12)$ \\
$-$ & 15 & $16.81009(16)$ & $1.81009(16)$ & $0.000055870(15)$ \\
$-$ & 16 & $18.025(11)$ & $2.025(11)$ & $0.0000313(30)$ \\
$-$ & 17 & $18.81794(18)$ & $1.81794(18)$ & $0.0000144219(91)$ \\
$-$ & 18 & $20.01947(94)$ & $2.01947(94)$ & $8.442(28)\.10^{-6}$ \\
$-$ & 19 & $20.8246(11)$ & $1.8246(11)$ & $3.690(54)\.10^{-6}$ \\
$-$ & 20 & $22.0152(36)$ & $2.0152(36)$ & $2.131(28)\.10^{-6}$ \\
$-$ & 21 & $22.83035(11)$ & $1.83035(11)$ & $9.5120(13)\.10^{-7}$ \\
$-$ & 22 & $24.01143(53)$ & $2.01143(53)$ & $5.4746(61)\.10^{-7}$ \\
$-$ & 23 & $24.83518(65)$ & $1.83518(65)$ & $2.428(11)\.10^{-7}$ \\
$-$ & 24 & $26.00809(94)$ & $2.00809(94)$ & $1.3908(17)\.10^{-7}$ \\
$-$ & 25 & $26.8394(13)$ & $1.8394(13)$ & $6.16(18)\.10^{-8}$ \\
$-$ & 26 & $28.0045(17)$ & $2.0045(17)$ & $3.523(20)\.10^{-8}$ \\
$-$ & 27 & $28.84330(31)$ & $1.84330(31)$ & $1.5809(50)\.10^{-8}$ \\
$-$ & 28 & $30.0042(38)$ & $2.0042(38)$ & $8.86(18)\.10^{-9}$ \\
$-$ & 29 & $30.84667(23)$ & $1.84667(23)$ & $4.0311(33)\.10^{-9}$ \\
$-$ & 30 & $31.99996(74)$ & $1.99996(74)$ & $2.2555(81)\.10^{-9}$ \\
$-$ & 31 & $32.84955(61)$ & $1.84955(61)$ & $1.0144(28)\.10^{-9}$ \\
$-$ & 32 & $33.9976(28)$ & $1.9976(28)$ & $5.82(11)\.10^{-10}$ \\
$-$ & 33 & $34.85245(50)$ & $1.85245(50)$ & $2.669(34)\.10^{-10}$ \\
$-$ & 34 & $35.99600(99)$ & $1.99600(99)$ & $1.374(72)\.10^{-10}$ \\
$-$ & 35 & $36.85548(90)$ & $1.85548(90)$ & $5.94(34)\.10^{-11}$ \\
$-$ & 36 & $37.9939(12)$ & $1.9939(12)$ & $4.02(45)\.10^{-11}$ \\
$-$ & 37 & $38.85691(49)$ & $1.85691(49)$ & $1.99(19)\.10^{-11}$ \\
$-$ & 38 & $39.9895(17)$ & $1.9895(17)$ & $9.5(18)\.10^{-12}$ \\
$-$ & 39 & $40.8583(11)$ & $1.8583(11)$ & $3.7(13)\.10^{-12}$ \\
$-$ & 40 & $41.9886(15)$ & $1.9886(15)$ & $2.50(96)\.10^{-12}$ \\
$-$ & 41 & $42.8607(14)$ & $1.8607(14)$ & $1.32(24)\.10^{-12}$ \\
$-$ & 42 & $43.9915(21)$ & $1.9915(21)$ & $7.9(19)\.10^{-13}$ \\
\hline
\end{tabular}
}
\end{center}
\caption{
Operators in the family $[\s\e]_0$.
}
\label{tab:sigeps0}
\end{table}

\begin{table}
\begin{center}
{\small
\begin{tabular}{|c|l|l|l|l|l|}
\hline
$\Z_2$  & $\ell$ & $\De$ & $\tau=\De-\ell$ & $f_{\s\s\cO}$ & $f_{\e\e\cO}$ \\
\hline
$+$ & 0 & $1.412625{\bf\boldsymbol(10\boldsymbol)}$ & $1.412625{\bf\boldsymbol(10\boldsymbol)}$ & $1.0518537{\bf\boldsymbol(41\boldsymbol)}$ & $1.532435{\bf\boldsymbol(19\boldsymbol)}$ \\
$+$ & 2 & $7.0758(58)$ & $5.0758(58)$ & $0.0004773(62)$ & $0.21882(73)$ \\
$+$ & 4 & $8.9410(99)$ & $4.9410(99)$ & $0.0001173(21)$ & $0.08635(18)$ \\
$+$ & 6 & $10.975(13)$ & $4.975(13)$ & $0.00002437(59)$ & $0.02775(17)$ \\
$+$ & 0 & $6.8956(43)$ & $6.8956(43)$ & $0.0007338(31)$ & $0.1279(17)$ \\
$+$ & 0 & $7.2535(51)$ & $7.2535(51)$ & $0.000162(12)$ & $0.1874(31)$ \\
\hline
\hline
$\Z_2$  &$\ell$ & $\Delta$ & $\tau=\De-\ell$ & $f_{\s\e\cO}$ & -  \\
\hline
$-$ & 0 & $0.5181489{\bf\boldsymbol(10\boldsymbol)}$ & $0.5181489{\bf\boldsymbol(10\boldsymbol)}$ & $1.0518537{\bf\boldsymbol(41\boldsymbol)}$ & \\
$-$ & 0 & $5.2906(11)$ & $5.2906(11)$ & $0.057235(20)$ &\\
$-$ & 2 & $6.9873(53)$ & $4.9873(53)$ & $0.017413(73)$ &\\
\hline
\end{tabular}
}
\end{center}
\caption{
Stable operators not in one of the families $[\s\s]_0$, $[\e\e]_0$, $[\s\s]_1$, $[\s\e]_0$. Errors in bold are rigorous. All other errors are non-rigorous.
}
\label{tab:sporadic}
\end{table}

\section{Sums of $\SL(2,\R)$ blocks with general coefficients and general spacing}

\label{app:nonintegerspacing}

Consider the sum
\be
\sum_{\substack{h = h_0 + \ell + \de(h) \\ \ell=0,1,\dots}} \pdr{h}{\ell} p(h) k_{2 h}(1-z) 
&=
\sum_a c_a y^a + \sum_{k=0}^\oo y^k \p{\a_k[p,\de](h_0) \log y + \b_k[p,\de](h_0)},
\label{eq:mynonintegerspaced}
\ee
with general coefficients $p(h)$, and where the weights are the solutions of
\be
h &= h_0+\ell + \de(h),\qquad \ell=0,1,\dots.
\ee

To compute the Casimir-singular terms, we must match asymptotic expansions
\be
p(h) &\sim \sum_{a\in A} c_a S_a(h).
\ee
To compute Casimir-regular terms, we expand $k_{2h}(1-z)$ in small $y$ and naively switch the order of summations,
\be
\sum_{k=0}^\oo \pdr{}{k}\p{y^k \sum_{\substack{h}} \pdr{h}{\ell} p(h) \p{-\frac{\G(2h)}{\G(h)^2}T_{-k-1}(h)}}.
\ee
We must now regularize the sum over $h$.
Using $\pdr{h}{\ell} = 1+\pdr{\de}{h_0}$, one can show
\be
\pdr{h}{\ell} p(h) \p{-\frac{\G(2h)}{\G(h)}^2 T_{-k-1}(h)}
&=
\sum_{m=0}^\oo \ptl_{h_0}^m\p{p(h_0+\ell)\frac{\de(h_0+\ell)^m}{m!}\p{-\frac{\G(2(h_0+\ell))}{\G(h_0+\ell)^2}T_{-k-1}(h_0+\ell)}}.
\ee
Now form the asymptotic expansions
\be
p(h) \frac{\de(h)^m}{m!}
&\sim
\sum_{a\in A_{m}} c_{a}^{(m)} S_{a}(h).
\ee
with coefficients $c_{a}^{(m)}$ and sets $A_{m}$. (When $m=0$, these reduce to $c_a$ and $A$ above.)
Note that $-\frac{\G(2h)}{\G(h)^2}S_a(h) = (1-2h)T_a(h) \sim h^{-2a-1}$.  The derivative $\ptl_{h_0}$ decreases degree in $\ell$ by $1$.  Thus, the combination
\be
\label{eq:nicecombo}
&f_k(\ell,h_0) \equiv \nn\\
&\pdr{h}{\ell} p(h) \p{-\frac{\G(2h)}{\G(h)}^2 T_{-k-1}(h)} - \sum_{m=0}^M \sum_{\substack{a\in A_{m} \\ a \leq k-m/2}} 
c_{a}^{(m)}
\ptl_{h_0}^m
\p{
 (1-2(h_0+\ell))T_{a}(h_0+\ell)T_{-k-1}(h_0+\ell)
}
\ee
falls off faster than $\ell^{-1}$, so its sum over $\ell$ converges.  Here, we must choose $M$ so that $\min(A_{m})\geq k-m/2$ for all $m>M$. If $\de$ approaches zero as $h\to \oo$, it is sufficient to take $M\geq 2k-2\min(A)$.

Summing (\ref{eq:nicecombo}) over $\ell$ and adding back the regularized sum of the subtractions, we find
\be
&\a_k[p,\de](h_0)= 
\sum_{m=0}^M 
\sum_{\substack{a\in A_{m} \\ a \leq k-m/2}} c_a^{(m)} \ptl_{h_0}^m\cA_{a,-k-1}(h_0) +\sum_{\ell=0}^\oo f_k(\ell,h_0).
\ee
Note that $f_k(\ell, h_0)$ as we've defined it is analytic in $k$, so we can form the derivative $\b_k[p,\de](h_0)$.  The above result generalizes easily to the case of alternating or even sums, where we must simply replace $\cA\to \cA^-$ or $\cA\to\cA^\mathrm{even}$ and modify the sum over $\ell$ appropriately.

\subsection{Special cancellations between singular and regular parts}
\label{app:regularizationofsums}

We sometimes encounter sums where both the Casimir-singular and Casimir-regular part naively diverge, but the divergences cancel to leave a finite quantity. This occurs in sums over un-mixed blocks with coefficients $\lim_{\e\to 0}\G(-\e)^2 S_\e(h)$ and in sums over mixed blocks with coefficients $\G(-\e)S^{r,s}_{\e-r}(h)$.  In such sums, the naive Casimir-singular parts are proportional to
\be
\lim_{\e\to 0}\G(-\e)^2 y^\e &= \frac 1 2 \log ^2y+ \lim_{\e\to 0} \p{ \frac{1}{\e^2}+\frac{\log y+2 \gamma }{\e}+ 2 \gamma  \log y+2 \gamma ^2+\frac{\pi ^2}{6} },
\label{eq:badcassingone}
\ee
or in the case of mixed blocks
\be
\lim_{\e\to 0} \G(-\e) y^{\e-r} &= -y^{-r} \log y-\lim_{\e\to 0} y^{-r}\p{\frac 1 \e + \gamma}.
\label{eq:badcassingtwo}
\ee
We define regularized quantities $\bar \a_k$ by replacing $S_0(h)\to S_{\e}(h)$ and $S_{-r}^{r,s}(h) \to S_{\e-r}^{r,s}(h)$, adding the quantities in parentheses in (\ref{eq:badcassingone}) or (\ref{eq:badcassingtwo}) to $\a_k$, and then taking the limit $\e\to 0$. Examples of this procedure are given in~(\ref{eq:extractfunnylimit}) and (\ref{eq:otherexamplewithpolecancelling}). In general, we must apply it whenever the asymptotic large-$h$ expansion of $p(h)$ contains terms of the form $\lim_{\e\to 0}\G(-\e)^2 S_\e(h)$ or $\G(-\e)S^{r,s}_{\e-r}(h)$.

\section{Box diagrams}
\label{app:boxdiagrams}

\begin{figure}
\begin{center}
\begin{tikzpicture}[xscale=0.6,yscale=0.6]
\draw[thick] (0,0) -- (0,1.5);
\draw[thick] (0,1.5) -- (0,5.5);
\draw[thick] (0,5.5) -- (0,7);
\draw[thick] (5,0) -- (5,1.5);
\draw[thick] (5,1.5) -- (5,5.5);
\draw[thick] (5,5.5) -- (5,7);
\draw[thick] (0,1.5) -- (5,1.5);
\draw[thick] (0,5.5) -- (5,5.5);
\node[below] at (0,0) {$1$};
\node[below] at (5,0) {$2$};
\node[above] at (0,7) {$4$};
\node[above] at (5,7) {$3$};
\node[below] at (2.5,1.5) {$a$};
\node[right] at (5,3.5) {$b$};
\node[above] at (2.5,5.5) {$c$};
\node[left] at (0,3.5) {$d$};
\end{tikzpicture}
\end{center}
\caption{
The Casimir-singular part of the above diagram is the same, whether we interpret it as $[db]_n$ exchange in the $12\to 34$ channel or $[ac]_m$ exchange in the $23\to 14$ channel.
}
\label{fig:boxdiagram}
\end{figure}
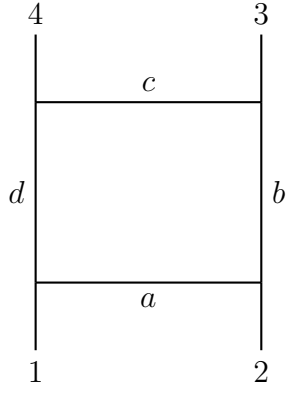

We claim that the Casimir-singular part of the box diagram in figure~\ref{fig:boxdiagram} is the same whether we read the diagram from left-to-right or bottom-to-top.
Consequently, the Casimir-singular part of the sum of box diagrams over all possible internal legs is crossing-symmetric.\footnote{This is equivalent to the claim that exponentiation of the twist Hamiltonian in section~\ref{sec:twisthamiltonian} is consistent with crossing-symmetry at asymptotically large spin, to second order.}

We can regard any CFT as a 2d theory with $\SL(2,\R)\x\SL(2,\R)$ symmetry. If our claim holds in 2d, it holds in general dimensions. A benefit of working in 2d is that tensor structures are extremely simple, so we can prove the claim for external operators of any spin (not just scalars).

Some conventions in 2d theories are different from those in the main text.  We have the two and three-point functions
\be
\<\f(z_1,\bar z_1)\f(z_2,\bar z_2)\> &= \frac{1}{z_{12}^{2h}\bar z_{12}^{2\bar h}}\nn\\
\<\f_1(z_1,\bar z_1)\f_2(z_2,\bar z_2)\f_3(z_3,\bar z_3)\> &= f_{123}\frac{1}{z_{12}^{h_1+h_2-h_3} z_{23}^{h_2+h_3-h_1} z_{31}^{h_3+h_1-h_2}} \x (z\leftrightarrow \bar z).
\ee
In unitary theories with these conventions, operators satisfy the reality property
\be
\f(z,\bar z)^\dag &= (-1)^\ell \f(z,\bar z).
\ee
That is, even-spin operators are real and odd-spin operators are imaginary.  The three-point coefficients have the same reality properties as the product of operators.  They satisfy the symmetry property $f_{abc}=f_{bac}(-1)^{\ell_a+\ell_b+\ell_c}$ and similarly for other permutations.

Consider a four-point function $\<\f_1\cdots \f_4\>$ where the operators $\f_i$ have weights $(h_i,\bar h_i)$ (not necessarily equal).  We have the conformal block expansion
\be
\<\f_1(z_1,\bar z_2)\cdots \f_4(z_4,\bar z_4)\>
&=
\frac{1}{z_{21}^{h_1+h_2}z_{43}^{h_3+h_4}}\frac{z_{41}^{h_{34}}z_{32}^{h_{12}}}{z_{31}^{h_{12}+h_{34}}} \x (z\leftrightarrow \bar z)\nn\\
& \x \sum_a f_{12a} f_{34a} k_{2h_a}^{h_{12},h_{34}}(z) k_{2\bar h_a}^{\bar h_{12},\bar h_{34}}(\bar z).
\ee
The crossing equation reads
\be
\label{eq:twodcrossing}
\bar y^{-\bar h_1-\bar h_3} \sum_a f_{12a}f_{34a}k^{h_{12},h_{34}}_{2h_a}(1-z) k^{\bar h_{12},\bar h_{34}}_{2\bar h_a}(\bar z)
&= y^{-h_1-h_3} \sum_{a} f_{23a}f_{41a} k_{2h_a}^{h_{32},h_{14}}(z)k_{2\bar h_a}^{\bar h_{32},\bar h_{14}}(1-\bar z).
\ee
On the right-hand side, we have the expansions
\be
\bar y^{\bar h_1+\bar h_3} k_{2\bar h_\cO}^{\bar h_{32},\bar h_{14}}(1-\bar z) &= \sum_{k=0}^\oo K_k^{\bar h_{32},\bar h_{14}}(\bar h_\cO) \bar y^{k+\bar h_1 + \bar h_2} + (3\leftrightarrow 1, 2\leftrightarrow 4),\\
y^{-h_1-h_3} k_{2 h_\cO}^{h_{32},h_{14}}(z) &= \sum_{m=0}^\oo \frac{(h_\cO-h_{32})_m(h_\cO-h_{14})_m}{(2h_\cO)_m m!}(-1)^m y^{h_\cO+m-h_1-h_3}.
\ee
On the left-hand side, this matches to the contribution of the double-twist operators $[12]_n$. We have
\be
\l_{12[12]_n}(h)\l_{34[12]_n}(h) &\sim \sum_a f_{23a}f_{41a} \a^{1234}_a(h) \b^{\bar 1\bar 2\bar 3\bar 4}_{\bar a}(n) + (-1)^{\ell+\ell_1+\ell_2}(1\leftrightarrow 2)\nn\\
\a^{1234}_a(h) &= \sum_{m=0}^\oo \frac{(h_a-h_{32})_m(h_a-h_{14})_m}{(2h_a)_m m!}(-1)^m S^{h_{12},h_{34}}_{h_a+m-h_1-h_3}(h)\\
\b^{\bar 1\bar 2\bar 3\bar 4}_{\bar a}(n) &= \sum_{k=0}^n K_k^{\bar h_{32},\bar h_{14}}(\bar h_a) S^{\bar h_{12},\bar h_{34}}_{-k-\bar h_1-\bar h_2}(\bar h_1+\bar h_2+n).
\ee
As in section~\ref{sec:matchingtosides}, the $\l$'s differ from the $f$'s by a Jacobian factor $\pdr{h}{\ell}$.
To get the first line, we used $f_{abc}=(-1)^{\ell_a+\ell_b+\ell_c} f_{bac}$. One can check that swapping $3\leftrightarrow 4$ gives the same quantity as swapping $1\leftrightarrow 2$ because
\be
(-1)^{\ell_1+\ell_2} f_{13a} f_{42a} &= (-1)^{\ell_3+\ell_4} f_{24a} f_{31a},
\ee
and  
\be
\a^{2134}_a(h) &\sim \a^{1243}_a(h),\\
\b^{2134}_a(n) &= \b^{1243}_a(n).
\ee
(Here, ``$\sim$" means the two quantities have the same large-$h$ expansion to all orders in $1/h$.)

Let us now look at the contribution of the double-twist operator $[bd]_n$ in the $12\to34$ channel.  We have
\be
\l_{bd[bd]_n}(h) \l_{12[bd]_n}(h) &\sim \sum_a f_{d1a}f_{a2b} \a^{bd12}_a(h) \b^{\bar b\bar d\bar 1\bar 2}_{\bar a}(n) + (-1)^{\ell+\ell_b+\ell_d} (b\leftrightarrow d) \nn\\
\l_{db[bd]_n}(h) \l_{34[bd]_n}(h) &\sim \sum_c f_{b3c}f_{c4d} \a^{db34}_c(h) \b^{\bar d\bar b\bar 3\bar 4}_{\bar c}(n) + (-1)^{\ell+\ell_b+\ell_d} (b\leftrightarrow d)\nn\\
\l_{db[bd]_n}(h) \l_{bd[bd]_n}(h) &\sim S^{h_{bd},h_{db}}_{-h_b-h_d}(h)S^{\bar h_{bd},\bar h_{db}}_{-\bar h_b - \bar h_d}(\bar h_b+\bar h_d+n) + \dots
\ee
where we have used $f_{abc}=f_{cab}$.
Thus,
\be
\l_{12[bd]_n}(h)\l_{43[bd]_n}(h) &\sim \sum_{a,c} f_{d1a}f_{a2b} f_{b3c}f_{c4d} \frac{\a^{bd12}_a(h) \a^{db34}_c(h)}{S_{-h_b-h_d}^{h_{bd},h_{db}}(h)}  \frac{\b^{\bar b\bar d\bar 1\bar 2}_{\bar a}(n)\b^{\bar d\bar b\bar 3\bar 4}_{\bar c}(n)}{S_{-\bar h_b - \bar h_d}^{\bar h_{bd},\bar h_{db}}(\bar h_b+\bar h_d + n)}\nn\\
&\qquad\qquad \phantom{\frac{1^1_1}{1^1_1}}+ (b\leftrightarrow d) + \dots,
\ee
where ``$\dots$" represent terms proportional to $(-1)^\ell$ which do not contribute to the Casimir-singular part. 

Now we would like to compute the Casimir-singular terms in
\be
\bar y^{-\bar h_1 -\bar h_3} \sum_{h,n} f_{12[bd]_n}(h)f_{43[bd]_n}(h) k^{h_{12},h_{34}}_{2h_{[bd]_n}}(1-z) k^{\bar h_{12},\bar h_{34}}_{2\bar h_{[bd]_n}}(\bar z).
\ee
Isolating the contribution from $a,c$, the sums factorize into holomorphic and antiholomorphic parts. The antiholomorphic sum is
\be
&\sum_{n=0}^\oo \frac{\b^{\bar b\bar d\bar 1\bar 2}_{\bar a}(n)\b^{\bar d\bar b\bar 3\bar 4}_{\bar c}(n)}{S_{-\bar h_b - \bar h_d}^{\bar h_{bd},\bar h_{db}}(\bar h_b+\bar h_d + n)} k^{\bar h_{12},\bar h_{34}}_{2(\bar h_b+\bar h_d + n)}(\bar z)
\\
&= \sum_{n,m=0}^\oo \frac{\b^{\bar b\bar d\bar 1\bar 2}_{\bar a}(n)\b^{\bar d\bar b\bar 3\bar 4}_{\bar c}(n)}{S_{-\bar h_b - \bar h_d}^{\bar h_{bd},\bar h_{db}}(\bar h_b+\bar h_d + n)} \frac{(\bar h_b + \bar h_d +n - \bar h_{12})_m(\bar h_b+\bar h_d+n - \bar h_{34})_m}{(2(\bar h_b+\bar h_d+n))_m m!}(-1)^m \bar y^{\bar h_b+\bar h_d+n+m}\nn\\
&= \sum_{m=0}^\oo \sum_{n=0}^m \frac{\b^{\bar b\bar d\bar 1\bar 2}_{\bar a}(n)\b^{\bar d\bar b\bar 3\bar 4}_{\bar c}(n)}{S_{-\bar h_b - \bar h_d}^{\bar h_{bd},\bar h_{db}}(\bar h_b+\bar h_d + n)} \frac{(\bar h_b + \bar h_d +n - \bar h_{12})_{m-n}(\bar h_b+\bar h_d+n - \bar h_{34})_{m-n}}{(2(\bar h_b+\bar h_d+n))_{m-n} (m-n)!}(-1)^{m-n} \bar y^{\bar h_b+\bar h_d+m}.\nn\\
&= \sum_{m=0}^\oo \g_{\bar 1\bar 2\bar 3\bar 4; \bar a \bar b \bar c \bar d}(m) \bar y^{\bar h_b+\bar h_d+m},
\ee
where
\be
\g_{\bar 1\bar 2\bar 3\bar 4; \bar a \bar b \bar c \bar d}(m) &\equiv
\sum_{n=0}^m (-1)^{m-n} \frac{\b^{\bar b\bar d\bar 1\bar 2}_{\bar a}(n)\b^{\bar d\bar b\bar 3\bar 4}_{\bar c}(n)}{S_{-\bar h_b - \bar h_d}^{\bar h_{bd},\bar h_{db}}(\bar h_b+\bar h_d + n)} \frac{(\bar h_b + \bar h_d +n - \bar h_{12})_{m-n}(\bar h_b+\bar h_d+n - \bar h_{34})_{m-n}}{(2(\bar h_b+\bar h_d+n))_{m-n} (m-n)!}.
\ee

Meanwhile, to do the sum over $h$, we must find a large-$h$ expansion
\be
\frac{\a^{bd12}_a(h) \a^{db34}_c(h)}{S_{-h_b-h_d}^{h_{bd},h_{db}}(h)}  &\sim \sum_{m=0}^\oo c_m S^{h_{12},h_{34}}_{h_a+h_c-h_1-h_3+m}(h).
\ee
We claim that these coefficients are
\be
c_m &= \g_{3214;badc}(m).
\ee
That is, we have the truly remarkable identity
\be
\frac{\a^{bd12}_a(h) \a^{db34}_c(h)}{S_{-h_b-h_d}^{h_{bd},h_{db}}(h)}  &\sim \sum_{m=0}^\oo \g_{3214;badc}(m) S^{h_{12},h_{34}}_{h_a+h_c-h_1-h_3+m}(h).
\ee
Thus, we obtain
\be
&\p{\sum_{a,c}\sum_{m,n} f_{d1a}f_{a2b} f_{b3c}f_{c4d}  \g_{3214;badc}(m) \g_{\bar 1\bar 2\bar 3\bar 4; \bar a \bar b \bar c \bar d}(n) + b\leftrightarrow d}\nn\\
& \qquad \x y^{h_a+h_c-h_1-h_3+m} \bar y^{\bar h_b + \bar h_d - \bar h_1 - \bar h_3 + n}
\ee
We should sum over unordered pairs $b,d$, with weight $1/2$ for the case $b=d$ because only even-spin operators are present. This is equivalent to dropping the second term and summing over ordered pairs $b,d$.  Overall, the Casimir-singular part in $y,\bar y$ is
\be
\sum_{a,b,c,d}\sum_{m,n=0}^\oo f_{d1a}f_{a2b} f_{b3c}f_{c4d}  \g_{3214;badc}(m) \g_{\bar 1\bar 2\bar 3\bar 4; \bar a \bar b \bar c \bar d}(n) y^{h_a+h_c-h_1-h_3+m} \bar y^{\bar h_b + \bar h_d - \bar h_1 - \bar h_3 + n}.
\ee
The summand is invariant under $(1,a,d,y,m) \leftrightarrow (3,b,c,\bar y,n)$ and multiplication by the phase $(-1)^{\ell_1+\ell_2+\ell_3+\ell_4}$ coming from the prefactor in the crossing equation (\ref{eq:twodcrossing}). Thus, the above result is crossing-symmetric.

There are some special cases of this calculation that we must take care with.  Each special case will require us to modify the calculation slightly. We should check that the results are still crossing-symmetric after these modifications.

Suppose $\bar h_d+\bar h_b=\bar h_1+\bar h_2$.  For simplicity we assume $[bd]=[12]$, though this is not actually necessary.  Then in the $1,2\to 3,4$ channel we have
\be
&\bar y^{-\bar h_1-\bar h_3} \sum_{n,h} f_{12[12]_n}(h) f_{34[12]_n}(h) k_{2(\bar h_1 + \bar h_2 + n + \g_{[12]_n}/2)}(\bar z) k_{2h}^{h_{12},h_{34}}(1-z)\\
&= \bar y^{-\bar h_1-\bar h_3} \sum_{n,h} f_{12[12]_n}(h) f_{34[12]_n}(h) \frac{\g_{[12]_n}(h)}{2}\x\nn\\
& \qquad \qquad \qquad
 \frac 1 2\p{\pdr{}{\bar h_1}+ \pdr{}{\bar h_2}} k^{\bar h_{12},\bar h_{34}}_{2(\bar h_1 + \bar h_2 + n)}(\bar z) k_{2h}^{h_{12},h_{34}}(1-z) + \dots
\ee
The combination above is
\be
f_{12[12]_n}(h) f_{34[12]_n}(h) \frac{\g_{[12]_n}(h)}{2} &\sim 
\lim_{b,d\to 1,2}(\bar h_b+\bar h_d-\bar h_1-\bar h_2) f_{12[bd]_n}(h) f_{34[bd]_n}(h).
\ee
Recall that we are interested in computing only the Casimir-singular terms with respect to the crossed-channel. The Casimir-regular terms are proportional to
\be
\bar y^{k-\bar h_3+\bar h_2},\quad \bar y^{k-\bar h_1+ \bar h_4}
\ee
Only the parts proportional to $\bar y^{k-\bar h_3 + \bar h_2} \log \bar y$ or $\bar y^{k-\bar h_1 + \bar h_4} \log \bar y$ are Casimir-singular.  To determine them, we can use
\be
\pdr{}{\bar h} k_{2\bar h}^{r,s}(\bar z) &= \log \bar y\, k_{2\bar h}^{r,s}(\bar z) + \sum_{k=1}^\oo \# \bar y^{k+\bar h}.
\ee
Thus,
\be
\bar y^{-\bar h_1 - \bar h_3} \frac 1 2\p{\pdr{}{\bar h_1}+ \pdr{}{\bar h_2}} k^{\bar h_{12},\bar h_{34}}_{2(\bar h_1 + \bar h_2 + n)}(\bar z) &= \bar y^{-\bar h_1 - \bar h_3} \log \bar y\, k^{\bar h_{12},\bar h_{34}}_{2(\bar h_1 + \bar h_2 + n)}(\bar z) + [\dots]_{\bar y},
\ee
and the Casimir-singular terms above are
\be
&\bar y^{-\bar h_1 - \bar h_3} \log \bar y \sum_{n,h} f_{12[12]_n}(h) f_{34[12]_n}(h) \frac{\g_{[12]_n}(h)}{2} k^{\bar h_{12},\bar h_{34}}_{2(\bar h_1 + \bar h_2 + n)}(\bar z) k_{2h}^{h_{12},h_{34}}(1-z)\nn\\
&= 
\lim_{b,d \to 2,1} (\bar h_d + \bar h_b - \bar h_1 - \bar h_2) f_{d1a}f_{a2b} f_{b3c}f_{c4d}\nn\\
&\qquad\x  \sum_{m,n} \g_{3214;badc}(m) \g_{\bar 1\bar 2\bar 3\bar 4; \bar a \bar b \bar c \bar d}(n)
\x y^{h_a+h_c-h_1-h_3+m} \bar y^{\bar h_b + \bar h_d - \bar h_1 - \bar h_3 + n}\log \bar y.
\label{eq:curehbarpoles}
\ee
The $\bar y$-Casimir-regular terms in the $12\to 34$ channel start at $n=1$ (i.e.\ not at leading order in $\bar y$), because of the fact that $k_{2\bar h}$ has leading order $y^{\bar h}$, with a coefficient that is independent of $\bar h$.

In the case where we have $[db]=[12]=[34]$, the Casimir-singular terms are given by expanding to second order in the anomalous dimension $\g_{[12]_n}$.
\be
&\bar y^{-\bar h_1-\bar h_3} \sum_{n,h} f_{12[12]_n}(h)^2 k_{2(\bar h_1 + \bar h_2 + n + \g_{[12]_n}/2)}(\bar z) k_{2h}^{h_{12},h_{34}}(1-z)\\
&= \bar y^{-\bar h_1-\bar h_3} \sum_{n,h} f_{12[12]_n}(h)^2 \frac 1 2\p{\frac{\g_{[12]_n}(h)}{2}}^2\nn\\
& \qquad \qquad
 \x\p{\frac 1 2\p{\pdr{}{\bar h_1}+ \pdr{}{\bar h_2}}}^2 k^{\bar h_{12},\bar h_{34}}_{2(\bar h_1 + \bar h_2 + n)}(\bar z) k_{2h}^{h_{12},h_{34}}(1-z) + \dots
\ee
This time, only the $\log \bar y^2$ terms are Casimir-singular. By similar logic as before, they are given by
\be
&\lim_{b,d \to 2,1} \frac 1 2 (\bar h_d + \bar h_b - \bar h_1 - \bar h_2)^2 f_{d1a}f_{a2b} f_{b3c}f_{c4d}\nn\\
&\qquad\x  \sum_{m,n} \g_{3214;badc}(m) \g_{\bar 1\bar 2\bar 3\bar 4; \bar a \bar b \bar c \bar d}(n)
\x y^{h_a+h_c-h_1-h_3+m} \bar y^{\bar h_b + \bar h_d - \bar h_1 - \bar h_3 + n}\log^2 \bar y,
\label{eq:prescriptiondoublepoles}
\ee
where $3,4=1,2$ or $3,4=2,1$, depending on which correlator we're studying. The Casimir-regular terms start at order $n=1$, by the same logic as before.

Another type of special case is when $h_a+h_c=h_1+h_4$ or $h_a+h_c=h_2+h_3$ or both. Naively, this leads to poles in the holomorphic part of the Casimir-singular terms.  However, the Casimir-singular and Casimir-regular parts combine to give a finite result in these cases. (We saw an example in section~\ref{sec:sigepszero}.) In the first case, we have a sum proportional to
\be
&\lim_{a\to -r} \G(-a-r) \sum_{h} S^{r,s}_{a}(h) k_{2h}^{r,s}(1-z)\nn\\
&= \lim_{a\to -r} \G(-a-r) \p{y^a + \frac{\pi}{\sin(\pi(s-r))}\frac{\G(-a)^2}{\G(-a-r)\G(-a-s)} \p{\frac{\G(1-r)^2 \cA_{a,r-1}(h_0)}{\G(1+s-r)}y^{-r} - (r\leftrightarrow s)}}\nn\\
&= - y^{-r} \log y + [\dots]_y,
\label{eq:correctcassing}
\ee
where
\be
r = h_{12},\quad s=h_{34},\quad a=h_a+h_c-h_1-h_3,
\ee
and $h_a+h_c\to h_1+h_4$.  The correct Casimir-singular term (\ref{eq:correctcassing}) can be obtained by multiplying the naive answer (containing the pole in $h_1+h_3-h_a-h_c$) by
\be
(h_a+h_c-h_1-h_4)\log y
\ee
and then taking the limit. This is exactly the same prescription we gave for curing $\bar h$ poles in (\ref{eq:curehbarpoles}), and the two special cases are related by crossing.

Finally, if $h_a+h_c=h_1+h_4=h_2+h_3$, then we have
\be
\lim_{a\to -r} \G(-a-r)^2 \sum_h S^{r,r}_a(h) k_{2h}^{r,r}(1-z) &= \frac 1 2 y^{-r} \log^2 y + [\dots]_y,
\ee
where now $[\dots]_y$ includes $y^{-r} \log y$ and $y^{-r}$ terms.  To get the correct Casimir-singular part, we multiply the naive answer by
\be
(h_a+h_c-h_1-h_4)^2 \frac 1 2 \log^2 y
\ee
and then take the limit. This exactly corresponds with the prescription for $\bar h$ double poles (\ref{eq:prescriptiondoublepoles}), and again these cases are related by crossing.

\clearpage
\bibliography{Biblio}{}
\bibliographystyle{utphys}

\end{document}